\documentclass[a4paper]{tlp}
\usepackage{amsfonts}
\usepackage[leqno]{amsmath}
\usepackage{rotating}
\usepackage[english]{babel}
\usepackage{url}
\usepackage{alltt}
\usepackage{hyperref}
\usepackage{verbatim}
\usepackage{amssymb}
\usepackage{stmaryrd}
\usepackage{proof}
\usepackage{graphicx}
\usepackage[svgnames]{xcolor}

\newtheorem{definition}{Definition} 
\newtheorem{proposition}{Proposition} 
\newtheorem{lemma}{Lemma} 
\newtheorem{theorem}{Theorem} 

\newcommand{\codesize}{\small}

\newcommand{\den}[1]{[\![#1]\!]}  

\newtheorem{example}{Example}[section]
\newtheorem{corollary}{Corollary}[section]
\newcommand{\nc}{\newcommand}
\nc{\trs}{{TRS}} 
\nc{\ctrs}{CS} 
\nc{\trss}{\trs's} 
\nc{\ctrss}{\ctrs's} 
\nc{\prog}{{\cal P}} 
\nc{\progh}{\hat{\prog}} 
\nc{\progm}{{\cal M}} 
\nc{\progmh}{\hat{\progm}} 
\nc{\progr}{{\cal R}} 
\nc{\rw}{\to} 
\nc{\tor}{\to} 
\nc{\crwlto}{\rightarrowtriangle}  
\nc{\clto}{\crwlto}                
\nc{\conscrwl}{\vdash_{\crwl}}                
\nc{\cltop}{\clto}                
\nc{\pcrwl}{\apcrwl}
\nc{\crwl}{{\it CRWL}}
\nc{\spcrwl}{\sapcrwl}
\nc{\sapcrwl}{{\mbox{${\crwl}^\sigma_{\pi^\alpha}$}}}
\nc{\sbpcrwl}{{\mbox{${\crwl}^\sigma_{\pi^\beta}$}}}
\nc{\apcrwl}{\mbox{{\it $\pi^\alpha$\!CRWL}}} 
\nc{\bpcrwl}{\mbox{{\it $\pi^\beta$\!CRWL}}} 
\nc{\vds}{\vdash_{\crwl}} 
\nc{\vdap}{\vdash_{\apcrwl}} 
\nc{\vdbp}{\vdash_{\bpcrwl}} 
\nc{\vdsap}{\vdash_{\sapcrwl}} 
\nc{\vdsbp}{\vdash_{\sbpcrwl}} 
\nc{\apor}{\textbf{POR}$^\alpha$} 
\nc{\bpor}{\textbf{POR}$^\beta$} 
\nc{\sapor}{\textbf{OR}$^\sigma_{\pi^\alpha}$} 
\nc{\sbpor}{\textbf{OR}$^\sigma_{\pi^\beta}$} 
\nc{\pst}[1]{\mi{pST}(#1)}
\nc{\groot}{{R_p}}
\nc{\var}{{\cal V}}
\nc{\plrity}[1]{\plrityZ({#1})}
\nc{\plrityZ}{\mi{plurality}}
\nc{\cjto}{{\cal A}} 
\nc{\cjtoD}{{\cal S}} 
\nc{\vran}{\mi{vran}}
\nc{\icsus}{\CSubst_\perp^?}
\nc{\crule}[1]{\textbf{\textsc{#1}}} 
\nc{\rrw}{{\,\rw^{rt}\,}} 
\nc{\con}{{\cal C}}
\nc{\dens}[1]{{\den{#1}^{\mi{sg}}}} 
\nc{\densp}[2]{{\den{#1}^{\mi{sg}}_{#2}}} 
\nc{\denr}[1]{{\den{#1}^{\mi{rt}}}} 
\nc{\denrp}[2]{{\den{#1}^{\mi{rt}}_{#2}}} 
\nc{\denap}[1]{{\den{#1}^{\alpha\! \mi{pl}}}}
\nc{\denapp}[2]{{\den{#1}^{\alpha\! pl}_{#2}}} 
\nc{\denbp}[1]{{\den{#1}^{\beta\! \mi{pl}}}}
\nc{\denbpp}[2]{{\den{#1}^{\beta\! \mi{pl}}_{#2}}} 
\nc{\densap}[1]{{\den{#1}^{s\alpha\! p}}}
\nc{\densapp}[2]{{\den{#1}^{s\alpha\! p}_{#2}}}
\nc{\densbp}[1]{{\den{#1}^{s\beta\! p}}}
\nc{\densbpp}[2]{{\den{#1}^{s\beta\! p}_{#2}}}
\nc{\cc}[1]{cc(#1)} 
\nc{\empsec}{[]} 
\nc{\vextra}[1]{\mi{vExtra}(#1)} 
\nc{\vE}{{{\cal V}_e}} 
\nc{\thetae}{\theta_e} 
\nc{\ordap}{\sqsubseteq} 
\nc{\pordap}{\ordap_{\pi}} 
\nc{\dsord}{\unlhd} 
\nc{\bdsord}{\dsord^{\beta \mi{pl}}} 
\nc{\adsord}{\dsord^{\alpha \mi{pl}}} 
\nc{\ccon}{c\con} 
\nc{\scon}{s\con} 
\nc{\refp}[1]{\ref{#1} (page~\pageref{#1})}
\nc{\bunion}{\sqcup}
\nc{\bbunion}{\bigsqcup}
\nc{\bubsus}{\rightarrow_\bunion}
\nc{\pps}{\vdap}
\nc{\denp}{\denap}
\nc{\pA}[1]{a\prog({#1})} 
\nc{\sA}[1]{sa\prog({#1})} 
\nc{\size}[1]{\mi{size}({#1})} 
\nc{\ordLex}{\lessdot}
\nc{\prem}[1]{\Pi^{{#1}}} 
\nc{\pru}{\Delta} 
\nc{\uprem}{\pi} 
\nc{\pruD}[1]{\pru_{{#1}}}
\nc{\pruDp}[1]{\pru'_{{#1}}}

\nc{\mt}{mt} 
\nc{\mdtc}{mdtc} 
\nc{\tree}{{\cal T}} 
\nc{\pat}{\pi}
\nc{\ordSus}{\precsim}

\nc{\mi}[1]{\mathit{#1}}
\nc{\unravel}{\mi{unravel}}
\nc{\leftFR}{\mi{leftFR}}
\nc{\CTerm}{\mi{CTerm}}
\nc{\CSubst}{\mi{CSubst}}
\nc{\Cntxt}{\mi{Cntxt}}
\nc{\Subst}{\mi{Subst}}
\nc{\FS}{\mi{FS}}
\nc{\CS}{\mi{CS}}
\nc{\Exp}{\mi{Exp}}
\nc{\PVar}{\mi{PVar}}
\nc{\FVar}{\mi{FVar}}
\nc{\sM}{\mi{sM}}
\nc{\sg}{\mi{sg}}
\nc{\match}{\mi{match}}
\nc{\project}{\mi{project}}
\nc{\true}{\mi{true}}
\nc{\fng}{\mathit{find2NG}}
\nc{\pepe}{\mathit{pepe}}
\nc{\paco}{\mathit{paco}}
\nc{\jaime}{\mathit{jaime}}
\nc{\maria}{\mathit{maria}}
\nc{\employees}{\mi{employees}}
\nc{\branches}{\mi{branches}}
\nc{\mad}{\mi{madrid}}
\nc{\vigo}{\mi{vigo}}
\nc{\man}{\mi{man}}
\nc{\woman}{\mi{woman}}
\nc{\boss}{\mi{boss}}
\nc{\clerk}{\mi{clerk}}
\nc{\twoclerks}{\mi{twoclerks}}
\nc{\find}{\mi{find}}
\nc{\thn}{\mi{then}}
\nc{\tru}{\mi{true}}
\nc{\prj}{\mi{project}}
\nc{\mivar}{\mi{var}}
\nc{\dom}{\mi{dom}}
\nc{\pair}{\mi{pair}}
\nc{\leaf}{\mi{leaf}}
\nc{\branch}{\mi{branch}}
\nc{\hypothesis}{hypothesis}
\nc{\pos}{{\cal O}} 
\nc{\posV}{\pos_{\var}} 
\nc{\posH}[2]{{\pos({#1}, {#2})}} 
\nc{\apPos}[2]{{#1}|_{#2}}
\nc{\apPosD}[2]{{#1}\upharpoonright_{#2}}

\nc{\classAB}{{\cal C}^{\alpha \beta}} 
\nc{\proofs}{\cite{arjrhTPLPproofs}} 
\nc{\refap}[1]{\proofs}
\nc{\podaSecImp}[1]{} 

\newcommand{\longFlag}{ON} %
\nc{\longText}[2]{\ifthenelse{\equal{\longFlag}{ON}}{{#1}}{{#2}}}
\nc{\todo}[1]{\textcolor{red}{#1}}
\nc{\regla}[4]{
\textbf{\crule{#1}}
 $\ $ $\begin{array}{c}
{#3}\\ 
\hline 
{#2}
\end{array}$
\quad {#4}
}

\graphicspath{{images/}}

\title{Singular and Plural Functions for Functional Logic Programming%
}

\author[A. Riesco and J. Rodr{\'\i}guez-Hortal{\'a}]{Adri\'an Riesco%
\thanks{Research partially supported by 
the Spanish projects \emph{DESAFIOS10} (TIN2009-14599-C03-01) and \emph{PROMETIDOS-CM} (S2009/TIC-1465).}  
and Juan Rodr{\'\i}guez-Hortal{\'a}\thanks{Research partially supported by the Spanish projects
  \emph{FAST-STAMP} (TIN2008-06622-C03-01/TIN), \emph{PROMETIDOS-CM} (S2009TIC-1465) and \emph{GPD-UCM} (UCM-BSCH-GR58/08-910502).}\\
Facultad de Inform\'atica, Universidad Complutense de Madrid\\ 
Dpto. Sistemas Inform\'aticos y Computaci\'on\\
\email{\{ariesco, juanrh\}@fdi.ucm.es}}

\date{}

\submitted{11 February, 2011}
\revised{17 November 2011}
\accepted{6 March 2012}

\begin{document}
\maketitle

\begin{abstract}
Modern functional logic programming (FLP) languages use non-terminating and non-confluent 
constructor systems (\ctrss) 
as programs 
in order to define non-strict non-determi-nistic functions. 
Two semantic alternatives have been usually considered for parameter passing with this kind of functions: call-time choice and run-time choice. While the former is the standard choice of modern FLP languages, the latter lacks some basic properties---mainly compositionality---that have prevented its use in practical FLP systems. Traditionally it has been considered that call-time choice induces a 
singular denotational semantics, while run-time choice induces a plural semantics. We have discovered that this latter identification 
is wrong when pattern matching is involved, 
and thus in this paper we propose two novel compositional plural semantics for \ctrss\ that are different from run-time choice. 

We investigate the basic properties of our plural semantics---compositionality, polarity, monotonicity for substitutions, and a restricted form of the bubbling property for constructor systems---and the relation between them and to previous proposals, concluding that these semantics form a hierarchy in the sense of set inclusion of the set of values computed by them. Besides, we have identified a class of programs characterized by a simple syntactic criterion for which the proposed plural semantics behave the same, and a program transformation that can be used to simulate one of the proposed plural semantics by term rewriting. At the practical level, we study how to use the new expressive capabilities of these semantics for improving the declarative flavour of programs. As call-time choice is the standard semantics for FLP, it still remains the best option for many common programming patterns. Therefore we propose a language which combines call-time choice and our plural semantics, that we have implemented in the Maude system.
The resulting interpreter is then employed to develop and test several significant examples showing the capabilities of the combined semantics. 

\end{abstract}



\section{Introduction}\label{sec:intro}

%
The combination of functional and logic features has been addressed in several proposals for multi-paradigm programming languages \cite{CiaoSurvery2011,OzBook2004,Henderson96themercury,Frolog08}, with different variants---lazy or eager evaluation of functions, concurrent capabilities, support for object-oriented programming\ldots\ In this work we focus on the integration into a single language of the main features of lazy functional programming (FP) and logic programming (LP) 
that is described in \cite{AntoyHanusComACM10}. 
Term rewriting \cite{baader-nipkow} and term graph rewriting systems \cite{Plump98} have often been used for modeling the semantics and operational behaviour of that approach to functional logic programming (FLP) \cite{deGroot86,Hanus07ICLP}. 
In particular, the class of left-linear constructor-based term rewriting systems---or simply constructor systems (\ctrss)---, in which the signature is divided into two disjoint sets of constructor and function symbols, is 
used frequently to represent programs. There the notion of value as a term built using only constructor symbols---called constructor term or just c-term---arises naturally, and this way a term rewriting derivation from an expression to a c-term 
represents the 
reduction of 
that expression to one of its values in the language being modelled. 
This corresponds to a value-based semantic view, in which the purpose of computations is to produce values made of constructors.
Besides, term graphs are 
used for modelling subexpression sharing, where several occurrences of the same subexpression are represented by several pointers to a single node in a term graph, resulting in a potential improvement of the time and space performance of programs. Sharing is at the core of implementations of lazy 
FP and FLP languages, and so several variations of term graph rewriting have also been used in formulations of the semantics of call-by-need in FP \cite{Lau93,AriolaFMOW95,PlasmeijerE93} and FLP \cite{EchahedJanodet98JICSLP,AHHOV05,lrs07,LMRS10LetRwTPLP}. 

On the other hand, non-determinism is an expressive feature that has been used for a long time in programming \cite{Dijkstra97,Hughes89,mccarthy63basis} and system specification \cite{maude-book,cafe-report,BorovanskyKirchnerKirchnerMoreauRingeissenWRLA98}. In both fields, one of the appeals of term rewriting is its elegant way to express non-determinism through the use of 
non-confluent term rewriting systems, obtaining a clean and high level representation of complex systems and programs. Non-determinism is 
integrated in FLP languages by means of a backtracking mechanism in the style of Prolog \cite{SterlingShapiro86}. It is 
introduced by employing 
possibly non-terminating and non-confluent \ctrss\ 
as programs, thus expressing non-strict non-deterministic functions, which are one of the most distinctive features of the paradigm \cite{GHLR99,DBLP:conf/flops/AntoyH02,AntoyHanusComACM10}.\footnote{Non-determinism also appears in FLP as a result of the utilization of narrowing as the fundamental operational mechanism \cite{Han05TR} but, as usual in many works in the field, we will focus on rewriting aspects only, so our conclusions could be lifted to the narrowing case in subsequent works.}

The point is that this combination of non-strictness and non-determinism gives rise to several semantic alternatives \cite{Sondergaard95,hussmann93}. In particular in \cite{Sondergaard95} the different language variants that result after adding non-determinism to a basic functional language were expounded, structuring the comparison as a choice among different options over several dimensions: strict/non-strict functions, angelic/demonic/erratic non-determi\-nis\-tic choices, and \emph{singular}/\emph{plural} \emph{semantics} for parameter passing, also called \emph{call-time choice}/\emph{run-time choice} in \cite{hussmann93}. In the present paper we assume non-strict angelic non-determinism, 
so we focus on the last dimension only. To do that, let us take a look at the following example. 

\begin{example}\label{EIntro1}
Consider the program $\{f(c(X)) \tor d(X,X), X~?~Y \tor X, X~?~Y \tor Y\}$ and the expression $f(c(0~?~1))$.  From an operational perspective we have to decide when it is time to fix 
the values for the arguments of functions:
\begin{itemize}
	\item Under a \emph{call-time choice} semantics a value for each argument will be fixed on parameter passing and shared between every copy of that argument which arises during the computation. This corresponds to call-by-value in a strict setting and to call-by-need in a non-strict setting, in which a partial value instead of a total value is computed. 
So when applying the rule for $f$ the two occurrences of $X$ in $d(X, X)$ will share the same value, 
hence $d(0,0)$ and $d(1,1)$ are correct values for $f(c(0~?~1))$ in this semantics, while it is not the case either for $d(0,1)$ or for $d(1,0)$. 
    \item On the other hand \emph{run-time choice} corresponds to call-by-name,
so the values of the arguments are fixed as they are used---i.e., as their evaluation is demanded by the matching process---and the copies of each argument created by parameter passing may evolve independently afterwards. Under this semantics not only $d(0,0)$ and $d(1,1)$ but also $d(0,1)$ and $d(1,0)$ are correct values for $f(c(0~?~1))$. 
\end{itemize}
\end{example}
In general, a call-time choice semantics produces less results than run-time choice. Modern functional-logic languages like Toy \cite{LS99} or Curry \cite{Han05TR} are heavily influenced by lazy functional programming and so they implement sharing in their operational mechanism, which results in call-by-need evaluation and the adoption of call-time choice.
On the other hand, {term rewriting} is considered a standard formulation for run-time choice,\footnote{In fact angelic non-strict run-time choice.} and is the basis for the semantics of languages like Maude \cite{maude-book}.

But we may also see things from another perspective. 
\begin{example}\label{EIntro2}
Consider again the program in Example \ref{EIntro1}. From a denotational perspective we have to think about the domain used to instantiate the variables of the program rules:
\begin{itemize}
	\item Under a \emph{singular semantics} variables will be instantiated with single values (which may be partial in a non-strict setting). \emph{This is equivalent to having call-time choice parameter passing}. 

    \item The alternative is having a \emph{plural semantics}, in which the variables are instantiated with sets of values. Traditionally it has been considered that run-time choice has its denotational counterpart on a plural semantics, but we will see that this identification is wrong. Consider the expression $f(c(0)~?~c(1))$, 
under run-time choice, that is, term rewriting, the evaluation of the subexpression $c(0)~?~c(1)$ is needed in order to get an instance of the left-hand side of the rule for $f$. Hence a choice between $c(0)$ and $c(1)$ is performed and so neither $d(0,1)$ nor $d(1,0)$ are correct values for $f(c(0)~?~c(1))$. Nevertheless, under a plural semantics we may consider the set $\{c(0), c(1)\}$ which is a subset of the set of values for $c(0)~?~c(1)$ in which every element matches the argument pattern $c(X)$. Therefore, the set $\{0, 1\}$ can be used for parameter passing obtaining a kind of ``set expression'' $d(\{0,1\}, \{0,1\})$ that yields the values $d(0,0)$, $d(1,1)$, $d(0,1)$, and $d(1,0)$. 
\end{itemize}
The conclusion is clear: \emph{the traditional identification of run-time choice with a plural semantics is wrong when pattern matching is involved}. 
\end{example}
Which of these is the more suitable perspective for FLP? This problem did not appear in \cite{Sondergaard95} because no pattern matching was present, nor in \cite{hussmann93} because only call-time choice was adopted there. This fact was pointed out for the first time 
in \cite{rodH08}, where the \pcrwl\ logic---named  $\pi \crwl$ in that work---was proposed as a novel formulation of a plural semantics with pattern matching. This proves that one can conceive a meaningful plural semantics that is different to run-time choice, i.e., run-time choice is not the only plural semantics we should consider.  We have seen that, using the program above, the expression $f(c(0~?~1))$ has more values than the expression $f(c(0)~?~c(1))$ under run-time choice 
although they only differ in 
the subexpressions $c(0~?~1)$ and $c(0)~?~c(1)$, which have the same values under all three call-time choice, run-time choice, and plural semantics. That violates a fundamental property of FLP languages stating that any expression can be replaced by any other expression which could be reduced to exactly the same set of values. We will see that our plural semantics shares with \crwl\footnote{\textbf{C}onstructor-based \textbf{R}e\textbf{W}riting \textbf{L}ogic.} \ \cite{GHLR99} (the standard logic for call-time choice\footnote{In fact angelic non-strict call-time choice.}) a compositionality property for values that makes it more suitable than run-time choice for a value-based language like current implementations of FLP. Nevertheless run-time choice can be a good option for other kind of rewriting-based languages like Maude, in which the notion of value is not necessarily present, at least in the sense it is in FLP languages.

~\\
In this paper we have put together our previous results about plural semantics, integrating our presentation of \apcrwl\ from \cite{rodH08} with
a user level introduction to a Maude-based transformational prototype for \apcrwl\ \cite{RRLDTA09}. 
\podaSecImp{ our description of a Maude-based transformational prototype for \apcrwl\ \cite{RRLDTA09}.}
We have also included the results obtained in \cite{rrpepm10}, which is devoted to the exploration of the new expressive capabilities of our plural semantics. Although our plural semantics allows an elegant encoding of some problems---in particular those with an implicit manipulation of sets of values---, call-time choice still remains the best option for many common programming patterns \cite{GHLR99,DBLP:conf/flops/AntoyH02}. Therefore we propose a combined semantics for a language in which the user can specify, for each function symbol, which arguments are considered ``plural arguments''---thus being evaluated under our plural semantics---and which ``singular arguments''---thus being evaluated under call-time choice. This semantics is precisely specified by a modification of the \crwl\ logic, which retains the important properties of \crwl\ and \apcrwl, like compositionality. These new features were implemented by extending our Maude prototype, and then used to develop and test several significant examples showing the expressive capabilities of the combined semantics

Apart from giving a unified and revised presentation, we have made several relevant advances.\podaSecImp{  both in theoretical and practical aspects. In the theoretical side we}
We have extended most of our results to deal with programs with extra variables, and above all, we have introduced the new plural semantics \bpcrwl\ inspired by the proposal from \cite{BB09}. The properties of this semantics and its relation to call-time choice, run-time choice, and \apcrwl\ have been studied in depth and with technical accuracy. \podaSecImp{ At the practical level, we provide much more detailed explanations about our Maude implementation and how we adapt the natural rewriting on-demand strategy \cite{DBLP:conf/ppdp/Escobar03} to deal with the representation of term graphs internally used by our interpreter.}
Our current implementation does not deal with extra variables because they cause an explosion in the search space  when evaluated by term rewriting---we consider the development of a suitable plural narrowing mechanism that could effectively handle extra variables a possible subject of future work.  
\podaSecImp{ Finally, our implementation of natural rewriting for Maude system modules has been carefully crafted as an independent library, in order to improve its reusability. }

~\\
The rest of the paper is organized as follows. Section \ref{sec:preliminaries} contains some technical preliminaries and notations about term rewriting systems and the \crwl\ logic. 
In Section \ref{sec:pl_sem} we introduce \apcrwl\ and \bpcrwl, two variations of \crwl\ to express plural semantics, and present some of their properties, in particular compositionality. In Section \ref{sec:hierarchy} we study the relation between call-time choice, run-time choice, and our plural semantics, focusing on the set of values computed by each semantics and concluding that these four semantics form a hierarchy in the sense of set inclusion. We also present a class of programs characterized by a simple syntactic criterion under which our two plural semantics are equivalent, and conclude the section providing a simple program transformation that can be used to simulate \apcrwl\ with term rewriting. Section \ref{sec:examples} begins with a presentation of our combinations of call-time choice and plural semantics that are formalized through the \sapcrwl\ and \sbpcrwl\ logics, which correspond to the combination of call-time choice with \apcrwl\ and \bpcrwl, respectively. Then follows a user level introduction to our Maude prototype, which implements the \sapcrwl\ logic, as it is based on the transformation from Section \ref{sec:hierarchy}. 
The prototype is then employed to illustrate the use of the combined semantics for improving the declarative flavour of programs.\podaSecImp{ In Section \ref{sec:implemen} we discuss the implementation of our interpreter in detail: the core language that is used to represent the term graphs used to achieve call-time choice, the program transformations employed to get the semantic specified for each part of the program, our adaptation of the natural rewriting strategy to deal with the core language, and how we take advantage of the metaprogramming capabilities of Maude to perform all these tasks.} 
This section concludes with a short sketch of the implementation of our prototype. 
Finally, in Section~\ref{sec:conclusions} 
we outline some possible lines of future work. 
For the sake of readability, some of the proofs have been moved to 
\proofs, although the intutitions behind our main results have been presented in the text. 

\section{Preliminaries}\label{sec:preliminaries}

We present in this section the main notions needed throughout the rest of the paper:
Section \ref{sec:cterms_intro} introduces constructor-based systems, while Section
\ref{sec:crwl_intro} describes the \crwl\ framework.

\subsection{Constructor systems}\label{sec:cterms_intro}


We consider a first order signature $\Sigma = \CS \uplus \FS$, where $\CS$
and $\FS$ are two disjoint sets of \emph{constructor}
and defined \emph{function} symbols respectively, all of them with associated arity. We
write $\CS^n$ ($\FS^n$ resp.) for the set of constructor (function)
symbols of arity $n \in \mathbb{N}$.
We write $c,d,\ldots$ for constructors, $f,g,\ldots$ for functions, and
$X,Y,\ldots$ for variables of a numerable set $\var$.
The notation $\overline{o}$ stands for tuples of any kind of syntactic objects.
Given a set $\cjto$ we denote by $\cjto^*$ the set of finite sequences of
elements of that set. We denote the empty sequence by $\empsec$. For any sequence
$a_1 \ldots a_n \in \cjto^*$ and function
$f : \cjto \rightarrow \{\mi{true}, \mi{false}\}$, we denote by
$a_1 \ldots a_n~|~f$ the sequence constructed by taking in order every
element from $a_1 \ldots a_n$ for which $f$ holds. 
Finally, for any $1 \leq i \leq n$, $(a_1 \ldots a_n)[i]$ denotes $a_i$. 

The set $\mi{Exp}$ of {\it expressions} is defined as $\mi{Exp}\ \ni e::= X
\mid h(e_1,\ldots,e_n)$, where $X\in\var$, $h\in \CS^n\cup \FS^n$ and
$e_1,\ldots,e_n\in \mi{Exp}$. 
We use the symbol $\equiv$ for the syntactic equality between expressions, and in general for any syntactic construction. 
The set $\mi{CTerm}$ of {\it constructed terms}
(or {\it c-terms}) is defined like $\mi{Exp}$, but with $h$ restricted to
$\CS^n$ (so $\mi{CTerm} \subseteq \mi{Exp}$). The intended
meaning is that $\mi{Exp}$ stands for evaluable expressions, i.e., expressions
that can contain function symbols, while $\mi{CTerm}$ stands for data
terms representing \textbf{values}. We
will write $e,e',\ldots$ for expressions and $t,s,\ldots$ for c-terms. The set
of variables occurring in an expression $e$ will be denoted as $\mi{var}(e)$.
We will frequently use \emph{one-hole contexts}, defined as
$\mi{Cntxt}\ni {\cal C} ::= [\ ] \mid h(e_1,\ldots,{\cal C},\ldots,e_n)$,
with $h\in \CS^n\cup \FS^n$, $e_1, \ldots, e_n \in \mi{Exp}$.
The application of a context ${\cal C}$ to an expression $e$, written by
${\cal C}[e]$, is defined inductively as $[\ ][e] = e$ and $h(e_1,\ldots,{\cal
  C},\ldots,e_n)[e] = h(e_1,\ldots,{\cal C}[e],\ldots,e_n)$. 

A position of an expression is a chain of natural numbers separated by dots that determines one of its subexpressions. Given an expression $e$ by $\pos(e)$ we denote the set of positions in $e$, which is defined as $\pos(X) = \epsilon$; $\pos(h(e_1, \ldots, e_n)) = \{\epsilon\}~\cup~\{i.o~|~i \in \{1, \ldots, n\} ~\wedge~ o \in \pos(e_i)\}$, where $X \in \var$, $h \in \Sigma$, and $\epsilon$ denotes the empty or top position. We will write $o, p, q, u, v, \ldots$ for positions. By $e|_o$ we denote the subexpression of $e$ at position $o \in \pos(e)$, defined as $e|_\epsilon = e$; $h(e_1, \ldots, e_n)|_{i.o} = e_i|_o$. The set of variable positions in $e$ is denoted as $\posV(e)$ and defined as $\posV(e) = \{o \in \pos(e)~|~e|_o \in \var\}$. 

{\em Substitutions} $\theta \in \mi{Subst}$ are finite mappings $\theta:\var
\longrightarrow \mi{Exp}$, extending naturally to $\theta:\mi{Exp} \longrightarrow
\mi{Exp}$. We write $\epsilon$ for the identity (or empty) substitution.
We write $e\theta$ for the application of $\theta$ to $e$, and $\theta\theta'$
for the composition, defined by $e(\theta\theta') = (e\theta)\theta'$. The domain and
variable range of $\theta$ are defined as $\mi{dom}(\theta) = \{X\in \var \mid X\theta \neq
X\}$ and $\vran(\theta) = \bigcup_{X\in \mi{dom}(\theta)}\mi{var}(X\theta)$.
If $\mi{dom}(\theta_0) \cap \mi{dom}(\theta_1) = \emptyset$, their disjoint
union $\theta_0 \uplus \theta_1$ is defined by $(\theta_0 \uplus \theta_1)(X) = \theta_i(X)$, if $X \in \mi{dom}(\theta_i)$ for some 
$i \in \{0, 1\}$; $(\theta_0 \uplus \theta_1)(X) = X$ otherwise.
Given $W \subseteq {\cal V}$ we write
$\theta|_{W}$ for the restriction of $\theta$
to $W$, and $\theta|_{\backslash D}$ is a shortcut for
$\theta|_{({\cal V} \backslash D)}$. We will
sometimes write $\theta = \sigma[W]$ instead of  $\theta|_W =\sigma|_W$.
\emph{C-substitutions}
$\theta \in \CSubst$ verify that $X\theta \in \CTerm$ for all
$X\in \mi{dom}(\theta)$.
We say that \emph{e subsumes} $e'$, and write $e \ordSus e'$,
if $e\sigma \equiv e'$ for some substitution $\sigma$. 

A \emph{constructor-based term rewriting system} (\emph{\ctrs}) or just \emph{constructor system} or \emph{program} ${\cal P}$  
is a set of \emph{rewrite rules} or \emph{program rules} of the form
$f(t_1, \ldots, t_n)\to r$ where $f\in \FS^n$, $e\in \mi{Exp}$, and
$(t_1, \ldots, t_n)$ is a linear tuple of
c-terms, where linearity means that variables occur only once in $(t_1, \ldots, t_n)$. 
Notice that we allow $r$ to contain {\it extra variables}, i.e., variables not occurring in $(t_1, \ldots, t_n)$. 
To be precise, we say that $X \in \var$ is an extra variable in the rule $l \tor r$
iff $X \in \mi{var}(r) \setminus \mi{var}(l)$, and by $\vextra{R}$ we denote
the set of extra variables in a program rule $R$. 
For any program $\prog$ the set $\FS^\prog$ of functions defined by $\prog$ is $\FS^\prog = \{f \in \FS~|~\exists (f(\overline{p}) \tor r) \in \prog\}$. 
We assume that every 
program  $\prog$ contains the rules $\{X~?~Y \tor X, X~?~Y \tor Y, \mi{if}~\mi{true}~\mi{then}~X \tor X\}$, defining the behaviour of the infix function $? \in \FS^2$ and the mixfix function $\mi{if}\,\mi{then}\, \in \FS^2$ (used as $\mi{if} \,e_1\, \mi{then}\, e_2$), and that those are the only rules for that function symbols.
Besides $?$ is right-associative, so $e_1~?~e_2~?~e_3 \equiv e_1~?~(e_2~?~e_3)$. For the sake of conciseness we will often omit these rules when presenting a program. 

Given a 
program $\prog$, its associated \emph{term rewriting relation} $\rw_{\cal P}$ is defined as:
${\cal C}[l\sigma] \rw_{\cal P} {\cal C}[r\sigma]$
for any context ${\cal C}$, rule $l \tor r \in {\cal P}$ and $\sigma \in
\emph{Subst}$. 
We write $\stackrel{*}{\rw_{\cal P}}$ for the reflexive and
transitive closure of the relation $\rw_{\cal P}$. In the following, we will usually omit the reference to ${\cal  P}$ or denote it by $\prog \vdash e \rw e'$ and $\prog \vdash e \rw^* e'$. 

\subsection{The CRWL framework}\label{sec:crwl_intro}


The \crwl\ framework \cite{GHLR96,GHLR99} is considered a standard formulation of call-time choice by the FLP community \cite{Hanus07ICLP,AntoyHanusComACM10}. To deal with non-strictness at the semantic level, 
$\Sigma$ is enlarged 
with a new constant constructor symbol $\perp$. The sets $\mi{Exp}_\perp$,
$\mi{CTerm}_\perp$, $\mi{Subst}_\perp$, $\mi{CSubst}_\perp$
of \emph{partial expressions}, etc., are defined naturally. Our contexts will contain partial expressions from now on unless explicitly specified. Expressions, substitutions,
etc. not containing $\perp$ are called \emph{total}. Programs in \crwl\ still consist of rewrite rules with total expressions in both sides, so $\perp$ does not appear in programs. 
Partial expressions are ordered by the {\em approximation} ordering $\sqsubseteq$ defined as the least
partial ordering satisfying $\perp \sqsubseteq e$ and $e \sqsubseteq e'
\Rightarrow {\cal C}[e] \sqsubseteq {\cal C}[e']$
for all $e,e' \in \mi{Exp}_\perp, {\cal C} \in {\it \mi{Cntxt}}$. This partial ordering
can be extended to substitutions: given $\theta,\sigma\in \mi{Subst}_\bot$ we say
$\theta\sqsubseteq\sigma$ if $X\theta\sqsubseteq X\sigma$ for all $X\in\var$.

The semantics of a program ${\cal P}$ is determined in {\it CRWL} by means of a
proof calculus able to derive \emph{reduction statements} of the form $e \rightarrowtriangle t$, with $e \in \mi{Exp}_\perp$ and $t \in \mi{CTerm}_\perp$,
meaning informally that $t$ is (or approximates to) a \textit{possible value} of $e$, obtained by
iterated reduction of $e$ using ${\cal P}$ under call-time choice. 
The {\it CRWL}-proof calculus is presented in Figure \ref{fig:crwl}.
\begin{figure}
\begin{center}
 \framebox{
\begin{minipage}{.95\textwidth}
\begin{center}
\begin{small}
\begin{tabular}{l@{~~~~~}l}
\regla{RR}{X \clto X}{}{$X \in \var$} \hspace{2ex}
\regla{DC}{c(e_1, \ldots, e_n) \clto c(t_1, \ldots, t_n)}{e_1 \clto t_1 \ldots e_n \clto t_n}{$c \in {CS}^n$} \\
\regla{B}{e \clto \perp}{}{} \hspace{1.4ex}
\regla{OR}{f(e_1, \ldots, e_n) \clto t}{e_1 \clto p_1\theta \ldots e_n \clto p_n\theta ~~r\theta \clto t}{$
      \begin{array}{l}
        f(p_1, \ldots, p_n)\tor r \in \prog\\
        \theta\in CSubst_\perp
      \end{array}$} \\
\end{tabular}
\end{small}
\end{center}
\end{minipage}
}
\end{center}
    \caption{Rules of \crwl}
    \label{fig:crwl}
\end{figure}
Rules \textbf{RR} (restricted reflexivity) and \textbf{DC} (decomposition) are used to reduce any variable to itself, and to decompose the evaluation of constructor-rooted expressions.
Rule \textbf{B} (bottom) allows us to avoid the evaluation of expressions, in order to get a non-strict semantics.
Finally rule \textbf{OR} (outer reduction) expresses that to evaluate a function call we must first evaluate its arguments to get an instance of a program rule, perform parameter passing (by means of some substitution $\theta \in \mi{CSubst}_\perp$),
and then reduce the correspondingly instantiated right-hand side. The use of partial c-substitutions in \textbf{OR} is essential
to express call-time choice, as only single partial values are used for parameter passing. Notice also that by the effect of $\theta$ in \textbf{OR},
extra variables in the right-hand side of a rule can
be replaced by any partial c-term, but not by any expression as in term rewriting. 

We write ${\cal P} \conscrwl e \crwlto t$ to express that $e\crwlto t$ is derivable in the {\it CRWL}-calculus using
the program ${\cal P}$. Given a program
$\mathcal{P}$, the \emph{CRWL-denotation} of an expression $e \in \mi{Exp}_\perp$ is
defined as $\den{e}^{\sg}_{\mathcal{P}}=\{t\in \mi{CTerm}_\perp \mid \mathcal{P}
\conscrwl e\crwlto t\}$. In the following, we will usually omit the reference to ${\cal  P}$ when implied by the context.

\section{Two plural semantics for constructor systems}\label{sec:pl_sem}

In this section we present two semantic proposals for constructor systems that are plural in the sense described in the introduction, but at the same time are different to the run-time choice semantics induced by term rewriting. We will formalize them by means of two
modifications of the \crwl\ proof calculus, that will now consider sets of partial values for parameter passing instead of single partial values. Thus 
only the rule $\crule{OR}$ should be modified. To avoid the need to extend the syntax with new constructions to represent those ``set expressions'' that we mentioned in the introduction, we will exploit the fact that $\den{e_1~?~e_2} = \den{e_1} \cup \den{e_2}$ for any sensible semantics---in particular each of the semantics considered in this work. Therefore the substitutions used for parameter passing will map variables to ``disjunctions of values.'' 
We define the set $\icsus = \{\theta \in \Subst_\perp~|~\forall X \in \mi{dom}(\theta), \theta(X) = t_1~?~\ldots~?$ $t_n$ such that $t_1, \ldots, t_n \in \CTerm_\perp, n>0\}$, for which $\CSubst_\perp \subseteq \icsus \subseteq \Subst_\perp$ obviously holds. The operator $? : \CSubst_\perp^* \rightarrow \icsus$ constructs the $\icsus$ corresponding to a non-empty sequence of $\CSubst_\perp$
, and it is defined as follows: 
$$
 ?(\theta_1 \ldots \theta_n)(X) = \left\{
 \begin{array}{ll}
X~?~\rho_1(X) ~?~ \ldots ~?~ \rho_m(X) & \mbox{  if } \exists \theta_i \mbox{ such that }  X \not\in \mi{dom}(\theta_i) \\
 \theta_1(X) ~?~ \ldots ~?~ \theta_n(X) & \mbox{  otherwise }
 \end{array}\right.
$$
where $\rho_1 \ldots \rho_m = \theta_1 \ldots \theta_n~|~\lambda \theta.(X \in \mi{dom}(\theta))$. 
This operator is overloaded to handle non-empty sets $\Theta \subseteq \CSubst_\perp$ as $?\Theta = ?(\theta_1 \ldots \theta_n)$ where the sequence $\theta_1 \ldots \theta_n$ corresponds to an arbitrary reordering of the elements of $\Theta$---for example using some standard order of terms in the line of \cite{SterlingShapiro86}.

\begin{lemma}\label{lemIntCjtoCsus}
For any $\theta_1, \ldots, \theta_n \in \CSubst_\perp$, $\mi{dom}(?\{\theta_1 \ldots \theta_n\}) = \bigcup_i \mi{dom}(\theta_i)$. 
\end{lemma}
\begin{proof}
Simple calculations using the definition of $?\{\theta_1 \ldots \theta_n\}$, see \refap{PROOFlemIntCjtoCsus} for details. 
\end{proof}

\subsection{\apcrwl}\label{sec:aplural}
Our first semantic proposal is defined by the \apcrwl-proof calculus in Figure \ref{fig:apcrwl}. The only difference with the \crwl\ proof calculus in Figure \ref{fig:crwl} is that the rule $\crule{OR}$ has been replaced by \textbf{\apor} (alpha plural outer reduction), in which we may compute more than one partial value for each argument, and then use a substitution from $\icsus$ instead of $\CSubst_\perp$ for parameter passing, achieving a plural semantics.\footnote{In fact angelic non-strict plural non-determinism.} Besides, extra variables are instantiated by an arbitrary $\thetae \in \icsus$ for the same reason. Just like \crwl, the calculus evaluates expressions in an innermost way, and avoids the use of any transitivity rule that would induce a step-wise semantics like e.g. term rewriting. The motivation for that is to get a compositional calculus in the values it computes, i.e., that the semantics of an expression would only depend on the semantics of its constituents, in a simple way---we will give a formal characterization for that in Theorem \ref{lCompPS} below. Note that the use of partial c-terms as values is crucial to prevent innermost evaluation from making functions strict, thus losing lazy evaluation. Fortunately the rule $\crule{B}$ combined with the use of partial
substitutions for parameter passing ensure a lazy behaviour for both \apcrwl\ and \crwl.
Therefore we could roughly describe the parameter passing of \crwl\ as call-by-partial-value,
while \apcrwl\ would perform call-by-set-of-partial-values.

The calculus derives \emph{reduction statements} of the form 
$\prog \vdap e \cltop t$, 
which expresses that $t$ is (or approximates to) a possible value for $e$ in this semantics, under the program $\prog$. For any {\emph \apcrwl-proof} we define its \emph{size} as the number of applications of rules of the calculus. The \emph{\apcrwl-denotation} of an expression $e \in \Exp_\perp$ under a program $\prog$ in \apcrwl\ is defined as $\denapp{e}{\prog} = \{t \in \CTerm_\perp~|~\prog \vdap e \cltop t\}$. In the following, we will usually omit the reference to $\prog$ and $\alpha pl$, and even will skip $\vdap$, when it is clearly implied by the context.

\begin{figure}[t!]
\begin{center}
 \framebox{
\begin{minipage}{.95\textwidth}
\begin{center}
\begin{small}
\begin{tabular}{l@{~~~~~}l}
\regla{RR}{X \clto X}{}{$X \in \var$} \hspace{2ex}
\regla{DC}{c(e_1, \ldots, e_n) \clto c(t_1, \ldots, t_n)}{e_1 \clto t_1 \ldots e_n \clto t_n}{$c \in {\CS}^n$} \\[.81cm]
\regla{B}{e \clto \perp}{}{} \hspace{2ex}
\regla{\apor}{f(e_1, \ldots, e_n) \clto t}{\begin{array}{c}
e_1 \cltop p_1\theta_{11} \\
 \ldots \\
e_1 \cltop p_1\theta_{1 m_1} \\
\end{array}
\ldots~ \begin{array}{c}
		e_n \cltop p_n\theta_{n1} \\
	      \ \ldots \\
              \  e_n \cltop p_n\theta_{n m_n} \\
	 \end{array}~
\begin{array}{c}
          ~ \\
          ~ \\
          r\theta \cltop t
          \end{array}
}{} \\
\hspace{30ex} $\begin{array}{l}
\textrm{if } (f(\overline{p}) \tor r)\in {\cal P}\mbox{, } 
\forall i \in \{1,\ldots,n\}~\Theta_i = \{\theta_{i 1}, \ldots, \theta_{i m_i}\}\\
\theta = (\biguplus\limits_{i=1}^n ?\Theta_i) \uplus \thetae, 
\forall i \in \{1,\ldots,n\}, j \in \{1, \ldots, m_i\}\\
\mi{dom}(\theta_{i j}) \subseteq \mi{var}(p_i), \forall i \in \{1, \ldots, n\}~m_{i} > 0\\ 
\mi{dom}(\thetae) \subseteq \vextra{f(\overline{p}} \tor r), \thetae \in \icsus
  \end{array}$
\end{tabular}
\end{small}
\end{center}
\end{minipage}
}
\end{center}
    \caption{Rules of \apcrwl}
    \label{fig:apcrwl}
\end{figure}

\begin{example}\label{EPlural1}
Consider the program of Example \ref{EIntro1}, that is $\{f(c(X)) \tor d(X,X),$ $X~?~Y \tor X,$ $X~?~Y \tor Y\}$. The following is a \apcrwl-proof for the statement $f(c(0)~?~c(1)) \cltop d(0, 1)$ (some steps have been omitted for the sake of conciseness):\\
\begin{center}
\begin{footnotesize}
$
\scalebox{.92}{
\infer[\mbox{\apor}]{~~~~~~~~~~~~~~~~f(c(0)?c(1)) \cltop d(0, 1)~~~~~~~~~~~~~~~~}
           {
           \infer[\mbox{\apor}]{~~~~~~~~~~~~~~~~~~~c(0)?c(1) \cltop c(0)~~~~~~~~~~~~~~~~~~~}
                 {
                 \infer[\crule{DC}]{c(0) \cltop c(0)}{\infer[\crule{DC}]{0 \cltop 0}{}} &
                \ \infer[\crule{B}]{c(1) \cltop \perp}{} &
                \ \infer{c(0) \cltop c(0)}{\vdots}
                  }~
          \ \infer{c(0)?c(1) \cltop c(1)}{\vdots}~
        \  \infer[\crule{DC}]{d(0?1, 0?1) \cltop d(0,1)}
                     {
                      \infer{0?1 \cltop 0}{\vdots} &
                      \infer{0?1 \cltop 1}{\vdots}
                     }
            }
}
$
\end{footnotesize}
\end{center}
%
\end{example}

One of the most important properties of \apcrwl\ is compositionality, a property very close to the DET-additivity property for algebraic specifications of \cite{hussmann93}, or the referencial transparency property of \cite{SS90}. This property shows that the \apcrwl-denotation of any expression put in a context only depends on the \apcrwl-denotation of that expression, and formalizes the idea that the semantics of a whole expression depends only on the semantics of its constituents, as we informally pointed above. 

\begin{theorem}[Compositionality of \apcrwl]\label{lCompPS} 
For any program, $\con \in \Cntxt$ and $e \in \Exp_\perp$:
$$
\denap{\con[e]} = \bigcup_{\{t_1, \ldots, t_n\} \subseteq \denap{e}} \denap{\con[t_1~?~\ldots~?~t_n]}
$$
for any arrangement of the elements of $\{t_1, \ldots,$ $t_n\}$ in $t_1~?~\ldots~?~t_n$. As a consequence, for any $e' \in \Exp_\perp$:
$$
\denap{e} = \denap{e'} \mbox{ iff  } \forall \con \in \Cntxt. \denap{\con[e]} = \denap{\con[e']}
$$
\end{theorem}
\begin{proof}
We have to prove that, for any $t \in \CTerm_\perp$, if $\con[e] \cltop t$ then $\exists \{s_1, \ldots, s_n\} \subseteq \denp{e}$ such that $\con[s_1~?~\ldots~?~s_n] \cltop t$; and conversely, that given $\{s_1, \ldots, s_n\} \subseteq \denp{e}$ such that $\con[s_1~?~\ldots~?~s_n] \cltop t$ then $\con[e] \cltop t$. Each of these statements can be proved by induction on the size of the starting proof, see \refap{PROOFlCompPS} for details. 
\end{proof}
Contrary to what happens to call-time choice \cite{LRSflops08,LMRS10LetRwTPLP}, we cannot have a compositionality result for single values like $\den{\con[e]} = \bigcup_{t \in \den{e}} \den{\con[t]}$ for any arbitrary context $\con$, because $e$ could appear in a function call when put inside $\con$, and that function might demand more that one value from $e$, because of the plurarity of \pcrwl. We can see this considering the program from Example \refp{EIntro1} extended with a function $coin$ defined by $\{coin \tor 0, coin \tor 1\}$, the context $f(c([]))$ and the expression $coin$: in order to compute the value $d(0,1) \in \den{f(c(coin))}$ we need $\{0,1\} \subseteq \den{coin}$ while a single value of $coin$ is not enough, which is reflected in the fact that $d(0,1) \in \den{f(c(0~?~1))}$ while $d(0,1) \not\in \den{f(c(0))} \cup \den{f(c(1))}$. On the other hand, note that we only need a finite subset of the denotation of the expression put in context, but not the whole denotation, which could be infinite thus leading to $t_1 ~?~ \ldots ~?~ t_n$ being a malformed expression, as we only consider finite expressions in this work. To illustrate this we may consider again the program from Example \ref{EIntro1}, the symbols $z \in \CS^0, s \in \CS^1$ for the Peano natural numbers representation, and the function $\mi{from}$ defined as $\{\mi{from}(X) \tor X, \mi{from}(X) \tor s(\mi{from}(X))\}$. Then, using the same context as above and the expression $\mi{from}(z)$, in order to compute $d(z, s(z)) \in \den{f(c(\mi{from}(z)))}$ we just need $\{z, s(z)\} \subseteq \den{\mi{from}(z)}$, but not the infinite set of elements in $\den{\mi{from}(z)}$. The intuition behind this is that, as we use c-terms as values and c-terms are finite, then any computation of a value is a finite process that only involves a finite amount of information: in this case a finite subset of the denotation of the expression put in context. 


\bigskip
Besides compositionality, \apcrwl\ enjoys other nice properties, like the following polarity property.
\begin{proposition}[Polarity of \apcrwl]\label{LMonPlural}
For any program $\prog$, $e, e' \in \Exp_\perp$, $t, t' \in \CTerm_\perp$ if $e \ordap e'$ and $t' \ordap t$ then $\prog \pps e \cltop t$ implies $\prog \pps e' \cltop t'$ with a proof of the same size or smaller.
\end{proposition}
\begin{proof}
By a simple induction on the structure of $e \cltop t$ using basic properties of $\ordap$, see \refap{PROOFLMonPlural} for details. 
\end{proof}

\pcrwl\ also has some monotonicity properties related to substitutions. These are formulated using the preorder $\pordap$ over $\icsus$ defined by $\theta \pordap \theta'$ iff $\forall X \in \var$, given $\theta(X) = t_1~?~\ldots~?~t_n$ and $\theta'(X) = t'_1~?~\ldots~?~t'_m$ then $\forall t \in \{t_1, \ldots, t_n\} \exists t' \in \{t'_1, \ldots, t'_m\}$ such that $t \ordap t'$; and the preorder $\adsord$ over $\Subst_\perp$ defined by $\sigma \adsord \sigma'$ iff $\forall X \in \var, ~\denp{\sigma(X)} \subseteq \denp{\sigma'(X)}$.
\begin{proposition}[Monotonicity for substitutions of \apcrwl]\label{LMonSusPlural} For any program, $e \in \Exp_\perp$, $t \in \CTerm_\perp$, $\sigma, \sigma' \in \Subst_\perp$, $\theta, \theta' \in \icsus$:
\begin{enumerate}
  \item \textbf{Strong monotonicity of {\boldmath $\Subst_\perp$}}: If $\forall X \in \var, s \in \CTerm_\perp$ given $\prog \pps \sigma(X) \cltop s$ with size $K$ we also have $\prog \pps \sigma'(X) \cltop s$ with size $K' \leq K$, then 
$\pps e\sigma \cltop t$ with size $L$ implies $\pps e\sigma' \cltop t$ with size $L' \leq L$.
   \item \textbf{Monotonicity of {\boldmath $\CSubst_\perp$}}: If $\theta, \theta' \in \CSubst_\perp$ and $\theta \ordap \theta'$ then $\prog \pps e\theta \cltop t$ with size $K$ implies $\prog \pps e\theta' \cltop t$ with size $K' \leq K$.
   \item \textbf{Monotonicity of {\boldmath $\Subst_\perp$}}: If $\sigma \adsord \sigma'$ then $\denp{e\sigma} \subseteq \denp{e\sigma'}$.
   \item \textbf{Monotonicity of {\boldmath $\icsus$}}: If $\theta \pordap \theta'$ then $\denp{e\theta} \subseteq \denp{e\theta'}$.
\end{enumerate}
\end{proposition}

The properties of \apcrwl\ we have seen so far are shared with \crwl, which is something natural taking into account that \apcrwl\ is a modification of that semantics. 
Nevertheless, there are some properties of \crwl---and as a consequence, of call-time choice---that do not hold for \apcrwl. One of these is the correctness of the \textit{bubbling} operational rule \cite{AntoyBrownChiang06Termgraph}, which can be formulated as ``under any program and for any $\con \in \Cntxt$, $e_1, e_2 \in \Exp_\perp$ we have that $\den{\con[e_1~?~e_2]} = \den{\con[e_1]~?~\con[e_2]}$''. 
Note that Examples \ref{EIntro1} and \ref{EIntro2} already show that this property does not hold for run-time choice, the following (counter)example proves that it is not the case for \apcrwl\ neither.
\begin{example}\label{ex:noBubling}
Consider the program $\prog = \{\pair(X) \tor (X,X), X~?~Y \tor X, X~?~Y \tor Y\}$ and the expressions $\pair(0~?~1)$ and $\pair(0)~?~\pair(1)$ which correspond to a bubbling step using $\con = \pair([])$. It is easy to check that $(0, 1) \in \denp{\pair(0~?~1)}$ while $(0, 1) \not\in \denp{\pair(0)~?~\pair(1)}$.
\end{example}
It was very enlightening for us to discover that the correctness of bubbling does not hold for \pcrwl, and in fact in  \cite{rodH08} it was wrongly considered as true. This shows that \crwl\ and \pcrwl\ are more different that it may appear at a first sight. In particular, regarding to bubbling, the important difference is that while \pcrwl\ is only compositional w.r.t.\ subsets of the denotation, \crwl\ is compositional w.r.t.\ single 
values of the denotation, as we saw above. 
Compositionality w.r.t.\ single values is stronger than compositionality w.r.t.\ subsets of the denotation, as the former implies the latter, and this is also exemplified by the fact that we need compositionality w.r.t.\ single values for bubbling to be correct, as we will see soon. On the other hand, compositionality w.r.t.\ subsets of the denotation is enough to obtain the result expressed at the end of Theorem \ref{lCompPS}, showing that expressions with the same values 
are indistinguishable, which corresponds to the value-based philosophy of FLP. 

As the bubbling rule is devised to improve the efficiency of computations \cite{AntoyBrownChiang06Termgraph}, it would be nice to be able to use it in some situations, although it would only be for a restricted class of contexts. 
In this line, we have found that bubbling is still correct under \apcrwl\ for a particular kind of contexts called \emph{constructor contexts} or just \emph{c-contexts}, which are contexts whose holes appear under a nested application of constructor symbols only, that is, $\ccon ::= [~]~|~c(e_1, \ldots, \ccon, \ldots$ $, e_n)$, with $c \in \CS^n, e_1, \ldots, e_n \in \Exp_\perp$. 
%
%
For c-contexts, \apcrwl\ enjoys the same compositionality  for single values as \crwl---that property holds in \crwl\ for arbitrary contexts---, as shown in the following result. 

\begin{proposition}[Compositionality of \apcrwl\ for c-contexts]\label{lemCompAPlCCon}
For any program, c-context $\ccon$ and $e \in \Exp_\perp$:
$$\denap{\ccon[e]} = \bigcup\limits_{t \in \denap{e}} \denap{\ccon[t]}$$
\end{proposition}
\begin{proof}
Very similar to the proof for the general compositionality of \apcrwl\ from Theorem~\ref{lCompPS}, see \refap{PROOFlemCompAPlCCon} for details. 
\end{proof}

As compositionality for single values is the key property needed for bubbling to be correct, we get the following result for bubbling in \apcrwl.
\begin{proposition}[Bubbling for c-contexts in \apcrwl]\label{bubAPlural}
For any program, c-context $\ccon$ and $e_1, e_2 \in \Exp_\perp$, $\denap{\ccon[e1~?~e_2]} = \denap{\ccon[e_1]~?~\ccon[e_2]}$.
\end{proposition}
\begin{proof}
It is easy to prove that $\forall e_1, e_2 \in \Exp_\perp$ we have $\denap{e1~?~e_2} = \denap{e_1} \cup \denap{e_2}$ (see 
\proofs). But then:
$$
\begin{array}{ll}
\denap{\ccon[e1~?~e_2]} \\
= \bigcup_{t \in \denap{e1~?~e_2}} \denap{\ccon[t]} & \mbox{by Proposition \ref{lemCompAPlCCon}} \\
= \bigcup_{t \in \denap{e1} \cup \denap{e_2}} \denap{\ccon[t]} \\ 
= \bigcup_{t \in \denap{e1}} \denap{\ccon[t]} \cup \bigcup_{t \denap{e_2}} \denap{\ccon[t]} \\
= \denap{\ccon[e_1]} \cup \denap{\ccon[e_2]} & \mbox{by Proposition \ref{lemCompAPlCCon}} \\
= \denap{\ccon[e_1]~?~\ccon[e_2]} \\ 
\end{array}
$$
\end{proof}


We end our presentation of \apcrwl\ with an example showing how we can use \apcrwl\ to model problems in which some collecting work has to be done. 
\begin{example}\label{EPlural2}
We want to represent the database of a bank in which we hold some data about its employees. This bank has several branches and we want to organize the information according to them. To do that we define a non-deterministic function $\branches$ to represent the set of branches: a set is then identified with a non-deterministic expression. We also use this technique to define non-deterministic function $\employees$ which conceptually returns, for a given branch, the set of records containing the information regarding the employees that work in that branch. Now we want to search for the names of two clerks, which may be working in different branches. To do that we define the function $\twoclerks$ which is based upon the function $\find$, which forces the desired pattern $e(N,G,\clerk)$ over the set defined by the expression $\employees(\branches)$:
 $$
 \begin{array}{ll}
 \prog = & \{\branches \tor \mad,\\
 & \;\; \branches \tor \vigo,\\
 & \;\; \employees(\mad) \tor e(\pepe, \man, \clerk), \\
 & \;\; \employees(\mad) \tor e(\paco, \man, \clerk), \\
 & \;\; \employees(\vigo) \tor e(\maria, \woman, \clerk), \\
 & \;\; \employees(\vigo) \tor e(\jaime, \woman, \clerk), \\
 & \;\; \twoclerks \tor \find(\employees(\branches)), \\
 & \;\; \find(e(N,G,\clerk)) \tor (N,N) \}\\
 \end{array}
 $$
With term rewriting $\twoclerks \rw \find(\employees(\branches)) \not\rw^* (\pepe,\maria)$, because in that expression the evaluation of $\branches$ is needed and thus one of the branches must be chosen. 
On the other hand with \pcrwl\ the value $(\pepe,\maria)$ can be computed for $\twoclerks$ (some steps have been omitted for the sake of conciseness,
$\mi{emps}$ abbreviates $\mi{employees}$, and $\mi{brs}$ abbreviates $\branches$):
\begin{center}
\begin{footnotesize}
$
\infer[\mbox{\apor}]{\twoclerks \cltop (\pepe, \maria)}
          {
         \infer[\mbox{\apor}]{\find(\mi{emps}(\mi{brs})) \cltop (\pepe, \maria)}
                    {
                    \begin{array}{c}
                     \!\!\!\infer[\mbox{\apor}]{\mi{emps}(\mi{brs}) \cltop e(\pepe, \perp, \clerk)}{\vdots}\\[.5cm]
                \!\!\!\infer[\mbox{\apor}]{\mi{emps}(\mi{brs}) \cltop e(\maria, \perp, \clerk)}{\vdots}
                    \end{array}
            \ \infer[\textbf{\textsc{DC}}]{(\pepe~?~\maria, \pepe~?~\maria) \cltop (\pepe, \maria)}{\vdots}
                    }
          }
$
\end{footnotesize}
\end{center}
where
\begin{center}
\begin{footnotesize}
$
\infer[\mbox{\apor}]{\mi{emps}(\mi{brs}) \cltop e(\pepe, \perp, \clerk)}
      {
      \infer[\mbox{\apor}]{\mi{brs} \cltop \mad}{
		\infer[\textbf{\textsc{DC}}]{\mad \cltop \mad}{}
      } &
      \infer[\textbf{\textsc{DC}}]{e(\pepe, \man, \clerk) \cltop e(\pepe, \perp, \clerk)}{\vdots}
                               }
$
\end{footnotesize}
 \end{center}

\end{example}

\subsection{\bpcrwl}\label{sec:bplural}
So far we have presented our first proposal for a plural semantics for constructor systems, seen some interesting properties, and how to use it to solve collecting problems. Nevertheless this semantics has also some weak points, that will be illustrated by the following example.
\begin{example}\label{ExBPlural1}
Starting from the program of Example \ref{EPlural2}, we want to search for the names of two clerks paired with their corresponding genders. Therefore, following the same ideas, we define a function $\fng$ that forces the desired pattern but now returning both the name and the gender of two clerks, by the rule $\fng(e(N, G, \clerk)) \tor ((N, G), (N, G))$.
Then, $((\pepe, \man), (\maria, \woman))$ would be one of the values computed for the expression $\fng(\employees(\branches))$, as expected.
Nevertheless we can also compute the value $((\pepe, \woman), (\maria, \man))$, which obviously does not correspond to the intended meaning of $\fng$, as can be seen in the following proof (using the abbreviations above and also $m$ for $\man$,
and $w$ for $\woman$).

\begin{footnotesize}
$$
\infer[\mbox{\apor}]{\hspace{2cm}\fng(\mi{emps}(\mi{brs})) \cltop ((\pepe, w), (\maria, m))\hspace{2cm}}
      {
       \begin{array}{l}
       \infer{\mi{emps}(\mi{brs}) \cltop e(\pepe, m, \clerk)}{\vdots} \\
       \infer{\mi{emps}(\mi{brs}) \cltop e(\maria, w, \clerk)}{\vdots} 
       \end{array} 
   \ \begin{array}{l}
   ~ 	\\    ~ 	\\
  \infer{\begin{minipage}{0.522\textwidth}$\begin{array}{l}
((\pepe~?~\maria, m~?~w), (\pepe ~?~ \maria, m~?~w)) \\ 
~~~\cltop ((\pepe, w), (\maria, m))   
   \end{array}$\end{minipage}
     }{\vdots} 
         \end{array}
  }
$$
\end{footnotesize}

\end{example}

This example is interesting because it shows a relevant flaw of \apcrwl, since there the matching substitutions $[N/\pepe, G/\man]$ and $[N/\maria, G/\woman]$ obtained for the different evaluations of the argument $\employees(\branches)$ become wrongly intermingled. Anyway the program is not well conceived, as it does not specify that each of the $(N, G)$ pairs correspond to a particular clerk in the database, thus preventing an unintended information mixup. Nevertheless a better semantic behaviour would have prevented ``mixed'' results like $((\pepe, \woman), (\maria, \man))$ thus getting $((\maria, \woman), (\maria, \woman))$ and $((\pepe, \man), (\pepe, \man))$ as the only total values for $\fng(\employees(\branches))$, which does not 
fix the program but at least avoids wrong information mixup.\footnote{In Section \ref{sec:examples} we will see how to combine singular and plural \emph{function arguments} to solve a generalization of this problem.}

This problem was also pointed out in \cite{BB09}, where an identification between $d(0,0)~?~d(1,1)$ and $d(0~?~1, 0~?~1)$---for $d \in \CS^2$ and $0,1 \in \CS^0$---made by \apcrwl\ for relevant contexts was reported. In the technical setting presented in 
that paper another plural semantics that avoids this problem is proposed, although its technical relation with call-time or run-time choice is not formally stated nor proved. 
In that work, that particular plurality is achieved 
by allowing bubbling steps for constructor applications by means of a rule that could be expressed in our syntax as $\den{c(e_1, \ldots, e'_1~?~e'_2, \ldots, e_n)} = \den{c(e_1, \ldots, e'_1, \ldots, e_n)~?~c(e_1, \ldots, e'_2, \ldots, e_n)}$. This kind of rules are well suited for a step-wise semantics like the one presented in \cite{BB09}, but are more difficult to integrate with a goal-oriented proof calculus in the style of \crwl\ or \apcrwl, 
which---as we saw in the presentation of \apcrwl\ above---perform a kind of innermost evaluation of expressions by exploiting the use of partial values to get a compositional calculus for a lazy semantics.

Hence, in order to adapt this idea to our framework, we could \emph{switch from bubbling under constructors to bubbling of $\icsus$, allowing the combination of substitutions that only differ in the value they assign to a single variable}. This can be realized by defining a binary operator $\bunion$ to combine partial c-substitutions and a reduction notion $\bubsus$ defined by the rule $(\theta ~\uplus~ [X/e_1]) \bunion (\theta ~\uplus~ [X/e_2]) \bubsus \theta ~\uplus~ [X/e_1~?~e_2]$, that corresponds to a bubbling step for substitutions. Using this 
we could for example perform the following bubbling derivation for substitutions.
$$
\begin{array}{l}
\underline{[X/0, Y/0] \bunion [X/0, Y/1]} \bunion [X/1, Y/0] \bunion [X/1, Y/1] \\
\bubsus [X/0, Y/0~?~1] \bunion \underline{[X/1, Y/0] \bunion [X/1, Y/1] }\\
\bubsus \underline{[X/0, Y/0~?~1] \bunion [X/1, Y/0~?~1]} 
\bubsus [X/0~?~1, Y/0~?~1]
\end{array}
$$
This derivation shows a criterion that determines that the set of c-substitutions $\{[X/0, Y/0], [X /0, Y/1], [X/1, Y/0], [X/1, Y/1]\}$ can be safely combined into $[X/0~?$ $1, Y/0~?~1] \in \icsus$ with no wrong substitution mixup. On the other hand, for $[X/0, Y/0] \bunion [X/1, Y/1]$ we should not be able to perform any $\bubsus$ step as these substitutions differ in more than one variable, thus failing to combine those c-substitutions into a single element from $\icsus$. As can be seen in Figure \ref{fig:apcrwl}, the key for getting a plural behaviour in \apcrwl\ is finding a way to combine different matching substitutions obtained from the evaluation of the same expression, therefore this new combination method should give rise to another plural semantic proposal. We conjecture that the resulting semantics expresses the same plural semantics proposed in \cite{BB09}---the one resulting in that setting when only variables of sort $\mi{Ch}$ (as defined in that paper) are used---, although we will not give any formal result relating both proposals. Let us call \bpcrwl\ to this new semantics 
in which parameter passing is only perfomed with substitutions from $\icsus$ that come from a succesful combination of c-substitutions  using the relation $\bubsus$,  
and consider the behaviour of the different plural semantics in the following example. 
\begin{example}\label{EPlural3}
Consider the constructors $c \in \CS^1$, $d \in \CS^2$, $l \in \CS^4$ and $0, 1 \in \CS^0$, and the following program.
$$
\begin{array}{ll}
f(c(X)) \tor d(X, X) & h(d(X,Y)) \tor d(X, X)\\
g(d(X, Y)) \tor l(X, X, Y, Y) & k(d(X, Y)) \tor d(X, Y) \\
\end{array}
$$

\begin{itemize}
	\item $f(c(0)~?~c(1))$ and $f(c(0~?~1))$ behave the same in both \apcrwl\ and \bpcrwl. In this case
there is only one variable involved in the matching substitution and thus no substitution mixup like the ones seen before may appear. That is, for both expressions we only have to combine the substitutions $[X/0]$ and $[X/1]$, thus reaching the values $d(0, 0)$, $d(0, 1)$, $d(1, 0)$, and $d(1,1)$ in both semantics. 
    \item More surprisingly we also get the same behaviour for $h(d(0,0)~?~d(1,1))$ and $h(d(0~?~1, 0~?~1))$ in both \apcrwl\ and \bpcrwl. There the suspicious expression is $h(d(0,0)~?~d(1,1))$ which generates the matching substitutions $[X/0, Y/0]$ and $[X/1, Y/1]$ which are wrongly combined by \apcrwl\ into the substitution $?\{[X/0, Y/0], [X/1, Y/1]\} = [X/0~?~1, Y/0~?~1]$, used to instantiate the right-hand side of the rule for $h$. But this mistake has no consequence because only $X$ appears in the right-hand side of the rule for $h$, therefore it has the same effect as combining $[X/0, Y/\perp]$ and $[X/1, Y/\perp]$ into $[X/0~?~1, Y/\perp]$, which is just what is done in \bpcrwl\ as we will see later on.

On the other hand $h(d(0~?~1, 0~?~1))$ is not problematic as it generates the matching substitutions $[X/0, Y/0]$, $[X/0, Y/1]$, $[X/1, Y/0]$ and $[X/1, Y/1]$ that already cover all the possible instantiations of $X$ and $Y$ caused by its combination in \apcrwl, the substitution $[X/0~?~1, Y/0~?~1]$. The point is that in a sense both $\{[X/0, Y/0], [X/0, Y/1], [X/1, Y/0], [X/1, Y/1]\}$ and $[X/0~?~1, Y/0~?~1]$ have the same power. This will also be reflected by the fact that \bpcrwl\ would be able to combine the former set into the latter $\icsus$.

Again, we can reach the values $d(0, 0)$, $d(0, 1)$, $d(1, 0)$ and $d(1,1)$ for each expression in both semantics. 

    \item It is for the expressions $g(d(0,0)~?~d(1,1))$ and $g(d(0~?~1, 0~?~1))$ that we can see a different behaviour of \apcrwl\ and \bpcrwl. Once again $g(d(0~?~1, 0~?~1))$ is not problematic, and for it we can get the values $l(0, 0, 0, 0)$, $l(0, 0, 0, 1)$, \ldots and all the combinations of $0$ and $1$, in both semantics. But for $g(d(0,0)~?~d(1,1))$ we have that, for example, to compute $l(0, 0, 0, 1)$ we need the expression $d(0,0)~?~d(1,1)$ to generate both $0$ and $1$ for $Y$ in the matching substitutions. 
The only (total) matching substitutions that can be obtained from the evaluation of $d(0,0)~?~d(1,1)$ are
$[X/0, Y/0]$ and $[X/1, Y/1]$, which cannot be combined by \bpcrwl, hence we cannot get both $0$ and
$1$ for $Y$ in the combined substitution.
%
As a consequence
$l(0,0,0,0)$ and $l(1,1,1,1)$ are the only values computed for $g(d(0,0)~?~d(1,1))$ by \bpcrwl. On the other hand, \apcrwl\ computes all the combinations of $0$ and $1$---like it did for $g(d(0~?~1, 0~?~1))$---, as it is able to combine $\{[X/0, Y/0], [X/1, Y/1]\}$ into $[X/0~?~1, Y/0~?~1]$.

    \item A more exotic discovery is that $k(d(0,0)~?~d(1,1))$ does not behave the same for call-time choice, run-time choice, \apcrwl, and \bpcrwl, 
even though it only uses a right-linear program rule, and 
it is a known fact that 
call-time choice and run-time choice are equivalent 
for right-linear programs \cite{hussmann93}. 
\crwl\ (call-time choice), term rewriting (run-time choice), and \bpcrwl\ only compute the values $d(0,0)$ and $d(1,1)$ for $k(d(0,0)~?~d(1,1))$, in the case of \bpcrwl\ because it fails to combine $[X/0, Y/0]$ and $[X/1, Y/1]$. Nevertheless \apcrwl\ is able to combine those substitutions into $[X/0~?~1,$ $Y/0~?~1]$, thus getting the additional values $d(0,1)$ and $d(1,0)$ for the expression $k(d(0,0)~?~d(1,1))$. 
However we still strongly conjecture that \bpcrwl---as formulated below---is equivalent to call-time and run-time choice for right-linear programs.
\end{itemize}

\end{example}

The previous example motivates the interest of a formal definition of \bpcrwl. It would be nice if it were by means of a proof calculus similar to \crwl\ and \apcrwl, because then their comparison would be easier, and maybe they could even share some of their properties, in particular compositionality. The ideas above regarding bubbling derivations for substitutions have given us the right intuitions, but those derivations are not so easy to handle as the following characterization of \emph{compressible sets of c-substitutions}
illustrates, which will be the only sets of substitutions that will be combined by \bpcrwl.

\begin{definition}[Compressible set of $\CSubst_\perp$]\label{def:compress}~\\
A finite set $\Theta \subseteq \CSubst_\perp$ is \emph{compressible} iff 
for $\{X_1, \ldots, X_n\} = \bigcup_{\theta \in \Theta} \mi{dom}(\theta)$
$$
\{(X_1\theta, \ldots, X_n\theta)~|~\theta \in \Theta\} = \{X_1\theta_1~|~\theta_1 \in \Theta\} \times \ldots \times \{X_n\theta_n~|~\theta_n \in \Theta\}
$$

Note that this property is easily computable for $\Theta$ finite, as we only consider finite domain substitutions.
\end{definition}
\begin{example}\label{ex:compressSetCSubst}
Let us see how 
the notion of compressible set of c-substitutions can be used to replace the relation $\bubsus$ sketched above. We have seen that the substituions $[X/0, Y/0]$ and $[X/1, Y/1]$ should not be combined in order to prevent a wrong substitution mixup. This is reflected in the fact that the set $\{[X/0, Y/0], [X/1, Y/1]\}$ is not compressible, because:
$$
\begin{array}{l}
\{(X\theta, Y\theta)~|~\theta \in \{[X/0, Y/0], [X/1, Y/1]\}\} 
= \{(0,0), (1,1)\} \\
\neq \{(0, 0), (0, 1), (1, 0), (1, 1)\} 
= \{0, 1\} \times \{0, 1\} \\
= \{X\theta_x~|~\theta_x \in \{[X/0, Y/0], [X/1, Y/1]\}\} \times \{Y\theta_y~|~\theta_y \in \{[X/0, Y/0], [X/1, Y/1]\}\}
\end{array}
$$
On the other hand for $\Theta = \{[X/0, Y/0], [X/0, Y/1], [X/1, Y/0], [X/1, Y/1]\}$ the substitutions it contains can be safely combined, therefore we should have that $\Theta$ is compressible, as it happens:
$$
\begin{array}{l}
\{(X\theta, Y\theta)~|~\theta \in \Theta\} 
= \{(0, 0), (0, 1), (1, 0), (1, 1)\} \\
= \{0, 1\} \times \{0, 1\} 
= \{X\theta_x~|~\theta_x \in \Theta\} \times \{Y\theta_y~|~\theta_y \in \Theta\}
\end{array}
$$
\end{example}

Our last proposal for a plural semantics for \ctrss\ is based on the notion of compressible set of c-substitutions, and it is defined by the \bpcrwl-proof calculus in Figure~\ref{fig:bpcrwlAltD}. 
Note that the only difference with \apcrwl\ is that the rule \apor\ is replaced by \bpor, that now demands the different matching substitutions obtained from the evaluation of each function argument to be compressible. Apart from that, compressible sets of partial c-substitutions are combined just like in \apcrwl, by means of the $?$ operator.

This calculus, like \crwl\ and \apcrwl, also derives \emph{reduction statements} of the form $\prog \vdbp e \cltop t$, 
which expresses that $t$ is (or approximates to) a possible value for $e$ in this semantics, under the program $\prog$. 
Then the \emph{\bpcrwl-denotation} of an expression $e \in \Exp_\perp$ under a program $\prog$ in \bpcrwl\ is defined as $\denbpp{e}{\prog} = \{t \in \CTerm_\perp~|~\prog \vdbp e \cltop t\}$. In the following, we will usually omit the reference to ${\cal  P}$ when implied by the context.

\begin{figure}[t]
\begin{center}
 \framebox{
\begin{minipage}{.95\textwidth}
\begin{center}
\begin{small}
\begin{tabular}{l@{~~~~~}l}
\regla{RR}{X \clto X}{}{$X \in \var$} \hspace{2ex}
\regla{DC}{c(e_1, \ldots, e_n) \clto c(t_1, \ldots, t_n)}{e_1 \clto t_1 \ldots e_n \clto t_n}{$c \in {\CS}^n$} \\[.25cm]
\regla{B}{e \clto \perp}{}{} \hspace{2ex}
\regla{\bpor}{f(e_1, \ldots, e_n) \clto t}{\begin{array}{c}
e_1 \cltop p_1\theta_{11} \\
 \ldots \\
e_1 \cltop p_1\theta_{1 m_1} \\
\end{array}
\ldots~ \begin{array}{c}
		e_n \cltop p_n\theta_{n1} \\
	      \ \ldots \\
              \  e_n \cltop p_n\theta_{n m_n} \\
	 \end{array}~
\begin{array}{c}
          ~ \\
          ~ \\
          r\theta \cltop t
          \end{array}
}{} \\
\hspace{23ex} $\begin{array}{l}
\textrm{if } (f(\overline{p}) \tor r)\in {\cal P}\mbox{, }
\forall i \in \{1, \ldots, n\}~\Theta_i = \{\theta_{i 1}, \ldots, \theta_{i m_i}\}\\
\mbox{is compressible,  }
\theta = (\biguplus\limits_{i=1}^n ?\Theta_i) \uplus \thetae, 
\forall i \in \{1, \ldots, n\}, \\
 j \in \{1, \ldots, m_i\} 
\mi{dom}(\theta_{i j}) \subseteq \mi{var}(p_i), \forall i \in \{1, \ldots, n\}~
m_{i} > 0\\ 
\mi{dom}(\thetae) \subseteq \vextra{f(\overline{p}} \tor r), \thetae \in \icsus
  \end{array}$
\end{tabular}
\end{small}
\end{center}
\end{minipage}
}
\end{center}
    \caption{Rules of \bpcrwl}
    \label{fig:bpcrwlAltD}
\end{figure}

\begin{example}\label{ExBPlural2}
Consider the program of Example \ref{ExBPlural1}, \bpcrwl\ is able to avoid computing the value $((\pepe,$ $\woman), (\maria, \man))$ for the expression $\fng(\employees(\branches))$ because the set of matching substitutions $\{[N/\pepe, G/\man], [N/\maria, G/\woman]\}$ is not compressible, as 
can be easily checked by applying Definition \ref{def:compress} in a way similar to Example \ref{ex:compressSetCSubst}.  
%
Nevertheless, 
the values $((\maria, \woman), (\maria, \woman))$ and $((\pepe, \man), (\pepe, \man))$ can be computed for $\fng(\employees(\branches))$ by using the sets of substitutions $\{[N/\pepe, G/\man]\}$ and $\{[N/\maria, G/\woman]\}$, respectively, for parameter passing, which are compressible as they are singletons. 
As we saw in Example \ref{ExBPlural1}, the function $\fng$ is wrongly conceived because it does not specify that in each pair $(N, G)$ the name $N$ and the genre $G$ must  correspond to the same clerk. \bpcrwl\ cannot fix a wrong program, but at least is able to prevent ``mixed'' results like $((\pepe, \woman), (\maria, \man))$. 

It is also easy to check that \bpcrwl\ has the same behaviour that \apcrwl\ for Example \ref{EPlural2}, as sets like $\{[N/\pepe, G/\!\!\!\perp], [N/\maria, G/\!\!\!\perp]\}$ are compressible. Similarly, in Example \ref{EPlural3} 
the functions $f$ and $h$ behave the same under both semantics, and \bpcrwl\ also behaves for $h$ and $k$ as specified there, because $\{[X/0, Y/0], [X/1, Y/1]\}$ is not compressible, just like $\{[N/\pepe, G/\man], [N/\maria,$ $G/\woman]\}$, while for $\Theta = \{[X/0, Y/0], [X/0, Y/1], [X/1, Y/0], [X/1, Y/1]\}$ we have that $\Theta$ is compressible, as 
seen in Example \ref{ex:compressSetCSubst}. 
\end{example}

The following result shows that part of the equality that defines compressibility always holds trivially, thus simplifying the definition of compressible set of c-substitutions. 

\begin{lemma}\label{lemmaCompress1}
For any finite set $\Theta \subseteq \CSubst_\perp$ for $\{X_1, \ldots, X_n\} = \bigcup_{\theta \in \Theta} \dom(\theta)$ we have 
$$
\{(X_1\theta, \ldots, X_n\theta)~|~\theta \in \Theta\} \subseteq \{X_1\theta_1~|~\theta_1 \in \Theta\} \times \ldots \times \{X_n\theta_n~|~\theta_n \in \Theta\}
$$
As a consequence $\Theta$ is compressible iff
$$
\{(X_1\theta, \ldots, X_n\theta)~|~\theta \in \Theta\} \supseteq \{X_1\theta_1~|~\theta_1 \in \Theta\} \times \ldots \times \{X_n\theta_n~|~\theta_n \in \Theta\}
$$
This gives another criterion to prove compressibility: $\Theta$ is compressible iff $\forall \theta_1, \ldots, \theta_n$ $\in \Theta.~ \exists \theta \in \Theta$ such that $\forall i. X_i\theta_i \equiv X_i\theta$ (which implies that $(X_1\theta_1, \ldots, X_n\theta_n) \equiv (X_1\theta, \ldots, X_n\theta)$).
\end{lemma}
%
In a way this result exemplifies why \bpcrwl\ is smaller than \apcrwl\ in the sense that in general it computes less values for a given expression under a given program, 
as $\{X_1\theta_1~|~\theta_1 \in \Theta\} \times \ldots \times \{X_n\theta_n~|~\theta_n \in \Theta\}$ corresponds to the substitution $?\Theta$ that is always used for parameter passing in \apcrwl, with no previous compressibility test. We will see more about the relations between call-time choice, run-time choice, \apcrwl, and \bpcrwl\ in Section \ref{sec:hierarchy}.

~\\
We have just seen how \bpcrwl\ corrects the excessive permissiveness of the combinations of substitutions performed by \apcrwl\ but, will it be able to do it while keeping the nice properties of \apcrwl---in particular compositionality---at the same time? Fortunately the answer is yes, as shown by the following result. 
\begin{theorem}[Basic properties of \bpcrwl]\label{ThPropsBetaPlHeredadas}
The basic properties of \apcrwl\ also hold for \bpcrwl\ under any program, i.e, the corresponding versions of Theorem \ref{lCompPS}, Proposition \ref{LMonPlural}, Proposition \ref{LMonSusPlural}, Proposition \ref{lemCompAPlCCon}, and Proposition \ref{bubAPlural} also hold for \bpcrwl. 

For Proposition \ref{LMonSusPlural} in particular we replace $\adsord$ with $\bdsord$, which is defined in terms of \bpcrwl\ instead of \apcrwl, i.e., $\sigma \bdsord \sigma'$ iff $\forall X \in \var, \denbp{\sigma(X)} \subseteq \denbp{\sigma'(X)}$. 
Nevertheless, in the following we will often omit the superscripts $\alpha \mi{pl}$ and $\beta \mi{pl}$ in $\adsord$ and $\bdsord$ when those are implied by the context. 
\end{theorem}
\begin{proof}
In each proof for the \apcrwl\ versions of these results we start from a given \apcrwl-proof and build another one using a bigger expression w.r.t.\ $\ordap$, a more powerful substitution, interchanging an expression with an alternative of some of its values\ldots\ Therefore we can use the same technique for \bpcrwl\ to replicate any \bpor\ step in the starting \bpcrwl-proof by using the substitution used there for parameter passing, which must be compressible by hypothesis, and that we are able to obtain by using a similar reasoning to that performed in the proof for the corresponding result for \apcrwl.
\end{proof}

~\\
In this section we have presented two different proposals for a plural semantics for non-deterministic constructor systems that are different from run-time choice. The first one, \apcrwl, is a pretty simple extension of \crwl\ that comes up naturally from allowing the combination of several matching substitution through the operator $?$ for c-substitutions. But it is precisely the simplicity of that combination which leads to a wrong information mixup in some situations. These problems are solved in \bpcrwl, in which a compressibility test is added to prevent a wrong combination of substitutions. This could suggest that \apcrwl\ is only a preliminary attempt that should now be put aside and forgotten. Nevertheless \apcrwl\ will still be very useful for us, again because of its simplicity, as we will see in subsequent sections. 

Finally note that both \apcrwl\ and \bpcrwl\ have been devised starting from \crwl\ and then adding some criterion for combining different matching substitutions for the same argument, so any number of alternative plural---and even also compositional, possibly---semantics for constructor systems could be conceived just by defining new combination procedures.

\section{Hierarchy, equivalence, and simulation}\label{sec:hierarchy}

In this section we will first compare the different characteristics of the semantics considered so far, with a special emphasis in the set of computed c-terms. 
Then we will present a class of programs characterized by a simple syntactic criterion under which our two plural semantics are equivalent. Finally we will conclude the section presenting 
a program transformation that can be used to simulate our plural semantics by using term rewriting. 

\subsection{A hierarchy of semantics}\label{sect:hierarchy}
We have already seen that \crwl, \apcrwl, and \bpcrwl\ enjoy similar properties like polarity, monotonicity for substitutions and, above all, compositionality, which implies that two expressions have the same denotation if and only if they have the same denotation when put under the same arbitrary context. 
This is not the case for run-time choice, as we saw when switching from $f(c(0~?~1))$ to $f(c(0)~?~c(1))$ in Examples \ref{EIntro1} and \ref{EIntro2}, taking into account that for the expressions $c(0~?~1)$ and $c(0)~?~c(1)$ the same values are computed under run-time choice, i.e., the same c-terms are reached by a term rewriting derivation.\footnote{In fact compositionality can be achieved for run-time choice by using a different set of values instead of the partial c-terms considered in this work. Those values essentially are recursively nested applications of constructor symbols to sets of values structured in the same way, therefore intrinsically more complicated than plain c-terms, and anyway not considered in the present work---See \cite{Lopez-FraguasRS09-RTA09} for details.}

But our main goal in this section is to study the relationship between call-time choice, run-time choice, \apcrwl, and \bpcrwl\ w.r.t.\ the denotations they define, which express the set of values computed by each semantics. To do that we will lean on a traditional notion from the \crwl\ framework, the notion of {\it shell $|e|$ of an expression $e$}, which represents the outer constructor (thus partially computed) part of $e$, defined as $|\!\!\perp\!\!| = \perp$, $|X| = X$, $c(e_1, \ldots, e_n) = c(|e_1|, \ldots, |e_n|)$, $|f(e_1, \ldots, e_n)| = \perp$, for $X \in \var, c \in \CS, f \in \FS$. Now we can define our notion of denotation of an expression in each of the semantics considered.

\begin{definition}[Denotations]\label{def:denotations}
For any program $\prog$, $e \in \Exp$ we define the denotation of $e$ under the different semantics as follows
\begin{itemize}
	\item $\densp{e}{\prog} = \{t \in \CTerm_\perp \mid \prog \vds e \clto t\}$.
	\item $\denrp{e}{\prog} = \{t \in \CTerm_\perp \mid \prog \vdash e \rw^* e' \wedge t \ordap |e'|\}$.
	\item $\denapp{e}{\prog} = \{t \in \CTerm_\perp \mid \prog \vdap e \clto t\}$.
	\item $\denbpp{e}{\prog} = \{t \in \CTerm_\perp \mid \prog \vdbp e \clto t\}$.
\end{itemize}
In the following, we will usually omit the reference to ${\cal  P}$ when implied by the context.
\end{definition}

As \apcrwl\ and \bpcrwl\ are modifications of \crwl, the relation between these three semantics is straightforward.
\begin{theorem}\label{tCrwlVsPCRWL}
For any \crwl-program $\prog$, $e \in \Exp_\perp$
$$
\dens{e} \subseteq \denbp{e} \subseteq \denap{e}
$$
None of the converse inclusions holds in general.
\end{theorem}
\begin{proof}
Given a \crwl-proof for $\vds e \clto t$ we can build a \apcrwl-proof for $\vdap e \cltop t$ just replacing every \crule{OR} step by the corresponding \crule{\bpor} step, as it is easy to see that any singleton set of c-substitutions is compressible, and that $?\{\theta\} = \theta$. As a consequence $\dens{e} \subseteq \denbp{e}$. On the other hand we can turn any \bpcrwl-proof into a \apcrwl-proof just replacing any \crule{\bpor} step by the corresponding \crule{\apor}, as \crule{\bpor} has stronger premises than \crule{\apor}, and the same consequence. Therefore $\denbp{e} \subseteq \denap{e}$.

Regarding the failure of the converse inclusions in the general case, consider the program $\{pair(X) \tor d(X, X), g(d(X, Y)) \tor d(X, Y)\}$ for which it is easy to check that $\dens{pair(0?1)} \not\ni d(0,1) \in \denbp{pair(0?1)}$ and $\denbp{g(d(0, 0) ? d(1,1))} \not\ni d(0,1) \in \denap{g(d(0, 0) ? d(1,1))}$.
\end{proof}

Concerning the relation between call-time choice and run-time choice, it was already explored in previous 
works of the authors \cite{lrs07,LMRS10LetRwTPLP}, and we recast it here in the following theorem.
\begin{theorem}\label{TCompRwvsCrwl}
For any \crwl-program $\prog$, $e \in \Exp$, $\dens{e} \subseteq \denr{e}$. The converse inclusion does not hold in general (as shown by Example \ref{EIntro1}).
\end{theorem}


On the other hand we cannot rely on any precedent in order to study the relation between $\bpcrwl$ and run-time choice. Therefore, putting run-time choice in the right place in the semantics inclusion chain from Theorem \ref{tCrwlVsPCRWL} will be one of the contributions of this work. We anticipate that the conclusion is that $\bpcrwl$ computes more values in general. 
\begin{theorem}\label{TCompPlvsRw}
For any \crwl-program $\prog$, $e \in \Exp$, $\denr{e} \subseteq \denbp{e}$. The converse inclusion does not hold in general.
\end{theorem}

It is easy to prove the last statement of Theorem \ref{TCompPlvsRw}, as in fact Example \ref{EIntro2} is a valid counterexample for that, but proving the first part is far more complicated. The key for this proof is the following lemma stating that every term rewriting step is sound w.r.t.\ \bpcrwl.
\begin{lemma}[One step soundness of $\rw$ w.r.t.\  \bpcrwl]\label{LCompPlvsRw1}
For any {\it CRWL}-program $\prog$, $e, e' \in \Exp$ if $e \rw e'$ then $\denbp{e'} \subseteq \denbp{e}$.
\end{lemma}

Note that any term rewriting step is of the shape $\con[f(\overline{p})\sigma] \tor \con[r\sigma]$ for some $\sigma \in \Subst$ and some program rule $f(\overline{p}) \tor r$. If we could prove Lemma \ref{LCompPlvsRw1} for any step performed at the root of the starting expression, i.e. that $f(\overline{p})\sigma \tor r\sigma$ implies $\denbp{r\sigma} \subseteq \denbp{f(\overline{p})\sigma}$, then we could use the compositionality of \bpcrwl\ from Theorem \ref{ThPropsBetaPlHeredadas} to propagate the result $\denbp{r\sigma} \subseteq \denbp{f(\overline{p})\sigma}$ to $\denbp{\con[r\sigma]} \subseteq \denbp{\con[f(\overline{p})\sigma]}$. 
To do that we will use the following notion of \bpcrwl-denotation of a substitution.
\begin{definition}[Denotation of substitutions]
For any \crwl-program $\prog$, $\sigma \in \Subst_\perp$ the \bpcrwl-denotation of $\sigma$ under $\prog$ is 
$$
\denbpp{\sigma}{\prog} = \{\theta \in \CSubst_\perp \mid \forall X \in \var,~ \prog \vdbp \sigma(X) \cltop \theta(X)\}
$$
\end{definition}

Denotations of substitutions enjoy several interesting properties. For example every $\sigma \in Subst_\perp$ is more powerful than any combination of substitutions from its denotation by means of the $?$ operator, in the sense that $\sigma$ is bigger than the combination w.r.t.\ the preorder $\bdsord$---which implies that if we apply $\sigma$ to an arbitrary expression we get an expression with a bigger denotation that if we apply the combination, thanks to the monotonicity of $Subst_\perp$ enjoyed by $\bpcrwl$. 
This is something natural, because c-substitutions in $\denbp{\sigma}$ only contain a finite part of the possibly infinite set of values generated for each expression in the range of $\sigma$. 
\begin{lemma}\label{lDenSubst2}
For any finite not empty $\Theta \subseteq \denbp{\sigma}$ we have $?\Theta \bdsord \sigma$.
\end{lemma}

Besides, it is clear that in any \bpcrwl-proof that uses some $\sigma \in \Subst_\perp$ only a finite amount of the information contained in $\sigma$. Therefore, in $\denbp{\sigma}$ is employed, just like in any proof for a statement $\vdbp e \clto t$ only a finite amount of the information in $e$ is used.
This follows because $t$ is a finite element and the \bpcrwl-proof is also finite, otherwise the statement $\vdbp e \clto t$ could not have been proved. These intuitions are formalized in the following result. 

\begin{lemma}\label{lDenSubst1}
For any $\sigma \in \Subst_\perp, e \in \Exp_\perp, t \in \CTerm_\perp$ 
if $\vdbp e\sigma \cltop t$ then $\exists \Theta \subseteq \denbp{\sigma}$ finite and not empty such that $\vdbp e(?\Theta) \cltop t$
\end{lemma}
\begin{proof}[Proof (sketch)]
First we prove the case where $e \equiv X \in \var$. If $X \in \mi{dom}(\sigma)$ then we define some $\theta \in \CSubst_\perp$ as
$$
\theta(Y) = \left\{\begin{array}{ll}
t & \mbox{ if } Y \equiv X \\
\perp & \mbox{ if } Y \in (\dom(\sigma) \setminus \{X\})\\
Y & \mbox{ if } Y \not\in \dom(\sigma)
\end{array} \right.
$$
Otherwise if $X \not\in dom(\sigma)$ then given $\overline{Y} = \dom(\sigma)$ we define $\theta = [\overline{Y/\perp}]$. In both cases it is easy to see that taking $\Theta = \{\theta\}$ then the conditions of the lemma are granted. To prove the general case where $e$ is not restricted to be a variable we perform an easy induction over the structure of $e\sigma \clto t$, using the property that for any $\Theta, \Theta' \subseteq CSubst_\perp$, if $\Theta \subseteq \Theta'$ then $?\Theta \pordap ?\Theta'$, combined with the monotonicity under substitutions of \bpcrwl. See \refap{PROOFlDenSubst1} for details. 
\end{proof}

This result is very interesting because it expresses a particular property of our plural semantics, as it can be also proved true for the corresponding definition of \apcrwl-denotation of a substitution. 
The key in this result is that the substitution obtained for rebuilding the starting derivation is 
a substitution from $\icsus$, which are precisely the kind of substitutions used for parameter passing in our plural semantics. 
On the other hand this is not true for \crwl, and it is one of the reasons why in general call-time choice computes less values than run-time choice: just consider the derivation $\vds d(X,X)[X/0~?~1] \clto d(0,1)$ for which there is no substitution $\theta$ in $\CSubst_\perp$---the kind of substitutions used for parameter passing in \crwl---such that $\vds d(X,X)\theta \clto d(0,1)$. Nevertheless if we restrict to deterministic programs this property becomes true for \crwl---and besides in that case run-time choice and call-time choice are equivalent too, see \cite{lrs07,LMRS10LetRwTPLP} for details.

Although Lemma \ref{lDenSubst1} is a nice result we still need an extra ingredient to be able to use it for proving Lemma \ref{LCompPlvsRw1}, thus enabling an easy proof for Theorem \ref{TCompPlvsRw}. The point is that we cannot use an arbitrary substitution from $\icsus$ for parameter passing in \bpcrwl\ but only a substitution which would be also compressible, in order to ensure that no wrong substitution mixup is performed, which is precisely the main feature of \bpcrwl. Therefore although a version of Lemma \ref{lDenSubst1} for \apcrwl\ can be used for proving that term rewriting is sound w.r.t.\ \apcrwl---as in fact it was done in \cite{rodH08}---, for proving its soundness w.r.t.\ \bpcrwl\ we will still need to do a little extra effort. And the missing piece is the following notion of compressible completion of a set of c-substitutions, which adds some additional c-substitutions to its input set in order to ensure that the resulting set is then compressible. 
\begin{definition}[Compressible completion]
Given $\Theta \subseteq \CSubst_\perp$ finite such that $\{X_1, \ldots, X_n\} = \bigcup_{\theta \in \Theta}\dom(\theta)$, its compressible completion $\cc{\Theta}$ is defined as
$$
\cc{\Theta} = \{[X_1/X_1\theta_1, \ldots, X_n/X_n\theta_n] \mid \theta_1, \ldots, \theta_n \in \Theta\}
$$
\end{definition}

Every compressible completion enjoys the following basic properties, which 
explain why we call it ``completion'' and also ``compressible.''
\begin{proposition}[Properties of $\cc{\Theta}$]\label{LemPropsCC1} For any $\Theta \subseteq CSubst_\perp$ finite such that $\{X_1, \ldots, X_n\} = \bigcup_{\theta \in \Theta}\dom(\theta)$
\begin{enumerate}
	\item[a)] $\cc{\Theta} \subseteq \CSubst_\perp$ and it is finite.
	\item[b)] $\Theta \subseteq \cc{\Theta}$. As a result, $?\Theta \pordap ?\cc{\Theta}$.
	\item[c)] $\bigcup_{\mu \in \cc{\Theta}}\dom(\mu) = \{X_1, \ldots, X_n\}$.
	\item[d)] $\cc{\Theta}$ is compressible.
\end{enumerate}
\end{proposition}

But, for the current task, the most interesting property of compressible completions is the following.
\begin{lemma}\label{LemPropsCC2} For any $\sigma \in \Subst_\perp$ and any $\Theta \subseteq \denbp{\sigma}$ finite and not empty we have that $\cc{\Theta} \subseteq \denbp{\sigma}$ too.
\end{lemma}

This is precisely the result we need to strengthen Lemma \ref{lDenSubst1} so it now becomes applicable for \bpcrwl, as it allows us to shift from any subset of the \bpcrwl-denotation of a substitution to its compressible completion, which will be also more powerful than the starting subset thanks to Proposition \ref{LemPropsCC1} b). 
\begin{lemma}\label{lDenSubstBPl1}
For any $\sigma \in \Subst_\perp, e \in \Exp_\perp, t \in \CTerm_\perp$ if $\vdbp e\sigma \cltop t$ then $\exists \Theta \subseteq \denbp{\sigma}$ finite, not empty, and compressible such that $\vdbp e(?\Theta) \cltop t$. 
\end{lemma}
\begin{proof}
By Lemma \ref{lDenSubst1} we get some $\Theta \subseteq \denbp{\sigma}$ finite and not empty such that $\vdbp e(?\Theta) \cltop t$. Then by Lemma \ref{LemPropsCC2} we get that $\cc{\Theta} \subseteq \denbp{\sigma}$ too, and that it is finite, not empty (as $\Theta \subseteq \cc{\Theta}$ and $\Theta$ is not empty), compressible and $?\Theta \pordap ?\cc{\Theta}$ by Proposition \ref{LemPropsCC1}. But then we can apply the monotonicity of Theorem \ref{ThPropsBetaPlHeredadas} to get $\vdbp e(?\cc{\Theta}) \cltop t$, so we are done.
\end{proof}

We can now use this result to prove a particularization of Lemma \ref{LCompPlvsRw1} (one step soundness of $\rw$ w.r.t.\ \bpcrwl) for steps performed at the root of the expression, i.e., of the shape $f(\overline{p})\sigma \tor r\sigma$. Thus, given some $t \in \denbp{r\sigma}$ our goal is proving that $t \in \denbp{f(\overline{p})\sigma}$.  First of all by Lemma \ref{lDenSubstBPl1} we get some compressible $\Theta \subseteq \denbp{\sigma}$ such that $t \in \denbp{r(?\Theta)}$. If we could use it to prove that $t \in \denbp{f(\overline{p})(?\Theta)}$ then  by Lemma \ref{lDenSubst2} we would get $?\Theta \bdsord \sigma $, so by the monotonicity of Theorem \ref{ThPropsBetaPlHeredadas} we could obtain $t \in \denbp{f(\overline{p})\sigma}$ as we wanted. 
As $\overline{p} \subseteq \CTerm_\perp$ and $\Theta \subseteq \CSubst_\perp$ we can easily prove that $\forall p_i \in \overline{p}, \theta_j \in \Theta$ we have $\vdbp p_i(?\Theta) \clto p_i\theta_j$. All this can be used to perform the following step, assuming $\Theta = \{\theta_1, \ldots, \theta_m\}$.
$$
\begin{scriptsize}
\infer[\mbox{\crule{\bpor}}]{f(p_1, \ldots, p_n)(?\Theta) \clto t}
           {
	\!\!\begin{array}{c}
        p_1(?\Theta) \cltop p_1\theta_1 \equiv p_1\theta_1|_{var(p_1)} \\
	      \ \ldots \\
		p_1(?\Theta) \cltop p_1\theta_m \equiv p_1\theta_m|_{var(p_1)} \\
	\end{array}~
       \ldots
       \ \!\!\begin{array}{c}
		p_n(?\Theta) \cltop p_n\theta_1 \equiv p_n\theta_1|_{var(p_n)} \\
	      \ \ldots \\
		p_n(?\Theta) \cltop p_n\theta_m \equiv p_n\theta_m|_{var(p_n)} \\
	\end{array}
       \ \!\!\begin{array}{c}
         ~ \\
         ~ \\
         r\theta' \equiv r(?\Theta) \cltop t
         \end{array}
            }
\end{scriptsize}
$$
for $\theta' = (\biguplus ?\Theta_i) \uplus \theta_e$ where $\forall i \in \{1, \ldots, n\}. \Theta_i = \{\theta_j|_{\mivar(p_i)}~|~\theta_j \in \Theta\}$, $\theta_e = (?\Theta)|_\vE$ for $\vE = \vextra{f(\overline{p}) \tor r}$. It can be easily proved that having $\Theta$ compressible implies that each $\Theta_i$ is also compressible---so the \bpor\ step above is valid---, and that $r\theta' \equiv r(?\Theta)$.

Therefore we have just proved the soundness w.r.t.\ \bpcrwl\ of term rewriting steps performed at the root of the starting expression. So all that is left is using the compositionality of \bpcrwl\ from Theorem \ref{ThPropsBetaPlHeredadas} for propagating this result for steps performed in an arbitrary context. A fully detailed proof for Lemma \ref{LCompPlvsRw1} can be found in \refap{PROOFLCompPlvsRw1}.

~\\
And now we are finally ready to prove Theorem \ref{TCompPlvsRw}.
\begin{proof}[Proof for Theorem \ref{TCompPlvsRw}]
Given some $t \in \denr{e}$, by definition $\exists e' \in \Exp$ such that $t \ordap |e'|$ and $e \rw^* e'$. We can extend Lemma \ref{LCompPlvsRw1} to $\rw^*$ by a simple induction on the length of $e \rw^* e'$, hence $\denbp{e'} \subseteq \denbp{e}$. As $\forall e \in Exp_\perp, |e| \in \denbp{e}$ (by a simple induction on the structure of $e$), then $t \ordap |e'| \in \denbp{e'} \subseteq \denbp{e}$, hence $t \in \denbp{e}$ by the polarity of Theorem \ref{ThPropsBetaPlHeredadas}. Example 
\ref{EIntro1} shows that the converse inclusion does not hold in general.\\
\end{proof}

The evident corollary for all these results is the 
following inclusion chain.
\begin{corollary}\label{CHierChain}
For any \crwl-program $\prog$, $e \in \mi{Exp}$ 
$$
\dens{e} \subseteq \denr{e} \subseteq \denbp{e} \subseteq \denap{e}
$$
Hence for any $t \in \CTerm$, $\prog \vds e \clto t$ implies $\prog \vdash e \rw^* t$, which implies $\prog \vdbp e \cltop t$, which implies $\prog \vdap e \cltop t$.
\end{corollary}
\begin{proof}
The first part holds just combining Theorems \ref{tCrwlVsPCRWL}, \ref{TCompRwvsCrwl}, and \ref{TCompPlvsRw}.\\
Concerning the second part, assume $\vdash_{\mi{CRWL}} e \clto t$, in other words, $t \in \dens{e}$. Then by the first part $t \in \denr{e}$, hence $e \rw^* e'$ such that $t \ordap |e'|$. But as $t \in \CTerm$ it is total and then $t$ is maximal w.r.t.\ $\ordap$ (a know property of $\ordap$ easy to check by induction on the structure of expressions), and so $t \ordap |e'|$ implies $t \equiv |e'|$, which implies $t \equiv e'$, as $t$ is total (easy to check by induction on the structure of $t$). Therefore $e \rw^* e' \equiv t \in \CTerm$, which implies $t \in \denr{e}$ by definition, as for c-terms $t$ we have $t \ordap t \equiv |t|$ (a property of shells proved by induction on the structure of $t$), but then $t \in \denbp{e} \subseteq \denap{e}$ by the first part, and so both $\vdbp e \cltop t$ and $\vdap e \cltop t$.
\end{proof}

\subsection{Restricted equivalence of \apcrwl\ and \bpcrwl}\label{sect:equivaABPlural}
In this section we will present a class of programs for which \apcrwl\ and \bpcrwl\ behave the same, thus yielding exactly the same denotation for any expression. In the previous section we saw that $\denbp{e} \subseteq \denap{e}$ for any expression and program, therefore we just have to find a class of programs such that $\denap{e} \subseteq \denbp{e}$ also holds for programs in that class.

The intuitions and ideas behind the characterization of that class of programs come from Example \refp{EPlural3}. The program used there contains two functions $f$ and $h$ defined by the rules $\{f(c(X)) \tor d(X,X), h(d(X, Y)) \tor d(X, X)\}$, with $d \in \CS^2$, under which it is easy to check that \apcrwl\ and \bpcrwl\ behave the same for the expressions $f(c(0)~?~c(1))$ and $h(d(0,0)~?~d(1,1))$.
\begin{itemize}
 	\item Regarding $f(c(0)~?~c(1))$, it is pretty natural for both plural semantics to behave the same, as no wrong information mixup can be performed when combining two substitutions with singleton domain, like $[X/0]$ and $[X/1]$, coming when evaluating $c(0)~?~c(1)$ to get an instance of $c(X)$. 
    \item The case for $h(d(0,0)~?~d(1,1))$ is more surprising at a first look, because then we can obtain the matching substitutions $[X/0, Y/0]$ and $[X/1, Y/1]$, which cannot be safely combined because the set $\{[X/0, Y/0], [X/1, Y/1]\}$ is not compressible. But, as seen in Example \ref{EPlural3}, this poses no problem, because the wrongly intermingled substitution $[X/0~?~1, Y/0~?~1]$ used by \apcrwl\ has the same effect over the right-hand side $d(X,X)$ of the rule for $h$ as the substitution $[X/0~?~1, Y/\perp]$, that can be obtained from combining the compressible set $\{[X/0, Y/\perp], [X/1, Y/\perp]\}$. This compressible set not only can be used for parameter passing by \bpcrwl, but also can be generated by evaluating the arguments of $h(d(0,0)~?~d(1,1))$ to get an instance of the left-hand side of the rule for $h$, as $[X/0, Y/\perp] \ordap [X/0, Y/0]$ and $[X/1, Y/\perp] \ordap [X/1, Y/1]$. 
 \end{itemize} 
What the functions $f$ and $h$ have in common is that, for each argument of the left-hand side of each of their program rules, at most one variable in that argument appears also in the right-hand side. If we only have to care about one variable then we can lower to $\perp$ the value obtained for the other variables in the matching substitution, thus getting a smaller---w.r.t.\ to $\ordap$---matching substitution corresponding to a smaller value, that then can be computed thanks to the polarity of \apcrwl\ from Proposition \ref{LMonPlural}. The effect of this is that we would get a compressible substitution that can be used 
by \apcrwl\ to turn a \apor\ step using a possibly non-compressible substitutions into a \apor\ step using a compressible substitutions, that would be then a valid \bpor\ step as well. 
Note that in this case extra variables pose no problem, as the only difference between \apor\ and \bpor\ is the way they handle the matching substitutions obtained by the evaluation of function arguments. Then, as extra variables are instantiated freely and independently of the matching substituions, they always behave the same both under \apcrwl\ and \bpcrwl. 

In the following definition we formally define the class $\classAB$ of programs in which the ideas above are materialized.

\begin{definition}[Class of programs $\classAB$]\label{def:classAB}
The class of programs $\classAB$ is defined by
$$
\prog \in \classAB \mbox{ iff } \forall (f(p_1, \ldots, p_n) \tor r) \in \prog. \forall i \in \{1, \ldots, n\}. \#(\mivar(p_i) \cap \mivar(r)) \leq 1
$$
\end{definition}

\noindent where, given a set $S$, $\#(S)$ stands for the cardinality of $S$.
Note that any program rule in which every argument in its left-hand side is ground or a variable passes the test that characterizes $\classAB$: for ground arguments no parameter passing is performed, only matching, so we conjecture that if the arguments in the left-hand side of each program rule 
are ground then both \crwl, term rewriting, \apcrwl, and \bpcrwl\ behave the same; on the other hand, for variable arguments we have the converse situation so matching is trivial and parameter passing is the important thing, so we conjecture that if the arguments in the left-hand side of each program rule are variables then both term rewriting, \apcrwl, and \bpcrwl\ behave the same---\crwl\ remains as the smaller semantics in this case, just consider the program $\{pair(X) \tor d(X, X)\}$ and the expression $pair(0~?~1)$ for which $d(0,1)$ cannot be computed by \crwl\ but it can be by any of the other three semantics.

Anyway, the class $\classAB$ is defined by a simple syntactic criterion, which can be easily implemented in any mechanized program analysis tool, 
and that we have implemented in our prototype from Section \ref{sec:examples}. 

~\\
The following theorem formalizes the expected equivalence between \apcrwl\ and \bpcrwl\ for programs in the class $\classAB$.

\begin{theorem}[Equivalence of \apcrwl\ and \bpcrwl\ for the class $\classAB$]\label{ThEquivABPClass}
For any program $\prog \in \classAB$, $e \in \Exp_\perp$
$$
\denapp{e}{\prog} = \denbpp{e}{\prog}
$$
\end{theorem}

This equivalence between \apcrwl\ and \bpcrwl\ 
will be very useful for us for several reasons. First of all, as we will see in Section \ref{sec:simulAPlural}, \apcrwl\ can be simulated by term rewriting through a simple program transformation, which implies that the same transformation can be used to simulate \bpcrwl\ for the class of programs $\classAB$, thanks to the equivalence from Theorem \ref{ThEquivABPClass}. On the other hand the class $\classAB$ is defined by a simple syntactic criterion, which allows its application to mechanized program analysis. 
Finally, this equivalence 
grows in importance after realising that 
the class $\classAB$ contains many relevant programs: as a matter of fact all the programs considered in Section \ref{sec:examples}---where we explore the expressive capabilities of our plural semantics---belong to the class $\classAB$.

\subsection{Simulating plural semantics with term rewriting}\label{sec:simulAPlural} 


In \cite{lrs07,lsr09,LMRS10LetRwTPLP} it was shown that neither \crwl\ can be simulated by term rewriting with a simple program transformation, nor vice versa. Nevertheless, $\apcrwl$ can be simulated by term rewriting using the transformation presented in the current section, which 
can then be used as the basis for a first implementation of \apcrwl. 
First we will present a naive version of this transformation, and show its adequacy; later we will propose some simple optimizations for it.
%

\emph{In this section we will restrict ourselves to programs not containing extra variables}, i.e., 
such that for any 
program rule $l \tor r$ 
we have that $\mivar(r) \subseteq \mivar(l)$ holds, a restriction usually adopted in texts devoted to term rewriting systems \cite{baader-nipkow,Terese03} for which term rewriting with extra variables is normally considered as an extension of standard term rewriting. Besides, in practical implementations extra variables are usually handled by using narrowing \cite{LS99,Han06curry} or additional conditions to restrict their possible instantiations \cite{maude-book}, in order to avoid 
a state space explosion in the search process. Therefore we leave the extension of our work to completely deal with extra variables as a subject of future work.

\subsubsection{A simple transformation}\label{subsect:simpletrans}
The main idea in our transformation is to postpone the pattern matching process in order to prevent an early resolution of non-determinism. Instead of presenting the transformation directly, we will first illustrate this concept by applying the transformation over the program $\prog = \{f(c(X)) \tor d(X, X)\}$ from Example \ref{EIntro1}, which results in the following program $\progh$. 
\begin{center}
\scalebox{0.91}{
$
\begin{array}{ll}
\progh = \{ & f(Y) \tor \mi{if}~\match(Y)~\thn~d(\prj(Y), \prj(Y)),\\
& \match(c(X)) \tor \tru, \prj(c(X)) \tor X ~~\} \\ 
\end{array}
$
}
\end{center}
In the resulting program $\progh$ the only rule for function $f$ has been transformed so matching is transferred from the left-hand side to the right-hand side of the rule, by means of the auxiliary functions $\match$ and $\prj$. As a consequence, when we evaluate by term rewriting under $\progh$ the function call to $f$, in the expression $f(c(0)~?~c(1))$ we are not forced anymore to solve the non-deterministic choice between $c(0)$ and $c(1)$ before parameter passing, because any expression matches the variable pattern $Y$. Therefore the term rewriting step
\begin{center}
\scalebox{0.91}{
$
f(c(0)~?~c(1)) \rw \mi{if}~\match(c(0)~?~c(1))~\thn~d(\prj(c(0)~?~c(1)), \prj(c(0)~?~c(1)))
$
}
\end{center}
is sound, thus replicating the argument of $f$ freely without demanding any evaluation, this way keeping its \apcrwl-denotation untouched: this is the key to achieve completeness w.r.t.\ \apcrwl. Note that the guard $\mi{if}~\match(c(0)~?~c(1))$ is needed to ensure that at least one of the values of the argument matches the original pattern, otherwise the soundness of the step could not be granted. For example if we drop this condition in the translation of the rule `$\mi{null}(\mi{nil}) \tor \tru$' for defining an emptiness test for the classical representation of lists in functional programming, we would get `$\mi{null}(Y) \tor \tru$', which is clearly unsound because it allows us to rewrite $null(cons(0, nil))$ into $\tru$. 
Later on, after resolving the guard, different evaluations of the occurrences of $\prj(c(0)~?~c(1))$ will solve the non-deterministic choice implied by $?$, and project the argument of $c$, thus leading us to the final values $d(0,0)$, $d(1,1)$, $d(0,1)$, and $d(1,0)$, which are the expected values for the expression in the original program under $\apcrwl$. 

In the following definition we formalize the transformation by means of the function $pST$, which for any program rule returns a rule to replace it, and a set of auxiliary $\match$ and $\prj$ rules for the replacement.  
\begin{definition}[\apcrwl\ to term rewriting transformation, simple version]
\label{DefPstNaive}
Given a program $\prog$, our transformation proceeds rule by rule. For every program rule $(f(p_1, \ldots, p_n) \tor r) \in \prog$ such that $f \not\in \{?, \mi{if}\,\thn\,\}$ we define its transformation as: \\[.2cm]
$
\begin{array}{l}
\pst{f(p_1, \ldots, p_n) \tor r} \\
 ~~~~~
= f(Y_1, \ldots, Y_n) \tor \mi{if}~\match(Y_1, \ldots, Y_n)~\thn~r[\overline{X_{i j}/\prj_{i j}(Y_i)}]
\end{array}
$\\[.2cm]
where\\
- $\forall i \in \{1, \ldots, n\}$, $\{X_{i1}, \ldots, X_{ik_i}\} = \mivar(p_i) \cap \mivar(r)$ and $Y_i \in \var$ is fresh.\\
- $\match \in \FS^{n}$ is a fresh function defined by the rule $\match(p_1, \ldots, p_n) \tor \tru$. \\
- Each $\prj_{ij} \in \FS^1$ 
is a fresh symbol defined by the single rule $\prj_{ij}(p_i) \tor X_{ij}$.\\[.2cm]
For $f \in \{?, \mi{if}\,\thn\,\}$ the transformation leaves its rules untouched.
\end{definition}

It is easy to check that if we use the program $\prog$ from Example \ref{EIntro1} as input for this transformation then it outputs the program $\progh$ from the discussion above, under which we can perform the following term rewriting derivation. 

\begin{small}
$$
\begin{array}{l}
\underline{f(c(0)?c(1))} 
\rw \mi{if}~\underline{\match(c(0)?c(1))}~\thn~d(\prj(c(0)?c(1)), \prj(c(0)?c(1))) \\
\rw^* \underline{\mi{if}~\tru~\thn~d(\prj(c(0)?c(1)), \prj(c(0)?c(1)))}\\
\rw d(\prj(\underline{c(0)?c(1)}), \prj(\underline{c(0)?c(1)}))
\rw^* d(\underline{\prj(c(0))}, \underline{\prj(c(1))}) 
\rw^* d(0, 1)
\end{array}
$$
\end{small}

We do not only claim that this transformation is sound, but also have technical results about the strong adequacy of our transformation $\pst{\_}$ for simulating the \apcrwl\ logic using term rewriting. The first one is a soundness result, stating that if we rewrite an expression under the transformed program then we cannot get more results that those we can get in \apcrwl\ under the original program. 
\begin{theorem}\label{TPstSound}
For any \crwl-program $\prog$, and any $e \in \Exp_\perp$ built up 
on the signature of $\prog$, we have 
$$\denapp{e}{\pst{\prog}} \subseteq \denapp{e}{\prog}$$
As a consequence $\denrp{e}{\pst{\prog}} \subseteq \denapp{e}{\prog}$.
\end{theorem}
\begin{proof}[Proof (sketch)]
The first part states the soundness within \apcrwl\ of the transformation. Assuming a \apcrwl-proof for a statement $\pst{\prog} \vdap e \clto t$ for some $t \in \CTerm_\perp$, we can then build another \apcrwl-proof for $\prog \vdap e \clto t$, by induction on the size of the starting proof---measured as the number of rules of \apcrwl\ used. Full details for that proof can be found in \refap{PROOFTPstSound}. 

Concerning the second part, it follows from combining the first part with Corollary \ref{CHierChain}, because then we can chain $\denrp{e}{\pst{\prog}} \subseteq \denapp{e}{\pst{\prog}} \subseteq \denapp{e}{\prog}$.
\end{proof}

~\\
Regarding completeness of the transformation we have obtained the following result stating that, for any expression one can build in the original program, we can refine by term rewriting under the transformed program any value computed for that expression by \apcrwl\ under the original program. 
\begin{theorem}\label{lCompPstAlt}
For any \crwl-program $\prog$, and any $e \in \Exp, t \in \CTerm_\perp$ built up 
on the signature of $\prog$, if $\prog \pps e \cltop t$ then exists some $e' \in \Exp$ built using symbols of the signature of $\pst{\prog}$ such that $\pst{\prog} \vdash e\rw^* e'$ and $t \ordap |e'|$. 
In other words, $\denapp{e}{\prog} \subseteq \denrp{e}{\pst{\prog}}$. 
\end{theorem}

The proof for this result is technically very involved. First of all we have to slightly generalize Theorem \ref{lCompPstAlt} to consider not only the functions of the original program but also the auxiliary $\match$ and $\prj$ functions generated by the transformation, in order to obtain strong enough induction hypothesis.

\begin{lemma}\label{lAuxCompPstAlt2}
Given a \crwl-program $\prog$ let $\progh \uplus \progm = \pst{\prog}$, where $\progm$ is the set containing the rules for
the new functions $\match$ and $\prj$, and $\progh$ contains the new versions of the original rules of $\prog$---note that by an abuse of notation the rules for $?, \mi{if}\,\thn\,$ presented in Section \ref{sec:cterms_intro} belong implictly both to $\prog \uplus \progm $ and $\progh \uplus \progm$. 

Then for any $e \in \Exp_\perp, t \in \CTerm_\perp$ constructed using just symbols in the signature of $\prog \uplus \progm$ we have $\prog \uplus \progm \pps e \cltop t$ implies $\progh \uplus \progm \vdash e \rw^* e'$ such that $t \ordap |e'|$.
\end{lemma}

The proof for Lemma \ref{lAuxCompPstAlt2}  is pretty complicated and it relies on several auxiliary notions, a fully detailed proof can be found in \proofs. 
Then Theorem \ref{lCompPstAlt} follows as an almost trivial consequence of Lemma \ref{lAuxCompPstAlt2}. 
\begin{proof}[Proof for Theorem \ref{lCompPstAlt}]
Let $\progh \uplus \progm = \pst{\prog}$ be, where $\progm$ is the set containing the rules for the new functions $\match$ and $\prj$, and $\progh$ contains the new versions of the original rules of $\prog$. 

If $e \in \Exp, t \in \CTerm_\perp$ are built using symbols on the signature of $\prog$, then $\prog \pps e \cltop t$ implies $\prog \uplus \progm \pps e \cltop t$, which implies $\progh \uplus \progm \vdash e \rw^* e'$ such that $t \ordap |e'|$ by Lemma \ref{lAuxCompPstAlt2}, that is, $\pst{\prog} \vdash e \rw^* e'$.
\end{proof}

To conclude, the following corollary summarizes the adequacy of the simulation performed by our program transformation. 
\begin{corollary}[Adequacy of $\pst{\_}$ for simulating \apcrwl]\label{CorPstAdeq}
For any program $\prog$, $e \in \Exp$ built using symbols of the signature of $\prog$ 
$$
\denapp{e}{\prog} = \denrp{e}{\pst{\prog}}
$$
Hence $\forall t \in \CTerm$ we have that $\prog \pps e \cltop t$ iff $\pst{\prog} \vdash e \rw^* t$.
\end{corollary}
\begin{proof}
The fist part holds by a combination of Theorem \ref{TPstSound} and Theorem \ref{lCompPstAlt}. 

\noindent For the second part, if $\prog \pps e \cltop t$ then $t \in \denp{e}_{\prog} = \denrp{e}{\pst{\prog}}$ by the first part, hence $\exists e' \in \Exp$ such that $\pst{\prog} \vdash e \rw^* e'$ and $t \ordap |e'|$. But as $t \in \CTerm$ then $t$ is maximal w.r.t.\ $\ordap$ and so $t \equiv |e'|$ which implies $t \equiv e'$ (these are known properties of shells and $\ordap$), therefore $\pst{\prog} \vdash e \rw^* e' \equiv t$. On the other hand if $\pst{\prog} \vdash e \rw^* t$ then as $t \ordap t \equiv |t|$ (again because $t$ is a total c-term) we have $t \in \denrp{e}{\pst{\prog}} = \denapp{e}{\prog}$, and so $\prog \pps e \cltop t$.
\end{proof}

As promised at the end of the previous subsection, we can now use the restricted equivalence between \apcrwl\ and \bpcrwl\ from Theorem \ref{ThEquivABPClass} to extend the adequacy results of the simulation of \apcrwl\ with term rewriting to \bpcrwl, for the class of programs $\classAB$.

\begin{corollary}[Restricted adequacy of $\pst{\_}$ for simulating \bpcrwl]\label{CorPstAdeqBP}
For any program $\prog \in \classAB$, $e \in Exp$ built using symbols of the signature of $\prog$
$$
\denbpp{e}{\prog} = \denrp{e}{\pst{\prog}}
$$
Hence $\forall t \in \CTerm$ we have that $\vdbp e \cltop t$ iff $\pst{\prog} \vdash e \rw^* t$.
\end{corollary}
\begin{proof}
A straightforward combination of Corollary \ref{CorPstAdeq} and Theorem \ref{ThEquivABPClass}.
\end{proof}

This last result illustrates the interest of \apcrwl. Because of its simplicity, \apcrwl\ sometimes combines matching substitutions in a wrong way, but it is precisely that same simplicity which allows it to be simulated by term rewriting through a simple program transformation. As a result we can use any available implementation of term rewriting, like the Maude system, to devise an implementation of \apcrwl. Besides, thanks to the restricted equivalence  between \apcrwl\ and \bpcrwl, that would also be an implementation of \bpcrwl\ for the class $\classAB$, and the membership check of program to the class $\classAB$ could be also mechanized, because $\classAB$ is defined by a simple syntactic criterion. We will see how these ideas are developed in the next sections, where a Maude-based implementation of our plural semantics is presented, and the interest of the class $\classAB$ is illustrated.

\subsubsection{An optimized transformation\label{subsubsec:opt_trans}}

As we already mentioned in our comments after presenting the class $\classAB$ in Definition \refp{def:classAB}, we expect that for ground or variable arguments run-time choice and our plural semantics behave the same. We can take advantage of this for applying some optimizations to the program transformation from Definition \ref{DefPstNaive}. 
\begin{itemize}
	\item 
When applied to $\mi{null}(\mi{nil}) \tor \tru$, the transformation returns the rules $\{null(Y) \tor \mi{if}~\match(Y)$ $\thn~\tru,$ $\match(nil) \tor \tru\}$, which behave the same as the original rule. The conclusion is that when a given pattern is ground then no parameter passing will be done for that pattern, and thus no transformation is needed.
    \item  Something similar happens with $\mi{pair}(X) \tor d(X,X)$ for which $\{\mi{pair}(Y) \tor \mi{if}~\match(Y)~\thn~d(\prj(Y),$ $\prj(Y)), \match(X) \tor \tru, \prj(X) \tor X\}$ is returned. In this case the pattern is a variable to which any expression matches without any evaluation, and the projection functions are trivial, so no transformation is needed neither.
\end{itemize}

We can apply these ideas to get the following refinement of our original program transformation. 
\begin{definition}[\apcrwl\ to term rewriting transformation, optimized version]\label{DefPstOpt}
Given a program $\prog$, our transformation proceeds rule by rule. For every program rule $(f(p_1, \ldots, p_n) \tor r) \in \prog$ we define its transformation as: \\[.2cm]
$
\begin{array}{l}
\pst{f(p_1, \ldots, p_n) \tor r}\\[.2cm]
~= \left\{\begin{array}{l@{\hspace{-0.1cm}}l}
f(p_1, \ldots, p_n) \tor r & \mbox{ if } m = 0\\
f(\tau(p_1), \ldots, \tau(p_n)) \tor \begin{array}{l}
\mi{if}~\match(Y_1, \ldots, Y_m)\\
 ~~~
~\thn~r[\overline{X_{i j}/\prj_{i j}(Y_i)}]\end{array} & \mbox{ otherwise}
\end{array} \right.
\end{array}
$\\[.2cm]
where $\rho_1 \ldots \rho_m = p_1 \ldots p_n~|~\lambda p.(p\not\in \var \wedge \mivar(p) \not= \emptyset)$. \\ 
- $\forall \rho_i$, $\{X_{i1}, \ldots, X_{ik_i}\} = \mivar(\rho_i) \cap \mivar(r)$ and $Y_i \in \var$ is fresh.\\
- $\tau : \CTerm \rightarrow \CTerm$ is defined by 
$\tau(p) = p$ if $p \not\in \{\rho_1, \ldots, \rho_m\}$; otherwise $\tau(\rho_i) = Y_i$.\\
- $\match \in \FS^{m}$ fresh is defined by the rule $\match(\rho_1, \ldots, \rho_m) \tor \tru$. \\
- Each $\prj_{ij} \in \FS^1$ 
is a fresh symbol defined by the rule $\prj_{ij}(\rho_i) \tor X_{ij}$.
\end{definition} 

Note that this transformation is well defined because each $\rho_i \in \{\rho_1, \ldots, \rho_m\}$ contains at least one variable, and so it can be distinguished from any other $\rho_j$ by using syntactic equality thanks to left linearity of program rules, therefore $\tau$ is well defined. 

We will not give any formal proof for the adequacy of this optimized transformation. Nevertheless note how this transformation leaves untouched the rules for $?$ and $\mi{if}\,\thn\,$ without defining a special case for them. As the simple transformation from Definition \ref{DefPstNaive} worked well for these rules, that suggests that we are doing the right thing. \\

We end this section with an example application of the optimized transformation, over the program from Example~\ref{EPlural2}. As expected the transformed program behaves under term rewriting like the original one under \apcrwl. 
\begin{example}\label{ETrans2}
The only rule modified is the one for $\find$, for which we get the following program
$$
\begin{array}{l}
\{\find(Y) \tor \mi{if}~\match(Y)~\thn~(\prj(Y),\prj(Y)), \\
\,\;\match(e(N,G,\clerk)) \tor \tru, \\
\,\;\prj(e(N,G,\clerk)) \tor N\}
\end{array}
$$
under which we can perform this term rewriting derivation for $twoclerks$
$$
\begin{array}{l}
\underline{\twoclerks} \rw \underline{\find(\employees(\branches))} \\
\rw \mathit{if}~\match(\underline{\employees(\branches)})\\
 ~~~~~~~~~
~\thn~(\prj(\employees(\branches)), \prj(\employees(\branches))) \\
\rw^* \underline{\mathit{if}~\match(e(\pepe, \man, \clerk))}\\
 ~~~~~~~~~
~\underline{\thn~(\prj(\employees(\branches)), \prj(\employees(\branches)))} \\
\rw^* (\prj(\underline{\employees(\branches)}), \prj(\underline{\employees(\branches)})) \\
\rw^* (\underline{\prj(e(\pepe, \man, \clerk))}, \underline{\prj(e(\maria, \woman, \clerk)}) \\
\rw^* (\pepe, \maria)
\end{array}
$$
\end{example}


\section{Programming with singular and plural functions}\label{sec:examples}

So far we have presented two novel proposals for the semantics of lazy non-determi-nistic functions, studied some of its properties, and explored their relation to previous proposals like call-time choice and run-time choice. Nevertheless, we have seen just a couple of program examples using the semantics, so until now we have hardly tested the way we can exploit the new expressive capabilities offered by our plural semantics to improve the declarative flavour of programs. 
The present section is devoted to the exploration of those expressive capabilities by means of several programs that try to illustrate the virtues of our new plural semantics.

In \cite{RRLDTA09} the authors already explored the capabilities of \apcrwl\ by using the Maude system \cite{maude-book} to develop an interpreter for this semantics based on the program transformation from Section~\ref{sec:simulAPlural}. The resulting interpreter was then used for experimenting with \apcrwl, showing how it allows an elegant encoding of some problems, in particular those with an implicit manipulation of sets of values. 
However, call-time choice still remains the best option for many common programming patterns \cite{GHLR99,DBLP:conf/flops/AntoyH02}, and this is the reason why it is the semantic option adopted by modern functional-logic programming systems like Toy \cite{LS99} or Curry \cite{Han06curry}. Therefore it would be nice to have a language in which both options could be available. 
In this section 
we 
propose such a language, where the user has the possibility to specify which arguments of each function symbol will be considered ``plural arguments.'' These arguments will be evaluated using our plural semantics, which intuitively means that they will be treated like sets of elements of the corresponding type\footnote{As types are not considered through this work here we mean the type naturally intended by the programmer.} instead of single elements, while the others will be evaluated under the usual singular/call-time choice semantics traditionally adopted for FLP. 
Thereby 
in \cite{rrpepm10} we 
extended our Maude-based prototype 
to support this combination of singular and plural arguments, and used it to develop and test several programs that we think are significant examples of the possibilities of the combined semantics. 
The source code for these examples and the interpreter to test them can be found at 
\url{http://gpd.sip.ucm.es/PluralSemantics}. 

~\\
As we have two different plural semantics available, then we get two different semantics resulting from their combination with call-time choice, that we have precisely formalized by means of two novel variants of \crwl\ called \sapcrwl\ and \sbpcrwl, corresponding to the combination of call-time choice with \apcrwl\ and \bpcrwl, respectively. Our prototype is based on the program transformation from Section~\ref{sec:simulAPlural}, therefore it is an implementation of \sapcrwl, and so \sbpcrwl\ is only supported for programs in the class $\classAB$ described in Section \ref{sect:equivaABPlural}. After those calculus, we introduce the concrete syntax of our interpreter and motivate the combination of singular and plural semantics with a simple example, while the next examples illustrate how to combine singular and plural arguments in depth. 
%
Then, after a short discussion about the use of singular and plural arguments, we conclude this section with a brief outline of the implementation of our prototype. 
%

\subsection{The logics \sapcrwl\ and \sbpcrwl\label{subsec:logics}}

We assume a mapping $\plrityZ : \FS \rightarrow \{\mi{sg}, \mi{pl}\}^*$ called \textit{plurality map} such that, for every $f\in \FS^n$, $\plrity{f} = b_1 \ldots b_n$ sets its plurality behaviour: if $b_i = \mi{sg}$ then the $i$-th argument of $f$ will be interpreted with a singular semantics, otherwise it will be interpreted under a plural semantics. In this line $\mi{sgArgs}(f) = \{i \in \{1, \ldots, \mi{ar}(f)\}~|~\plrity{f}[i] = \mi{sg}\}$ and $\mi{plArgs}(f) = \{i \in \{1, \ldots, \mi{ar}(f)\}~|$ $\plrity{f}[i] = \mi{pl}\}$ are the sets of singular and plural arguments of some $f \in \FS$. In particular we say that $f$ is a \emph{singular function} if $\mi{sgArgs}(f) = \{1, \ldots, \mi{ar}(f)\}$ and that it is a \emph{plural function} when $\mi{plArgs}(f) = \{1, \ldots, \mi{ar}(f)\}$. A related notion is that of singular and plural variables of a pattern: 
$\mi{sgVars}(f(\overline{p})) = \bigcup_{i \in \mi{sgArgs}(f)} \mi{var}(p_i)$ and $\mi{plVars}(f(\overline{p})) = \bigcup_{i \in \mi{plArgs}(f)} \mi{var}(p_i)$.

Thus we employ the plurality map to express which function arguments are considered singular arguments and which plural arguments. With this at hand
we now define the combined semantics \sapcrwl\ and \sbpcrwl\ as the result of taking the rules of \crwl\ and replacing the rule \textbf{OR} by either the rule \sapor\ or \sbpor\ from Figure \ref{fig:sabpcrwl}, respectively. 
\begin{figure}[t!]
\begin{center}
 \framebox{
\begin{minipage}{.95\textwidth}
\begin{center}
\begin{small}
\begin{tabular}{l}
\regla{\sapor}{f(e_1, \ldots, e_n) \clto t}{\begin{array}{c}
e_1 \cltop p_1\theta_{11} \\
 \ldots \\
e_1 \cltop p_1\theta_{1 m_1} \\
\end{array}
~\ldots~ \begin{array}{c}
		e_n \cltop p_n\theta_{n1} \\
	      \ \ldots \\
              \  e_n \cltop p_n\theta_{n m_n} \\
	 \end{array}~
\begin{array}{c}
          ~ \\
          ~ \\
          r\theta \cltop t
          \end{array}
}{} \\
\hspace{1cm} $\begin{array}{l}
\textrm{if } (f(\overline{p}) \tor r)\in {\cal P}\mbox{, } 
\forall i \in \{1, \ldots, n\}~\Theta_i = \{\theta_{i 1}, \ldots, \theta_{i m_i}\}\\
\theta = (\biguplus\limits_{i=1}^n ?\Theta_i) \uplus \thetae, 
\forall i \in \{1, \ldots, n\}, j \in \{1, \ldots, m_i\}~\dom(\theta_{i j}) \subseteq \mi{var}(p_i)\\
\dom(\thetae) \subseteq \vextra{f(\overline{p}} \tor r), \thetae \in \icsus \\
\forall i \in \{1, \ldots, n\}~m_{i} > 0, \forall i \in \mi{sgArgs}(f). m_i = 1
  \end{array}$ ~\\\\
\regla{\sbpor}{f(e_1, \ldots, e_n) \clto t}{\begin{array}{c}
e_1 \cltop p_1\theta_{11} \\
 \ldots \\
e_1 \cltop p_1\theta_{1 m_1} \\
\end{array}
\ldots~ \begin{array}{c}
		e_n \cltop p_n\theta_{n1} \\
	      \ \ldots \\
              \  e_n \cltop p_n\theta_{n m_n} \\
	 \end{array}~
\begin{array}{c}
          ~ \\
          ~ \\
          r\theta \cltop t
          \end{array}
}{} \\
\hspace{1cm} $\begin{array}{l}
\textrm{if } (f(\overline{p}) \tor r)\in {\cal P}\mbox{, } 
\forall i \in \{1, \ldots, n\}~\Theta_i = \{\theta_{i 1}, \ldots, \theta_{i m_i}\} \mbox{ is compressible }\\
\theta = (\biguplus\limits_{i=1}^n ?\Theta_i) \uplus \thetae, 
\forall i \in \{1, \ldots, n\}, j \in \{1, \ldots, m_i\}~dom(\theta_{i j}) \subseteq var(p_i)\\
dom(\thetae) \subseteq \vextra{f(\overline{p}} \tor r), \thetae \in \icsus \\
\forall i \in \{1, \ldots, n\}~m_{i} > 0, \forall i \in \mi{sgArgs}(f). m_i = 1
  \end{array}$  
 \end{tabular}
\end{small}
\end{center}
\end{minipage}
}
\end{center}
    \caption{The rules  \sapor\ and  \sbpor}
    \label{fig:sabpcrwl}
\end{figure}
As any variant of \crwl, these calculi derive {reduction statements} of the form $\prog \vdsap e \clto t$ and $\prog \vdsbp e \clto t$ that express that $t$ is (or approximates to) a possible value for $e$ in \sapcrwl\ or \sbpcrwl, respectively, under the program $\prog$. The denotations $\densapp{e}{\prog}$ and $\densbpp{e}{\prog}$ established by these semantics are defined as usual---see Definition \refp{def:denotations}. 

Just like in \apcrwl\ and \bpcrwl, we consider sets of partial values for parameter passing instead of single partial values, but the novelty is that now these sets are forced to be singleton for singular arguments. 
This is reflected in the new rules \sapor\ and \sbpor, 
corresponding to \apor\ and \bpor\ respectively, that now have been tuned to take account of the plurality map, as for singular arguments we are only allowed to compute a single value, thus performing parameter passing over it with a substitution from $CSubst_\perp$ (as obviously $?\{\theta\} = \theta$), and achieving a singular behaviour (call-time choice).

\begin{example}
Consider the program $\{f(X, c(Y)) \tor d(X, X,$ $Y, Y)\}$ and a plurality map such that $\plrity{f} = sg~pl$. 
The following 
is a \sapcrwl-proof for the statement $f(0~?~1, c(0)~?~c(1)) \clto d(0,0,0,1)$ (some steps have been omitted for the sake of conciseness). 
$$
{\small 
\infer[\mbox{\sapor}]{f(0~?~1, c(0)~?~c(1)) \clto d(0,0,0,1)}
{
~ \begin{array}{c}
\infer{c(0)~?~c(1) \clto c(0)}{(*)} \\
\infer{c(0)~?~c(1) \clto c(1)}{\vdots} \\
\infer{0~?~1 \clto 0}{\vdots}
\end{array}
~ \begin{array}{c}
~ \\ ~ \\~ \\[.25cm]
\infer[\crule{DC}]{~~d(0, 0, 0~?~1, 0~?~1) \clto d(0,0,0,1)~~}
       { 
       \infer{0 \clto 0}{\vdots} ~~ \infer{0 \clto 0}{\vdots} ~~ \infer{0~?~1 \clto 0}{\vdots} ~~ \infer{0~?~1 \clto 1}{\vdots}
       }
\end{array}
}
}
$$
where $(*)$ is the following proof:
$$
{\small
\infer[\mbox{\sapor}]{c(0)~?~c(1) \clto c(0)}
                  {
                  \infer[\crule{DC}]{c(0) \clto c(0)}{\infer[\crule{DC}]{0 \clto 0}{}}
                \ ~ \infer[\crule{B}]{c(1) \clto \perp}{}
                \ ~ \infer{c(0) \clto c(0)}{\vdots}
                  }
}
$$
Note that $d(0, 1, 0, 1)$ is not a correct value for the expression
$f(0~?~1, c(0)~?~c(1))$ under \spcrwl, because the first argument of $f$ is singular and therefore the two occurrences of $X$ in the right-hand side of its rule 
share the same single value, fixed on parameter passing. Besides, as this program is in the class $\classAB$, then it behaves the same under \apcrwl\ and \bpcrwl, and therefore also under \sapcrwl\ and \sbpcrwl, so the previous proof and comments also hold for \sbpcrwl. 

On the other hand if we take the same program and evaluate $f(0~?~1, c(0)~?~c(1))$ under term rewriting---which ignores the plurality map---, its behaviour is significantly different:
$$
\begin{array}{l}
f(0~?~1, \underline{c(0)~?~c(1)}) \rw \underline{f(0~?~1, c(0))} \rw d(\underline{0~?~1}, 0~?~1, 0, 0) \\
\rw d(0, \underline{0~?~1}, 0, 0) \rw d(0, 1, 0, 0)
\end{array}
$$
A first step resolving the choice between $c(0)$ and $c(1)$ is unavoidable in order to get an expression matching for the only rule for $f$, thus for any reachable c-term the last two arguments of $d$ will be the same, contrary to what happens in \sapcrwl\ and \sbpcrwl\ under the given plurality map. Nevertheless its first two arguments can be different, contrary to what happens under \spcrwl\ and \sbpcrwl. In conclusion, it is easy to define a program and a plurality map for them such that 
neither \sapcrwl\ nor \sbpcrwl\ are comparable  to term rewriting 
w.r.t.\ set inclusion of the computed 
values.
\end{example}

A useful intuition about programs comes from considering the singular arguments as fixed individual values, while thinking about the plural ones as sets. 
We could have chosen to specify the plurality or singularity of functions instead of that of its arguments, but the use of arguments with different plurality arises naturally in programs, in the same way it is natural to have arguments of different types. 
We will illustrate this fact later on by means of several examples. 

\bigskip
Regarding properties of these semantics (see \proofs\ for more details), both \sapcrwl\ and \sbpcrwl\ inherit the properties of \apcrwl\ from Section \ref{sec:aplural}, for the same reason \bpcrwl\ inherits the properties of \apcrwl. The most important among these properties is their compositionality, which expresses the value-based philosophy underlying \sapcrwl\ and \sbpcrwl\  ``all I know about an expression is its set of values,'' and that holds for the corresponding reformulation of Theorem \ref{lCompPS}---as it can be proved by a straightforward modification of the proof for that theorem. 
%
%
%
Bubbling is also incorrect for both \sapcrwl\ and \sbpcrwl, just like it happens for \apcrwl\ and \bpcrwl: in fact Example \ref{ex:noBubling} can be reused to prove it. Nevertheless, just like for \apcrwl\ and \bpcrwl, bubbling is correct for a particular kind of contexts, 
in this case 
not only for c-contexts but for the bigger class of \emph{singular contexts} $\scon$, which are contexts whose holes appear only under a nested application of constructor symbols or singular function arguments: $\scon ::= [~]~|~c(e_1, \ldots, \scon, \ldots, e_n)~|~f(e_1, \ldots, \scon, \ldots, e_n)$, with $c \in \CS^n$, $f \in \FS^n$ such that the subcontext appears in a singular argument of $f$, and $e_1, \ldots, e_n \in \Exp_\perp$. For singular contexts we get a compositionality result for singular contexts analogous to that of Proposition \ref{lemCompAPlCCon}---following the same scheme as the proof for Theorem \ref{lCompPS}---that can be used to easily prove the correctness of bubbling for singular contexts. 

%
%

We conclude our discussion about \sapcrwl\ and \sbpcrwl\ with the following result stating that they are in fact conservative extensions of both \crwl\ (call-time choice, or equivalently, singular non-determinism) and their corresponding plural semantics, as it was apparent from their rules.
\begin{theorem}[Conservative extension]
Under any program and for any $e \in Exp_\perp$:
\begin{enumerate}
\item If the program contains no extra variables and every function is singular then $\densap{e} = \dens{e} = \densbp{e}$.
\item If every function is plural then $\densap{e} = \denap{e}$ and $\densbp{e} = \denbp{e}$.
\end{enumerate}
\end{theorem}
\begin{proof}
If every function is singular and the program containts no extra variables then \sapor\ and \sbpor\ are equivalent to \crule{OR}, so both \crwl, \sapcrwl, and \sbpcrwl\ behave the same. Note that the absence of extra variables is essential, as for example from the program $\{f \tor d(X,X)\}$ we get $\dens{f} \not\ni d(0, 1) \in \densap{f} = \densbp{f}$.

Similarly, if every function is plural then \sapor\ and \sbpor\ are equivalent to \apor\ and \bpor, respectively. Note that extra variables pose no problem in this case, as any of these plural semantics is able to instantiate them with an arbitrary substitution from $\icsus$.
\end{proof}



\subsection{Commands}\label{subsec:commands}

In this section we introduce the concrete syntax of our language and the commands provided by our interpreter. 
The system
is started by loading in Maude the file \verb"plural.maude", available at
\url{http://gpd.sip.ucm.es/PluralSemantics}. 
It starts
an input/output loop that allows the user to introduce commands by enclosing
them in parens. 
Programs start with the keyword \verb"plural", followed
by the module name and the keyword \verb"is", and finish with \verb"endp", as exemplified in Figure \ref{fig:sampleProgPrototype}. The
body of each program is a list of statements of the form
$e_1\, \verb"->"\, e_2\, \verb"."$, indicating that the 
program rule $e_1 \tor e_2$ is part of the program.

\begin{figure}
\framebox{
\begin{minipage}{.5\textwidth}
\begin{flushleft}
\texttt{(plural SAMPLE-PROGRAM is\\
~~~f is plural .\\
~~~f(c(X)) -> p(X, X) .\\
~endp)
}
\end{flushleft}
\end{minipage}
}
    \caption{Concrete syntax of programs}
    \label{fig:sampleProgPrototype}
\end{figure}

The plurality map is specified by means of $\verb"is"\,$ annotations for each function of the program. These annotations have the form $\mathit{f}\, \verb"is"\, \mathit{plurality}\, \verb"."$, where $\mathit{plurality}$ can take the values \verb"singular" for singular functions, \verb"plural" for plural functions, or a sequence composed by the characters \verb"s" and \verb"p" specifying in more detail the plurality behaviour for each function argument, along the lines of the beginning of Section~\ref{subsec:logics}: 
if the i-th element of this chain is the character \verb"s" then the i-th argument of $\mathit{f}$ will be a singular argument, otherwise it will be considered a plural argument. \emph{Functions are considered singular by default when no} $\verb"is"\, $ \emph{annotation is provided}. 

The system is able to evaluate any expression built with the symbols of the program, under the semantics specified by the \sapcrwl\ logic. \emph{The prototype does not support programs with extra variables}, for two main reasons. First of all, it is based on the transformation from Section~\ref{sec:simulAPlural}, whose adequacy has been only proved for programs without extra variables. But the main reason is the lack of a suitable narrowing mechanism for plural variables, which is the resort usually employed by FLP systems to deal with the space explosion caused by extra variables \cite{Hanus07ICLP,LS99,Han05TR}. 
We consider the development of a plural narrowing mechanism an interesting subject of future work, but for now and for the rest of the paper, we restrict ourselves to programs not containing extra variables.

The system provides by default the constant c-terms  \verb"tt" (for \emph{true}) and \verb"ff" (for \emph{false}), 
and two more handy functions: the binary function \verb"_?_",
that is used with infix notation, and the \verb"if_then_" function, used with mixfix notation,
defined by the following rules:

{\codesize
\begin{verbatim}
 X ? Y -> X .
 X ? Y -> Y .
 if tt then E -> E .
\end{verbatim}
}

\noindent
Note that, since no \verb"is" annotation is provided, both functions are singular. 

Once a module has been introduced, the user can evaluate expressions with the
command:

{\codesize
\begin{alltt}
 (eval \([\)[depth = DEPTH]\(]\) EXPRESSION .)
\end{alltt}}

\noindent where \verb"EXPRESSION" is the expression to be evaluated and \verb"DEPTH"
is a bound in the number of steps. If this last value is omitted, the search is
assumed to be unbounded. If the term can be reduced to a c-term, it will be printed
and the user can use

{\codesize
\begin{alltt}
 (more .)
\end{alltt}}

\noindent until no more solutions are found.

It is also possible to switch between two evaluation strategies, depth-first
and breadth-first, with the commands:

{\codesize
\begin{verbatim}
(depth-first .)
(breadth-first .)
\end{verbatim}
}

Finally, the system can be rebooted with the command

{\codesize
\begin{verbatim}
(reboot .)
\end{verbatim}
}

\subsection{Examples}
In this section we show how to use the commands above, by means of two examples.

\subsubsection{Clerks}

First we show how to implement in our tool the program from Example~\refp{EPlural2}, slightly extended by adding a new branch to the bank. 
The different branches are defined by using the non-deterministic
function \texttt{?}, that here has to be understood as the set union operator. 
In the same line, for each branch the function \texttt{employees} returns the set of its employees:

{\codesize
\begin{verbatim}
 branches -> madrid ? vigo ? badajoz .

 employees(madrid) -> e(pepe, men, clerk) ? e(paco, men, boss) .
 employees(vigo) -> e(maria, women, clerk) ? e(jaime, men, boss) .
 employees(badajoz) -> e(laura, women, clerk) ? e(david, men, clerk) .
\end{verbatim}
}

Now, we define a function \texttt{twoclerks} which searches in the database for the names of two employees working as clerks. 
It calls the function \texttt{find}, which has been marked with
the keyword \texttt{plural} in order to express that its argument will be understood 
as a set of records from the database of the bank. Therefore, although the same
variable \verb"N" is used in the two components of the pair in the right-hand side of its rule, each one
can be instantiated with different values:

{\codesize
\begin{verbatim}
 twoclerks -> find(employees(branches)) .
 find is plural .
 find(e(N,G,clerk)) -> p(N,N) .
\end{verbatim}
}

Once the module has been loaded in our system,\footnote{The tool also indicates 
whether the program belongs to the class $\classAB$---remember that in that case \sapcrwl\ and \sbpcrwl\ would be equivalent and so both would be supported by the system---or not.} we can use the \texttt{eval}
command to evaluate expressions, and the command \texttt{more} to find the
next solutions:

{\codesize
\begin{verbatim}
Maude> load clerks.plural
Module introduced.
Both alpha and beta plural semantics supported for this program.

Maude> (eval twoclerks .)
Result: p(pepe,pepe)

Maude> (more .)
Result: p(pepe,maria)
\end{verbatim}
}

This program works as we expected, even if all the functions are marked as plural (i.e., if \pcrwl\ is used). However it can be improved in several directions. First of all, we are interested in getting two \emph{different} clerks. To do that we will define a function \texttt{vals} that generates a list containing different values of its argument. This function will use an auxiliary 
function \texttt{newIns} that 
appends an element at the beginning of a list ensuring that the remaining elements of the list are different to the new one.
This is checked by \texttt{diffL}, which returns the list in its second argument when it does not contain its first argument, and otherwise fails. Thus a disequality test is needed, but in our minimal framework we do not dispose of  disequality constraints, common in FLP languages \cite{Hanus07ICLP,AntoyHanusComACM10}. Nevertheless we can implement a ground version of disequality through regular program rules, as it is done here in the function \texttt{neq}.

{\codesize
\begin{verbatim}
 newIns is singular .
 newIns(X, Xs) -> cons(X, diffL(X, Xs)) .

 diffL(X, nil) -> nil .
 diffL(X, cons(Y, Xs)) -> 
      if neq(X, Y) then cons(Y, diffL(X, Xs)) .

 neq(pepe, paco) -> tt .
 neq(pepe, maria) -> tt .
 ...
\end{verbatim}
}

Note that we need \texttt{newIns}, \texttt{diffL}, and \texttt{neq} to be singular because they essentially perform tests, and when performing a test we naturally want the returning value to be the same which has been tested. For example, the following program:
{\codesize
\begin{verbatim}
 isWoman(maria) -> tt .
 isWoman(laura) -> tt .
 ...
 filterWomen(P) -> if isWoman(P) then P
\end{verbatim}
}
\noindent would have a funny behaviour if
\texttt{filterWomen} had been declared a plural function, because then for
\texttt{filterWomen(maria ? pepe)} we could compute
\texttt{pepe} as a correct value.\\

On the other hand the function \texttt{vals} is marked as \emph{plural} because it is devised to generate lists of different values of its argument. 
Note the combination of plurality, to obtain
more than one value from the argument of \texttt{vals}, and singularity, which is needed
for the tests performed by \texttt{newIns}:

{\codesize
\begin{verbatim}
 vals is plural .
 vals(X) -> newIns(X, vals(X)) .
\end{verbatim}
}

We generalize now our search function
to look for any number of clerks, not just two. To do that we will use 
the function \texttt{nVals} below, that returns a list of different values corresponding to different evaluations of its second argument. Therefore that second argument has to be declared as plural, while its first argument is singular as it fixes the number of values claimed (that is, the length of the returning list in the Peano notation for natural numbers):

{\codesize
\begin{verbatim}
 nVals is sp .
 nVals(N, E) -> take(N, vals(E)) .

 take(s(N), cons(X, Xs)) -> cons(X, take(N, Xs)) .
 take(z, Xs) -> nil .
\end{verbatim}
}

This \texttt{nVals} function is an example of how the use of plural arguments allows us to simulate some features that in a pure call-time choice context have to be defined at the meta level, in this case the \texttt{collect} \cite{LS99} or \texttt{findall} \cite{Han05TR} primitives of standard FLP systems. \\

Finally the function \texttt{nClerks} starts the search for a number of
different clerks specified by the user. It uses the auxiliary 
function \verb"findClerks", that returns the name of the clerks:

{\codesize
\begin{verbatim}
 nClerks is singular .
 nClerks(N) -> nVals(N, findClerk(employees(branches))) .

 findClerk is singular .
 findClerk(e(N,G,clerk)) -> N .
\end{verbatim}
}

Now we can search for three different clerks, obtaining \verb"pepe",
\verb"maria", and \verb"laura" as the first possible result:

{\codesize
\begin{verbatim}
Maude> (eval nClerks(s(s(s(z)))) .)
Result: cons(pepe,cons(maria,cons(laura,nil)))
\end{verbatim}
}

As anticipated in Example \refp{ExBPlural1}, we can use this technique to solve the problem of finding the names of the clerks paired with their genre, but avoiding the wrong information mixup caused by a purely plural approach using the style of the plural \verb"find" function above, under \apcrwl. To do that we just have to define a new auxiliary function \verb"findClerksNG" that this time returns a pair composed by the name of the clerk and his or her genre. 

{\codesize
\begin{verbatim}
 nClerksNG is singular .
 nClerksNG(N) -> nVals(N, findClerkNG(employees(branches))) .

 findClerkNG is singular .
 findClerkNG(e(N,G,clerk)) -> p(N, G) .
\end{verbatim}
}


The fact that \verb"findClerksNG" is singular, just like \verb"findClerks", ensures that the names and genres will be correctly paired. Besides, note that the whole Clerks program presented here belongs to the class $\classAB$, therefore its evaluation under \sapcrwl\ and \sbpcrwl\ is the same and any wrong information mixup is prevented. We can check this by searching again for three different clerks:

{\codesize
\begin{verbatim}
Maude> (eval nClerksNG(s(s(s(z)))) .)
Result: cons(p(pepe,men),cons(p(maria,women),cons(p(laura,women),nil)))
\end{verbatim}
}


In the next example 
we will see more clearly how to decide the plurality of functions. Remember that
the key idea is that singular arguments are used \emph{to fix their values} while plural arguments are needed when we want to use \emph{sets of
values}.

\subsubsection{Dungeon}\label{dungeon}

Ulysses has been captured and he wants to cheat his guardians using the gold
he carries from Troy. Thus, he needs to know whether there is an escape
(what we define as obtaining the \verb"key" of its jail) and, if
possible, which is the path to freedom (we define each step of this
path as a pair composed of a guardian and the item Ulysses obtains
from him).

He uses the function \texttt{ask}
to interchange items and information with his guardians. 
Since each guardian provides different information we have to assure
that they are not mixed, and thus its first argument will be singular; on the other hand he may 
offer different items to the same guardian, thus the second argument will be plural: this function needs plurality \texttt{sp}:

{\codesize
\begin{verbatim}
 ask is sp .
\end{verbatim}
}

The guardians have a complex behaviour, \verb"circe" exchanges
Ulysses' \verb"trojan-gold" by either the
\verb"sirens-secret" or
an \verb"item(treasure-map)"; \verb"calypso", once she receives the
\verb"sirens-secret", offers the \verb"item(chest-code)"; \verb"aeolus" can
\verb"combine" two items;\footnote{Note that we say \emph{two} items
when the function only shows \emph{one}. This rule uses the expressive
power of plural semantics to allow the combination of different items.}
and \verb"polyphemus" gives
Ulysses the \verb"key" once he can give him the combination of the
\verb"treasure-map" and the \verb"chest-code":

{\codesize
\begin{verbatim}
 ask(circe, trojan-gold) -> item(treasure-map) ? sirens-secret .
 ask(calypso, sirens-secret) -> item(chest-code) .
 ask(aeolus, item(M)) -> combine(M,M) .
 ask(polyphemus, combine(treasure-map, chest-code)) -> key .
\end{verbatim}
}

In the same line, \texttt{askWho} has as arguments a (\emph{fixed})
guardian and a message (probably with many items) for him, so it
also has plurality \verb"sp".
This function returns the next step in the Ulysses' path to freedom,
that is, a pair with the
guardian and the items obtained from him with the function
\verb"ask":

{\codesize
\begin{verbatim}
 askWho is sp .
 askWho(Guardian, Message) -> p(Guardian, ask(Guardian, Message)) .
\end{verbatim}
}

The following functions, which are in charge of computing the actions that must be performed in order 
to escape, are marked as plural because they treat their corresponding arguments as
sets of pairs where the second component is an item or some piece of information, and the first one is the actor which provided it. 
The function \verb"discoverHow" returns the set of pairs of that shape that can be obtained starting from those contained in its argument, and then chatting to the guardians. Hence it returns, either its argument, or 
the result of exchanging the current information with some guardian and
then iterating the process. That exchange is performed with \verb"discStepHow", that non-deterministically offers some of the items or information available, to one of the guardians:

{\codesize
\begin{verbatim}
 discoverHow is plural .
 discoverHow(T) -> T ? discoverHow(discStepHow(T) ? T) .

 discStepHow is plural .
 discStepHow(p(W, M)) -> askWho(guardians, M) .

 guardians -> circe ? calypso ? aeolus ? polyphemus .
\end{verbatim}
}

Note that the additional disjunction \texttt{? T} in the recursive call to \texttt{discStepHow} is needed in order to be able to combine the old information with the new one resulting after one exchanging step. 
This point can be illustrated better with the following program:
{\codesize
\begin{verbatim}
 genPairs is plural .
 genPairs(P) -> P ? genPairs(genPairsStep(P) ? P) .

 genPairsStep is plural .
 genPairsStep(P) -> p(P, P) .

 genPairsBad is plural .
 genPairsBad(P) -> P ? genPairsBad(genPairsStep(P)) .
\end{verbatim}
}
There the functions \texttt{genPairs} and \texttt{genPairsBad} follow the same pattern as \texttt{discoverHow}, but this time are designed to generate values made up with pairs and the 
supplied argument. Besides these functions share the same ``step function'' \texttt{genPairsStep}. Nevertheless their behaviour is very different, as we can see evaluating the expressions \texttt{genPairs(z)} and  \texttt{genPairsBad(z)}: the point is that the value \texttt{p(p(z,z),z)} can be computed for the former but not for the latter, because \texttt{z} and \texttt{p(z,z)} are values generated in different 
recursive calls to \texttt{genPairsBad}. But this poses no problem for \texttt{genPairs}, because the extra \texttt{? P} in its definition makes it possible to combine those values.\\

Finally, the search is started with the function \texttt{escapeHow},
that initializes the search with the trojan gold provided by Ulysses:

{\codesize
\begin{verbatim}
 escapeHow -> discoverHow(p(ulysses, trojan-gold)) .
\end{verbatim}
}

Once the module is introduced, we can start the search with the command:

{\codesize
\begin{verbatim}
Maude> (eval escapeHow .)
Result: p(ulysses,trojan-gold)
\end{verbatim}
}

When this first result has been computed, we can ask the tool for more
with the command \texttt{more}, that progressively will show the path
followed by Ulysses to escape:

{\codesize
\begin{verbatim}
Maude> (more .)
Result: p(circe,item(treasure-map))

Maude> (more .)
Result: p(circe,sirens-secret)

Maude> (more .)
Result: p(calypso,item(chest-code))

...

Maude> (more .)
Result: p(polyphemus,key)
\end{verbatim}
}

In this example the function \texttt{discoverHow} is an instance of an interesting pattern of plural function: a function that performs deduction by repeatedly combining the information we have fed it with the information it infers in one step of deduction. Therefore in its definition the function \texttt{?} has to be understood again as the set union operator, as it is used to add elements to the set of deduced information. 
On the other hand the use of a singular argument in \texttt{askWho} is unavoidable to be able to keep track of the guardian who answers the question, while its second argument has to be plural because it represents the knowledge accumulated so far.

Several variants of this problem can be conceived, in particular currently it is simplified because the items are not lost after each exchange---this is why Ulysses' bag is bottomless. Anyway we think that this version of the problem is relevant because in fact it corresponds to a small model of an intruder analysis for a security protocol, where 
Ulysses is the intruder, the guardians are the honest principals, the key is the secret and complex behaviours of the principals can be described through the patterns in left-hand sides of program rules. 
In this case we assume that the intruder is able to store any amount of information, and that this information can be used many times. Nevertheless we also think that different variants of the problem should be tackled in the future, and that the addition of equality and disequality constraints to our framework could be decisive to deal with those problems.

\bigskip
With this program we conclude our presentation of some examples that show the expressive capabilities of our plural semantics. 
In these examples we have tried to find a way of using \spcrwl\ for programming,  so it could be more than just a semantic eccentricity. Although we have found some interesting uses of our plural semantics, in particular the meta-like function \texttt{nVals}, and the deduction programming pattern correspondint to \texttt{discoverHow}, we cannot still say that we have found a ``killer application'' for our plural semantics. Only time will tell us if these semantics are useful, because these proposals are still too young to have a reasonable benchmark collection. Our prototype opens the door to experimenting with these new semantics, and in that sense it contributes to the development of such collection. Anyway, we admit that our plural semantics probably will only be useful in some fragments of the programs, and that is why we have proposed to combine it with the usual singular semantics of FLP. 
As a final remark the reader can check, by hand or by using our prototype, that all the program examples in this section belong to the $\classAB$ class---and hence they behave the same both under \sapcrwl\ and \sbpcrwl---, which motivates the relevance of that class of programs. But again, as the collection of examples is very small, this does not give a strong argument about the usefulness of this class of programs, but just an encouraging indicator.  

\subsection{Discussion: to be singular or to be plural?}\label{discussion}
After these examples, we (hopefully) should have some intuitions about how to decide the plurality of function arguments. 
Our first resort is considering that plural arguments are used to represent sets of values, while singular arguments denote single values. But this does not work for any situation, for example consider the function \texttt{findClerk} whose plurality is singular, although its argument intuitively denotes a set of records from the database. On the other hand we may consider 
that its argument denotes a single record, and that \texttt{findClerk} defines how to extract the name from a single employee, which motivates the final plurality choice. In this case the program 
behaves the same declaring \texttt{findClerk} either singular or plural, because 
the variables in its arguments are used only once. 
As a rule of thumb we should try to have as little plural arguments as possible, because these arguments increase the search space more than the singular ones, as using a plural semantics we can compute more values than under a singular semantics, as seen in Section~\ref{sect:hierarchy}. 
Hence in this case it is better to declare \texttt{findClerk} as singular.

Thus having a more formal criterion about the equivalence of 
plurality maps 
would be useful to 
minimize the search space of our programs and understand them better. 
%
%
A static adaptation of the determinism analysis of \cite{DBLP:journals/jflp/CaballeroL03} could be useful, as it would help us to detect deterministic functions of our programs, for which the plurality map would not matter, as we expect 
to easily extend the equivalence results of singular/call-time choice and run-time choice for deterministic programs of \cite{lrs07,LMRS10LetRwTPLP} to our plural semantics.
%
%
%
%
%
%
%
%
%
We also should try to develop equational laws about non-determinism. In fact a first step in this line is the %
discussion about the correctness of bubbling for singular contexts from Section \ref{subsec:logics}.  
%
Anyway, all these are subjects of future work.


\subsection{Implementation}\label{sec:implemen}

The system described in the previous sections has been implemented in the Maude
system~\cite{maude-book},
a high-level language and high-performance system supporting both equational and rewriting logic
computation for a wide range of applications. 
The fundamental ingredients for this implementation are a core language into which all programs are transformed, and an interpreter for the operational semantics of the core language that is used to execute programs.

The transformation into core language treats 
each program rule separately and 
applies two different
transformation stages to them. The first one applies a modification of the transformation described in
Definition~\ref{DefPstOpt}, but now taking into account only those arguments 
marked as \emph{pl} in the
plurality map described in Section~\ref{subsec:logics}. 
Then in 
the second stage, which consists
in a modification of the sharing transformation of \cite[Def.~1]{lsr09},
we introduce a let-binding for each singular variable that also appears in the right-hand side, therefore obtaining subexpression sharing, and as a consequence, a singular behaviour for those arguments. 

Once source programs have been transformed into core programs, we can execute 
them by using
a heap-based operational semantics for the core language \cite{rrpepm10}. 
A heap is just a mapping from variables to expressions that represents a graph structure, as the image of each variable is interpreted as a subgraph definition. The nodes of that implied graph are defined according to those let-bindings introduced by the transformation into core language. 
The operational semantics manipulates this heap, and contains  
rules for removing useless bindings, propagating the terms associated to a variable, and for creating 
new bindings for each singular argument when their corresponding let-bindings are found.
Finally, 
in order to turn the operational semantics of the core language into an effective operational mechanism for \spcrwl, we have adapted the natural rewriting strategy in \cite{escobar04} to deal with these heaps, ensuring that the evaluation is performed on-demand. Both the program transformation into core language, the interpreter, and our adaptation of natural rewriting, have been implemented in Maude with an intensive exploiting of its reflection capabilities, thus obtaining an executable interpreter for \spcrwl. More details about our implementation can be found in~\cite{rrpepm10}.

We decided to follow the line of employing the transformation from Definition~\ref{DefPstOpt} and then using a language that implements non-deterministic term rewriting to run the transformed program, because our motivation was to obtain a simple proof of concept prototype that could be used to experiment with the new semantics, 
but obtaining an optimized  implementation is out of the scope of this work. 
The addition of the natural rewriting on-demand strategy was neccesary in order to reduce the search space to a reasonable size, but we are aware of other approaches that maybe could improve the efficiency of the system. In particular we could have adapted the techniques in \cite{AntoyIM02ConsAlt,DBLP:conf/aplas/BrasselH07,BrasselHPR11KicsDos} that rely on turning the function $? \in \FS^2$ into a constructor in order to explicitly represent non-deterministic computations in a deterministic language, which results in additional advantages like a kind of backtracking memoization called ``sharing across non-determinism.'' 
That would have allowed us to use Maude functional modules, which are much more efficient than the non-deterministic system modules that are used in our current implementation. But as functional modules perform eager evaluation, then we should also employ the context-sensitive rewriting \cite{lucas98contextsensitive} features of Maude---offered as \texttt{strat} annotations---to get the lazy evaluation that correspond to our semantics. 
That would have entailed adapting the techniques from \cite{AntoyIM02ConsAlt,DBLP:conf/aplas/BrasselH07,BrasselHPR11KicsDos} from a call-time choice setting to the run-time choice semantics of term rewriting, and also using the techniques in \cite{DBLP:conf/alp/Lucas97} to introduce the \texttt{strat} annotations needed to ensure lazy evaluation. This is a possible roadmap that could be followed in case a more optimized implementation of \spcrwl\ should be developed. 
The monad transformer of \cite{FKC09} is another alternative in the same line, as it also provides a representation of non-determinism with support for memoization in a deterministic language, in this case Haskell, that could be used as the basis for an implementation of \spcrwl\ by modifying the transformation from FLP programs with call-time choice into Haskell from \cite{BrasselFHR10TransMonad}. As that work is placed in a higher order setting, our plural semantics should be first extended with higher order capabilities, following the line of \cite{GHR97}.



\section{Concluding remarks and future work}\label{sec:conclusions}

The starting point of this work is the observation that the traditional identification between run-time choice and a plural denotational semantics is wrong in a non-deterministic functional language with pattern matching. To illustrate that, we have provided formulations for two different plural semantics that are different from run-time choice: the \apcrwl\ and \bpcrwl\ semantics. 
We argue that the run-time choice semantics induced by term rewriting is not the best option 
for a value-based programming language like current implementations of FLP because of its lack of compositionality. Nevertheless, our plural semantics 
%
are compositional for a simple notion of value---the notion of partial c-term---, just like the usual call-time choice semantics adopted by modern FLP languages, following the value-based philosophy of the FLP paradigm: ``all I care about an expression is the set of its values.'' This, together with the fact that our concrete formulations for these plural semantics are variants of the \crwl\ logic
---a standard formulation for singular/call-time choice semantics in FLP---, turns the 
problem of devising 
a combined semantics for singular and plural non-determinism into a trivial task, getting the \sapcrwl\ and \sbpcrwl\ logics as a result. 
The combination of singular and plural semantics in the same language is interesting and follows naturally when programming, as 
it allows us to reuse known programming patterns 
from the more usual singular/call-time choice semantics, standard in modern FLP systems, 
while we are still able to use  
the new capabilities of the novel plural semantics for some interesting fragments of the program. 
In these logics, apart from the program, the user may specify for each function which of its arguments will be marked as singular and which as plural, resulting in different parameter passing mechanism. A simple intuition that works in most situations can be considering plural arguments as sets of values and singular arguments as individual values. 
We have not only 
proposed such 
semantic combinations but we have also provided a prototype implementation for \sapcrwl\ using the Maude system (see \cite{RRLDTA09,rrpepm10} for details about the implementation), 
in which the program transformation to simulate \apcrwl\ with term rewriting---a standard formulation for run-time choice---also presented in this work is a crucial ingredient.  
\podaSecImp{based on a program transformation to simulate \apcrwl\ with term rewriting---a standard formulation for run-time choice---also presented here, and a transformation to achieve a singular behaviour in certains parts of the program by means of a let-bindings contruction, developed in a previous work \cite{lsr09}. This let-construction has been simplified by taking advantage of some invariants of the program transformations, obtaining 
an operational semantics for the core language used by our prototype that relies on a heap-based representation of term graphs. 
Finally, the natural rewriting on-demand evaluation strategy 
has been adapted to deal with the configurations of our operational semantics, in order to perform lazy evaluation oriented to the computation of values. 
This is neccesary to handle the wild non-determinism explosion arising in the modules after the transformation. While this non-determinism is necessary to completely model check a system, because then we want to test whether a given property holds in every state reachable during the evolution of the system, it is not the case when we are just computing the set of values---reachable c-terms---for a given expression. In this case we may use an on-demand strategy like natural rewriting to prune the search space without losing interesting values. 
Maude was the natural choice for an object language, because both the transformations and the strategy upon which our implementation is based have been formulated at the abstraction level of source programs, that is, using the syntax and theory 
of term rewriting and term graph rewriting, which are the foundations of Maude. As a consequence there is a small distance between the specifications of the transformations, the operational semantics and the strategy, and the code which implements them, thanks to the meta-programming capabilities of Maude, that allows the
user to manipulate modules and rules as usual data, which has been used to implement the program transformation, the operational semantics, and the natural rewriting strategy in a natural way.}
The resulting system, available at \url{http://gpd.sip.ucm.es/PluralSemantics}, is an interpreter for the \sapcrwl\ logic that we have 
used to develop several programs examples that exploit the new expressive capabilities of the combined semantics, in order to improve the declarative flavour of programs. 

Along the way we have also made several contributions at the foundational level. We have studied the technical properties of \apcrwl\ and  \bpcrwl, providing formal proofs for its compositionality and also for other interesting properties like polarity, several monotonicity properties for substitutions, and a restricted form of bubbling 
for constructor contexts. Then we have compared the different semantics for non-determinism considered in this work w.r.t.\ the set of computed values, concluding that they form the inclusion chain $\crwl \subseteq $ term rewriting $\subseteq \bpcrwl \subseteq \apcrwl$, corresponding to the chain singular/call-time choice $\subseteq$ run-time choice $\subseteq$ $\beta$-plural $\subseteq$ $\alpha$-plural. Besides, we have determined that for the class of programs $\classAB$, characterized by a simple syntactic criterion, our plural semantics proposals \apcrwl\ and \bpcrwl\ are equivalent. We have also provided a formal proof of the adequacy of the (non-optimized version of the) transformation used by our prototype to simulate \apcrwl\ with term rewriting. As a consequence, this transformation can be used to simulate \bpcrwl\ for programs in the class $\classAB$. \podaSecImp{, hence implementing the \sbpcrwl\ logic for programs in that class when combined with the singular semantics transformation.} Regarding the combined semantics \sapcrwl\ and \sbpcrwl, it is easy to see that they inherit the good properties of both \crwl, \apcrwl, and \bpcrwl, \podaSecImp{. We have tackled compositionality and bubbling, that in this case is sound not only for constructor contexts but for the bigger class of singular contexts,} and 
we have also 
proved that the combined semantics are conservative extensions of both singular/call-time choice and their corresponding plural semantics. 

~\\
These questions were first approached in previous works by the authors \cite{rodH08,RRLDTA09,rrpepm10}, however in this paper we do not only give a revised and unified presentation, 
but we have also included several important novel results.
\begin{itemize}
	\item All the technical results from those works have been extended to deal with programs with extra variables, except those results regarding the simulation of \apcrwl\ with term rewriting from Section \ref{sec:simulAPlural}. The new technical results have been also proved for programs with extra variables. Besides, 
we have fixed some errata from the original works, in particular the formulation of bubbling for \apcrwl, the definition of the operator $?$ over sequences of $\CSubst_\perp$, and also some other minor mistakes in the proofs. The formulations of bubbling for constructor and singular contexts are novel contributions of this paper. 

    \item The plural semantics \bpcrwl, inspired in the proposal from \cite{BB09}, is introduced in this work for the first time. We give clear explanations of some problematic situations where \apcrwl\ performs a wrong information mixup, and how our attempts to fix those problems, inspired in the solutions from \cite{BB09}, led us to the current formulation of \bpcrwl, which leans on the notion of compressible set of partial c-substitutions. 

    \item As the formulations of \apcrwl\ and \bpcrwl\ are very similar, it was not difficult to check that \bpcrwl\ also enjoys the same basic properties of \apcrwl. Nevertheless, it was more difficult to place \bpcrwl\ in the semantic inclusion chain from \cite{rodH08}, being a key idea the notion of compressible completion of a set of $\CSubst_\perp$, and its related results. The characterization of the class of programs $\classAB$ for which \apcrwl\ and \bpcrwl\ are equivalent, and the formal proof for that equivalence are also novel, obviously. 
    
    \item Finally, the logic \sbpcrwl\ is also a novel contribution of this work, but in this case its definition was straightforward, because it follows the same pattern as the definition for \sapcrwl.

\podaSecImp{    \item We give much more explanations about our implementation. We have included the most important Maude code together with copious explanations about it, and we have shown our adaptation of the natural rewriting strategy to the configurations used by our operational semantics, for which we had to conceive a suitable notion of position in the heap, and the redefine the function $\mt$, which constitutes the core of our implementation of natural rewriting based on matching definitional trees.

    \item Our implementation of natural rewriting for Maude system modules has been carefully crafted as an independent library, in order to improve its reusability. The corresponding \texttt{natNext} function can be used for performing on-demand evaluation of any \ctrs\ specified in a Maude system module, and it is especially relevant because is the first on-demand strategy for this kind of
modules, complementing the default rewrite and breadth-first search Maude commands.
This library is available at \url{http://gpd.sip.ucm.es/NaturalRewriting}. }
\end{itemize}

Previously to ours, not much work has been done in the combination of singular and plural 
non-determinism in functional or functional-logic programming, since mainstream approaches \cite{Wadler85,LS99,Han05TR} only support the usual singular semantics.
Closer are the combinations of call-time and run-time choice of \cite{lsr09,DBLP:journals/corr/abs-0903-2205}, which anyway follow a different approach as the plural sides of \sapcrwl\ and \sbpcrwl\ are essentially different to run-time choice.
Anyway, we still think that the combination of call-time choice and run-time choice is not very suitable for value-based languages because of the lack of compositionality for values under run-time choice. 
The monad transformer of \cite{FKC09}, devised to improve the laziness of non-deterministic monads while retaining a call-time choice semantics, is based on a \texttt{share} combinator which plays a role similar to the let-bindings of our core language. The authors seem to be interested in staying in a pure call-time choice framework, but maybe a combination of call-time and run-time choice could be achieved there too, getting something similar to \cite{lsr09} but again essentially different to \sapcrwl\ and \sbpcrwl\ for the same reason. 
Besides that work is focused in implementation issues of FLP in concrete deterministic functional languages, while in ours we start from 
the more abstract world of \ctrss\ and are fundamentally concerned in exploring the language design space. 

~\\
We contemplate several interesting subjects of future work. As pointed in Sections~\ref{sec:simulAPlural} and \ref{subsec:commands}, the development of a suitable plural narrowing mechanism would be the key for finding an effective way of handling extra and free variables. Besides, in our examples it has arisen the necessity of equality and disequality constraints (whose ground versions have been simulated by using regular functions), that will ease and shorten the definition of programs, and increase the expressiveness of the setting. Both subjects would be interesting at the theoretical and practical levels, as we could then improve our prototype by extending it with those new features.

Similarly, adding higher order capabilities by an extension of \spcrwl\ in the line of \cite{GHR97}, and implementing them by means of the classic transformation of \cite{War82}, would also be interesting and it is standard in the field of FLP. Then, for example, we could define a more generic version of \texttt{discoverHow} with an additional argument for the function used to perform a deduction step (\texttt{discStepHow} in our dungeon problem). This higher order version of \spcrwl\ could also be used to face the challenges regarding the implementation of type classes in FLP through the classical transformational technique of \cite{Wadler:ad-hoc_polymorphism} pointed out by Lux in \cite{curryListLux09}. Although some solutions based on the frameworks of
\cite{lsr09,DBLP:journals/corr/abs-0903-2205} were already proposed in \cite{curryListRod09} we think that an alternative based on \spcrwl\ would be better thanks to its clean and compositional semantics. 
More novel would be using the matching-modulo capacities of Maude to enhance the expressiveness of the semantics, after a corresponding revision of the theory of \spcrwl. 
Besides,  as mentioned at the end of Section  \ref{sec:implemen}, 
some additional research must be done to improve the performance of the interpreter,
especially because of the increase of the size of the search space due to the use of plural arguments. 
As we pointed out there, an explicit representation of non-determinism on a deterministic language seems promising \cite{AntoyIM02ConsAlt,DBLP:conf/aplas/BrasselH07,BrasselHPR11KicsDos,FKC09}, in particular
the memoization capabilities of these approaches could be exploited to deal with the non-determinism overhead caused by plural arguments. Some possible concretization of this idea could be using Maude functional modules with \texttt{strat} annotations using the techniques in \cite{DBLP:conf/alp/Lucas97}, or adapting the transformation in \cite{BrasselFHR10TransMonad}. 
%


\podaSecImp{We also contemplate the extension of our Maude implementation of natural rewriting to \trss\ which are not necessarily constructor-based, following the theory developed in \cite{DBLP:conf/lopstr/EscobarMT04}. Another orthogonal extension of the strategy could go in the direction of combining non-deterministic rewriting rules evaluated on demand (from Maude system modules) with deterministic and terminating equations evaluated eagerly (from Maude functional modules).
This is particularly interesting when there are fragments of a \trs\ which constitute a confluent and terminating \trs\ (like the \texttt{match} and \texttt{project} functions introduced by our transformation), that could be executed dramatically more efficiently by treating them as Maude equations.}

As suggested in Section \ref{discussion}, finding a criterion for the equivalence of plurality maps and defining more equational laws for non-determinism, besides the restricted forms of bubbling proposed here, would improve the understanding of programs, which could finally lead to the development of more interesting program examples that maybe could illustrate the interest of the semantics. 
In this line we also find interesting the relation between the different notions of determinism entailed by \crwl, \apcrwl\ and \bpcrwl, and the relation between confluence of term rewriting and those notions of determinism. We already made some advances in this line in previous works \cite{lrs07,LMRS10LetRwTPLP}. 

To conclude, an investigation of the technical relation between \bpcrwl\ and the plural semantics from \cite{BB09} which inspired it, would be very interesting. We conjecture a strong semantic equivalence between them.

{\small

\begin{thebibliography}{}

\bibitem[\protect\citeauthoryear{Albert, Hanus, Huch, Oliver, and Vidal}{Albert
  et~al\mbox{.}}{2005}]{AHHOV05}
{\sc Albert, E.}, {\sc Hanus, M.}, {\sc Huch, F.}, {\sc Oliver, J.}, {\sc and}
  {\sc Vidal, G.} 2005.
\newblock Operational semantics for declarative multi-paradigm languages.
\newblock {\em Journal of Symbolic Computation\/}~{\em 40,\/}~1, 795--829.

\bibitem[\protect\citeauthoryear{Antoy, Brown, and Chiang}{Antoy
  et~al\mbox{.}}{2007}]{AntoyBrownChiang06Termgraph}
{\sc Antoy, S.}, {\sc Brown, D.}, {\sc and} {\sc Chiang, S.} 2007.
\newblock Lazy context cloning for non-deterministic graph rewriting.
\newblock In {\em Proc. Termgraph'06,}. ENTCS, 176(1), 61--70.

\bibitem[\protect\citeauthoryear{Antoy and Hanus}{Antoy and
  Hanus}{2002}]{DBLP:conf/flops/AntoyH02}
{\sc Antoy, S.} {\sc and} {\sc Hanus, M.} 2002.
\newblock Functional logic design patterns.
\newblock In {\em FLOPS}, {Z.~Hu} {and} {M.~Rodr\'{\i}guez-Artalejo}, Eds.
  Lecture Notes in Computer Science, vol. 2441. Springer, 67--87.

\bibitem[\protect\citeauthoryear{Antoy and Hanus}{Antoy and
  Hanus}{2010}]{AntoyHanusComACM10}
{\sc Antoy, S.} {\sc and} {\sc Hanus, M.} 2010.
\newblock Functional logic programming.
\newblock {\em Communications of the ACM\/}~{\em 53,\/}~4, 74--85.

\bibitem[\protect\citeauthoryear{Antoy, Iranzo, and Massey}{Antoy
  et~al\mbox{.}}{2002}]{AntoyIM02ConsAlt}
{\sc Antoy, S.}, {\sc Iranzo, P.~J.}, {\sc and} {\sc Massey, B.} 2002.
\newblock Improving the efficiency of non-deterministic computations.
\newblock {\em Electr. Notes Theor. Comput. Sci.\/}~{\em 64}, 73--94.

\bibitem[\protect\citeauthoryear{Ariola, Felleisen, Maraist, Odersky, and
  Wadler}{Ariola et~al\mbox{.}}{1995}]{AriolaFMOW95}
{\sc Ariola, Z.~M.}, {\sc Felleisen, M.}, {\sc Maraist, J.}, {\sc Odersky, M.},
  {\sc and} {\sc Wadler, P.} 1995.
\newblock The call-by-need lambda calculus.
\newblock In {\em Proceedings of 22nd Annual ACM SIGACT-SIGPLAN Symposium on
  Principles of Programming Languages, POPL 1995}. ACM, 233--246.

\bibitem[\protect\citeauthoryear{Baader and Nipkow}{Baader and
  Nipkow}{1998}]{baader-nipkow}
{\sc Baader, F.} {\sc and} {\sc Nipkow, T.} 1998.
\newblock {\em Term Rewriting and All That}.
\newblock Cambridge University Press.

\bibitem[\protect\citeauthoryear{Borovansk\'y, Kirchner, Kirchner, Moreau, and
  Ringeissen}{Borovansk\'y
  et~al\mbox{.}}{1998}]{BorovanskyKirchnerKirchnerMoreauRingeissenWRLA98}
{\sc Borovansk\'y, P.}, {\sc Kirchner, C.}, {\sc Kirchner, H.}, {\sc Moreau,
  P.-E.}, {\sc and} {\sc Ringeissen, C.} 1998.
\newblock An overview of {ELAN}.
\newblock In {\em Proceedings Second International Workshop on Rewriting Logic
  and its Applications, WRLA'98, Pont-\`a-Mousson, France, September 1--4,
  1998}, {C.~Kirchner} {and} {H.~Kirchner}, Eds. Electronic Notes in
  Theoretical Computer Science, vol.~15. Elsevier, 329--344.
\newblock \url{http://www.elsevier.nl/locate/entcs/volume15.html}.

\bibitem[\protect\citeauthoryear{Bra{\ss}el and Berghammer}{Bra{\ss}el and
  Berghammer}{2009}]{BB09}
{\sc Bra{\ss}el, B.} {\sc and} {\sc Berghammer, R.} 2009.
\newblock Functional (logic) programs as equations over order-sorted algebras.
\newblock In {\em Informal Proceedings 19th International Symposium on
  Logic-Based Program Synthesis and Transformation, LOPSTR 2009}.

\bibitem[\protect\citeauthoryear{Bra{\ss}el, Fischer, Hanus, and
  Reck}{Bra{\ss}el et~al\mbox{.}}{2010}]{BrasselFHR10TransMonad}
{\sc Bra{\ss}el, B.}, {\sc Fischer, S.}, {\sc Hanus, M.}, {\sc and} {\sc Reck,
  F.} 2010.
\newblock Transforming functional logic programs into monadic functional
  programs.
\newblock In {\em WFLP}, {J.~Mari{\~n}o}, Ed. Lecture Notes in Computer
  Science, vol. 6559. Springer, 30--47.

\bibitem[\protect\citeauthoryear{Bra{\ss}el, Hanus, Peem{\"o}ller, and
  Reck}{Bra{\ss}el et~al\mbox{.}}{2011}]{BrasselHPR11KicsDos}
{\sc Bra{\ss}el, B.}, {\sc Hanus, M.}, {\sc Peem{\"o}ller, B.}, {\sc and} {\sc
  Reck, F.} 2011.
\newblock {KiCS2}: A new compiler from {Curry} to {Haskell}.
\newblock In {\em Proceedings of the 20th international conference on
  Functional and constraint Logic Programming, WFLP 2011}, {H.~Kuchen}, Ed.
  Lecture Notes in Computer Science, vol. 6816. Springer, 1--18.

\bibitem[\protect\citeauthoryear{Bra{\ss}el and Huch}{Bra{\ss}el and
  Huch}{2007}]{DBLP:conf/aplas/BrasselH07}
{\sc Bra{\ss}el, B.} {\sc and} {\sc Huch, F.} 2007.
\newblock On a tighter integration of functional and logic programming.
\newblock In {\em APLAS}, {Z.~Shao}, Ed. Lecture Notes in Computer Science,
  vol. 4807. Springer, 122--138.

\bibitem[\protect\citeauthoryear{Caballero and L{\'o}pez-Fraguas}{Caballero and
  L{\'o}pez-Fraguas}{2003}]{DBLP:journals/jflp/CaballeroL03}
{\sc Caballero, R.} {\sc and} {\sc L{\'o}pez-Fraguas, F.} 2003.
\newblock Improving deterministic computations in lazy functional logic
  languages.
\newblock {\em Journal of Functional and Logic Programming\/}~{\em 2003,\/}~1.

\bibitem[\protect\citeauthoryear{Clavel, Dur{\'a}n, Eker, Lincoln,
  Mart{\'\i}-Oliet, Meseguer, and Talcott}{Clavel
  et~al\mbox{.}}{2007}]{maude-book}
{\sc Clavel, M.}, {\sc Dur{\'a}n, F.}, {\sc Eker, S.}, {\sc Lincoln, P.}, {\sc
  Mart{\'\i}-Oliet, N.}, {\sc Meseguer, J.}, {\sc and} {\sc Talcott, C.} 2007.
\newblock {\em All About Maude: A High-Performance Logical Framework}. Lecture
  Notes in Computer Science, vol. 4350.
\newblock Springer.

\bibitem[\protect\citeauthoryear{DeGroot and Lindstrom}{DeGroot and
  Lindstrom}{1986}]{deGroot86}
{\sc DeGroot, D.} {\sc and} {\sc Lindstrom, G.~e.} 1986.
\newblock {\em Logic Programming, Functions, Relations, and Equations}.
\newblock Prentice Hall.

\bibitem[\protect\citeauthoryear{Dijkstra}{Dijkstra}{1997}]{Dijkstra97}
{\sc Dijkstra, E.~W.} 1997.
\newblock {\em A Discipline of Programming}.
\newblock Prentice Hall PTR, Upper Saddle River, NJ, USA.

\bibitem[\protect\citeauthoryear{Echahed and Janodet}{Echahed and
  Janodet}{1998}]{EchahedJanodet98JICSLP}
{\sc Echahed, R.} {\sc and} {\sc Janodet, J.-C.} 1998.
\newblock {A}dmissible {G}raph {R}ewriting and {N}arrowing.
\newblock In {\em {P}roceedings of the 1998 {J}oint {I}nternational
  {C}onference and {S}ymposium on {L}ogic {P}rogramming}. MIT Press,
  Manchester, 325 -- 340.

\bibitem[\protect\citeauthoryear{Escobar}{Escobar}{2004}]{escobar04}
{\sc Escobar, S.} 2004.
\newblock Implementing natural rewriting and narrowing efficiently.
\newblock In {\em Proceedings of the 7th International Symposium on Functional
  and Logic Programming, {FLOPS} 2004}, {Y.~Kameyama} {and} {P.~Stuckey}, Eds.
  Lecture Notes in Computer Science, vol. 2998. Springer, 147--162.

\bibitem[\protect\citeauthoryear{Fischer, Kiselyov, and Shan}{Fischer
  et~al\mbox{.}}{2009}]{FKC09}
{\sc Fischer, S.}, {\sc Kiselyov, O.}, {\sc and} {\sc Shan, C.-c.} 2009.
\newblock Purely functional lazy non-deterministic programming.
\newblock In {\em ICFP '09: Proceedings of the 14th ACM SIGPLAN international
  conference on Functional programming}. ACM, New York, NY, USA, 11--22.

\bibitem[\protect\citeauthoryear{Futatsugi and Diaconescu}{Futatsugi and
  Diaconescu}{1998}]{cafe-report}
{\sc Futatsugi, K.} {\sc and} {\sc Diaconescu, R.} 1998.
\newblock {\em {CafeOBJ} Report}.
\newblock World Scientific, AMAST Series.

\bibitem[\protect\citeauthoryear{G{onz\'alez-Moreno}, H{ortal\'a-Gonz\'alez},
  and R{odr\'{\i}guez-Artalejo}}{G{onz\'alez-Moreno}
  et~al\mbox{.}}{1997}]{GHR97}
{\sc G{onz\'alez-Moreno}, J.}, {\sc H{ortal\'a-Gonz\'alez}, M.}, {\sc and} {\sc
  R{odr\'{\i}guez-Artalejo}, M.} 1997.
\newblock A higher order rewriting logic for functional logic programming.
\newblock In {\em Proc. International Conference on Logic Programming, ICLP
  1997}. MIT Press, 153--167.

\bibitem[\protect\citeauthoryear{G{onz\'alez-Moreno}, H{ortal\'a-Gonz\'alez},
  L{\'opez-Fraguas}, and R{odr\'{\i}guez-Artalejo}}{G{onz\'alez-Moreno}
  et~al\mbox{.}}{1996}]{GHLR96}
{\sc G{onz\'alez-Moreno}, J.~C.}, {\sc H{ortal\'a-Gonz\'alez}, T.}, {\sc
  L{\'opez-Fraguas}, F.}, {\sc and} {\sc R{odr\'{\i}guez-Artalejo}, M.} 1996.
\newblock {A} {R}ewriting {L}ogic for {D}eclarative {P}rogramming.
\newblock In {\em Proc. European Symposium on Programming, ESOP 1996}. Lecture
  Notes in Computer Science, vol. 1058. Springer, 156--172.

\bibitem[\protect\citeauthoryear{G{onz\'alez-Moreno}, H{ortal\'a-Gonz\'alez},
  L{\'opez-Fraguas}, and R{odr\'{\i}guez-Artalejo}}{G{onz\'alez-Moreno}
  et~al\mbox{.}}{1999}]{GHLR99}
{\sc G{onz\'alez-Moreno}, J.~C.}, {\sc H{ortal\'a-Gonz\'alez}, T.}, {\sc
  L{\'opez-Fraguas}, F.}, {\sc and} {\sc R{odr\'{\i}guez-Artalejo}, M.} 1999.
\newblock An approach to declarative programming based on a rewriting logic.
\newblock {\em Journal of Logic Programming\/}~{\em 40,\/}~1, 47--87.

\bibitem[\protect\citeauthoryear{Hanus}{Hanus}{2005}]{Han05TR}
{\sc Hanus, M.} 2005.
\newblock Functional logic programming: From theory to {C}urry.
\newblock Tech. rep., Christian-Albrechts-Universit\"at Kiel.

\bibitem[\protect\citeauthoryear{Hanus}{Hanus}{2007}]{Hanus07ICLP}
{\sc Hanus, M.} 2007.
\newblock Multi-paradigm declarative languages.
\newblock In {\em Proceedings of the International Conference on Logic
  Programming, ICLP 2007}. Springer LNCS 4670, 45--75.

\bibitem[\protect\citeauthoryear{H{anus (ed.)}}{H{anus
  (ed.)}}{2006}]{Han06curry}
{\sc H{anus (ed.)}, M.} 2006.
\newblock {C}urry: An integrated functional logic language (version 0.8.2).
\newblock Available at {\it http://www.informatik.uni-kiel.de/\verb+~+}{\it
  curry/}{\it report.html}.

\bibitem[\protect\citeauthoryear{Hermenegildo, Bueno, Carro,
  L{\'o}pez-Garc\'{\i}a, Mera, Morales, and Puebla}{Hermenegildo
  et~al\mbox{.}}{2011}]{CiaoSurvery2011}
{\sc Hermenegildo, M.~V.}, {\sc Bueno, F.}, {\sc Carro, M.}, {\sc
  L{\'o}pez-Garc\'{\i}a, P.}, {\sc Mera, E.}, {\sc Morales, J.~F.}, {\sc and}
  {\sc Puebla, G.} 2011.
\newblock An overview of ciao and its design philosophy.
\newblock {\em CoRR\/}~{\em abs/1102.5497}.

\bibitem[\protect\citeauthoryear{Hughes and O'Donnell}{Hughes and
  O'Donnell}{1990}]{Hughes89}
{\sc Hughes, J.} {\sc and} {\sc O'Donnell, J.} 1990.
\newblock {E}xpressing and {R}easoning {A}bout {N}on-deterministic {F}unctional
  {P}rograms.
\newblock In {\em Proceedings of the 1989 Glasgow Workshop on Functional
  Programming}. Workshops in Computing. Springer, London, UK, 308--328.

\bibitem[\protect\citeauthoryear{Hussmann}{Hussmann}{1993}]{hussmann93}
{\sc Hussmann, H.} 1993.
\newblock {\em Non-Determinism in Algebraic Specifications and Algebraic
  Programs}.
\newblock Birkh\"auser Verlag.

\bibitem[\protect\citeauthoryear{Launchbury}{Launchbury}{1993}]{Lau93}
{\sc Launchbury, J.} 1993.
\newblock A natural semantics for lazy evaluation.
\newblock In {\em Proc. ACM Symposium on Principles of Programming Languages,
  POPL 1993}. ACM Press, 144--154.

\bibitem[\protect\citeauthoryear{{L{\'o}pez-Fraguas}, Martin-Martin,
  {Rodr\'{\i}guez-Hortal{\'a}}, and
  {S{\'a}nchez-Hern{\'a}ndez}}{{L{\'o}pez-Fraguas}
  et~al\mbox{.}}{2010}]{LMRS10LetRwTPLP}
{\sc {L{\'o}pez-Fraguas}, F.}, {\sc Martin-Martin, E.}, {\sc
  {Rodr\'{\i}guez-Hortal{\'a}}, J.}, {\sc and} {\sc
  {S{\'a}nchez-Hern{\'a}ndez}, J.} 2010.
\newblock Rewriting and narrowing for constructor systems with call-time choice
  semantics.
\newblock {\em Submitted to Theory and Practice of Logic Programming, available
  under request\/}.

\bibitem[\protect\citeauthoryear{{L{\'o}pez-Fraguas},
  {Rodr\'{\i}guez-Hortal{\'a}}, and
  {S{\'a}nchez-Hern{\'a}ndez}}{{L{\'o}pez-Fraguas}
  et~al\mbox{.}}{2009}]{Lopez-FraguasRS09-RTA09}
{\sc {L{\'o}pez-Fraguas}, F.}, {\sc {Rodr\'{\i}guez-Hortal{\'a}}, J.}, {\sc
  and} {\sc {S{\'a}nchez-Hern{\'a}ndez}, J.} 2009.
\newblock A {F}ully {A}bstract {S}emantics for {C}onstructor {S}ystems.
\newblock In {\em Proceedings of the 20th International Conference on Rewriting
  Techniques and Applications, RTA 2009}. Lecture Notes in Computer Science,
  vol. 5595. Springer, 320--334.

\bibitem[\protect\citeauthoryear{L{\'o}pez-Fraguas, Rodr\'{\i}guez-Hortal{\'a},
  and S{\'a}nchez-Hern{\'a}ndez}{L{\'o}pez-Fraguas
  et~al\mbox{.}}{2009}]{DBLP:journals/corr/abs-0903-2205}
{\sc L{\'o}pez-Fraguas, F.}, {\sc Rodr\'{\i}guez-Hortal{\'a}, J.}, {\sc and}
  {\sc S{\'a}nchez-Hern{\'a}ndez, J.} 2009.
\newblock A lightweight combination of semantics for non-deterministic
  functions.
\newblock {\em CoRR\/}~{\em abs/0903.2205}.

\bibitem[\protect\citeauthoryear{L{\'opez-Fraguas}, R{odr\'{\i}guez-Hortal\'a},
  and S{\'anchez-Hern\'andez}}{L{\'opez-Fraguas} et~al\mbox{.}}{2007}]{lrs07}
{\sc L{\'opez-Fraguas}, F.}, {\sc R{odr\'{\i}guez-Hortal\'a}, J.}, {\sc and}
  {\sc S{\'anchez-Hern\'andez}, J.} 2007.
\newblock A simple rewrite notion for call-time choice semantics.
\newblock In {\em Proc. Principles and Practice of Declarative Programming}.
  ACM Press, 197--208.

\bibitem[\protect\citeauthoryear{L{\'opez-Fraguas}, R{odr\'iguez-Hortal\'a},
  and S{\'anchez-Hern\'andez}}{L{\'opez-Fraguas}
  et~al\mbox{.}}{2008}]{LRSflops08}
{\sc L{\'opez-Fraguas}, F.}, {\sc R{odr\'iguez-Hortal\'a}, J.}, {\sc and} {\sc
  S{\'anchez-Hern\'andez}, J.} 2008.
\newblock Rewriting and call-time choice: the {HO} case.
\newblock In {\em Proc. 9th International Symposium on Functional and Logic
  Programming, FLOPS 2008}. LNCS, vol. 4989. Springer, 147--162.

\bibitem[\protect\citeauthoryear{L{\'opez-Fraguas}, R{odr\'{\i}guez-Hortal\'a},
  and S{\'anchez-Hern\'andez}}{L{\'opez-Fraguas} et~al\mbox{.}}{2009}]{lsr09}
{\sc L{\'opez-Fraguas}, F.}, {\sc R{odr\'{\i}guez-Hortal\'a}, J.}, {\sc and}
  {\sc S{\'anchez-Hern\'andez}, J.} 2009.
\newblock A flexible framework for programming with non-deterministic
  functions.
\newblock In {\em Proceedings of the 2009 ACM SIGPLAN workshop on Partial
  evaluation and program manipulation, PEPM 2009}. ACM, 91--100.

\bibitem[\protect\citeauthoryear{L{\'opez-Fraguas} and
  S{\'anchez-Hern\'andez}}{L{\'opez-Fraguas} and
  S{\'anchez-Hern\'andez}}{1999}]{LS99}
{\sc L{\'opez-Fraguas}, F.} {\sc and} {\sc S{\'anchez-Hern\'andez}, J.} 1999.
\newblock $\mathcal{TOY}$: A multiparadigm declarative system.
\newblock In {\em Proc. Rewriting Techniques and Applications, RTA 1999}.
  Springer LNCS 1631, 244--247.

\bibitem[\protect\citeauthoryear{Lucas}{Lucas}{1997}]{DBLP:conf/alp/Lucas97}
{\sc Lucas, S.} 1997.
\newblock Needed reductions with context-sensitive rewriting.
\newblock In {\em ALP/HOA}, {M.~Hanus}, {J.~Heering}, {and} {K.~Meinke}, Eds.
  Lecture Notes in Computer Science, vol. 1298. Springer, 129--143.

\bibitem[\protect\citeauthoryear{Lucas}{Lucas}{1998}]{lucas98contextsensitive}
{\sc Lucas, S.} 1998.
\newblock Context-sensitive computations in functional and functional logic
  programs.
\newblock {\em Journal of Functional and Logic Programming\/}~{\em 1998,\/}~1.

\bibitem[\protect\citeauthoryear{Lux}{Lux}{2009}]{curryListLux09}
{\sc Lux, W.} 2009.
\newblock {C}urry mailing list: {T}ype-classes and call-time choice vs.
  run-time choice.
\newblock \url{http://www.informatik.uni-kiel.de/~curry/listarchive/0790.html}.

\bibitem[\protect\citeauthoryear{McCarthy}{McCarthy}{1963}]{mccarthy63basis}
{\sc McCarthy, J.} 1963.
\newblock {A Basis for a Mathematical Theory of Computation}.
\newblock In {\em Computer Programming and Formal Systems}. North-Holland,
  Amsterdam, 33--70.

\bibitem[\protect\citeauthoryear{Plasmeijer and van Eekelen}{Plasmeijer and van
  Eekelen}{1993}]{PlasmeijerE93}
{\sc Plasmeijer, R.~J.} {\sc and} {\sc van Eekelen, M. C. J.~D.} 1993.
\newblock {\em Functional Programming and Parallel Graph Rewriting}.
\newblock Addison-Wesley.

\bibitem[\protect\citeauthoryear{Plump}{Plump}{1999}]{Plump98}
{\sc Plump, D.} 1999.
\newblock {\em {T}erm {G}raph {R}ewriting}. {H}andbook of {G}raph {G}rammars
  and {C}omputing by {G}raph {T}ransformation, vol. 2: {A}pplications,
  {L}anguages, and {T}ools.
\newblock World Scientific Publishing Co., Inc., River Edge, NJ, USA, 3--61.

\bibitem[\protect\citeauthoryear{Riesco and R{odr\'{\i}guez-Hortal\'a}}{Riesco
  and R{odr\'{\i}guez-Hortal\'a}}{2010a}]{RRLDTA09}
{\sc Riesco, A.} {\sc and} {\sc R{odr\'{\i}guez-Hortal\'a}, J.} 2010a.
\newblock A natural implementation of plural semantics in {Maude}.
\newblock In {\em Proceedings of the 9th Workshop on Language Descriptions,
  Tools, and Applications, LDTA 2009}, {T.~Ekman} {and} {J.~Vinju}, Eds.
  Electronic Notes in Computer Science, vol. 253(7). Elsevier, 165--175.

\bibitem[\protect\citeauthoryear{Riesco and R{odr\'{\i}guez-Hortal\'a}}{Riesco
  and R{odr\'{\i}guez-Hortal\'a}}{2010b}]{rrpepm10}
{\sc Riesco, A.} {\sc and} {\sc R{odr\'{\i}guez-Hortal\'a}, J.} 2010b.
\newblock Programming with singular and plural non-deterministic functions.
\newblock In {\em Proceedings of the 2010 ACM SIGPLAN workshop on Partial
  evaluation and program manipulation, PEPM 2010}. ACM, 83--92.

\bibitem[\protect\citeauthoryear{Riesco and Rodr\'iguez-Hortal\'a}{Riesco and
  Rodr\'iguez-Hortal\'a}{2011}]{arjrhTPLPproofs}
{\sc Riesco, A.} {\sc and} {\sc Rodr\'iguez-Hortal\'a, J.} 2011.
\newblock Singular and plural functions for functional logic programming:
  Detailed proofs.
\newblock Tech. Rep. SIC-9/11, Dpto.\ Sistemas Inform\'aticos y Computaci\'on,
  Universidad Complutense de Madrid. November.
\newblock \url{http://gpd.sip.ucm.es/PluralSemantics}.

\bibitem[\protect\citeauthoryear{R{odr\'{\i}guez-Hortal\'a}}{R{odr\'{\i}guez-H%
ortal\'a}}{2008}]{rodH08}
{\sc R{odr\'{\i}guez-Hortal\'a}, J.} 2008.
\newblock A hierarchy of semantics for non-deterministic term rewriting
  systems.
\newblock In {\em Proceedings Foundations of Software Technology and
  Theoretical Computer Science, FSTTCS 2008}. Leibniz International Proceedings
  in Informatics. Schloss Dagstuhl--Leibniz-Zentrum fuer Informatik.

\bibitem[\protect\citeauthoryear{Rodr\'{\i}guez-Hortal{\'a}}{Rodr\'{\i}guez-Ho%
rtal{\'a}}{2009}]{curryListRod09}
{\sc Rodr\'{\i}guez-Hortal{\'a}, J.} 2009.
\newblock {C}urry mailing list: {R}e: {T}ype-classes and call-time choice vs.
  run-time choice.
\newblock \url{http://www.informatik.uni-kiel.de/~curry/listarchive/0801.html}.

\bibitem[\protect\citeauthoryear{Rodr\'{\i}guez-Hortal{\'a} and
  S{\'a}nchez-Hern{\'a}ndez}{Rodr\'{\i}guez-Hortal{\'a} and
  S{\'a}nchez-Hern{\'a}ndez}{2008}]{Frolog08}
{\sc Rodr\'{\i}guez-Hortal{\'a}, J.} {\sc and} {\sc S{\'a}nchez-Hern{\'a}ndez,
  J.} 2008.
\newblock Functions and lazy evaluation in prolog.
\newblock {\em Electr. Notes Theor. Comput. Sci.\/}~{\em 206}, 153--174.

\bibitem[\protect\citeauthoryear{Roy and Haridi}{Roy and
  Haridi}{2004}]{OzBook2004}
{\sc Roy, P.~V.} {\sc and} {\sc Haridi, S.} 2004.
\newblock {\em Concepts, Techniques, and Models of Computer Programming}.
\newblock MIT Press, Cambridge, MA, USA.

\bibitem[\protect\citeauthoryear{S{\o}ndergaard and Sestoft}{S{\o}ndergaard and
  Sestoft}{1990}]{SS90}
{\sc S{\o}ndergaard, H.} {\sc and} {\sc Sestoft, P.} 1990.
\newblock Referential transparency, definiteness and unfoldability.
\newblock {\em Acta Inf.\/}~{\em 27,\/}~6, 505--517.

\bibitem[\protect\citeauthoryear{S{\o}ndergaard and Sestoft}{S{\o}ndergaard and
  Sestoft}{1992}]{Sondergaard95}
{\sc S{\o}ndergaard, H.} {\sc and} {\sc Sestoft, P.} 1992.
\newblock Non-determinism in functional languages.
\newblock {\em The Computer Journal\/}~{\em 35,\/}~5, 514--523.

\bibitem[\protect\citeauthoryear{Sterling and Shapiro}{Sterling and
  Shapiro}{1986}]{SterlingShapiro86}
{\sc Sterling, L.} {\sc and} {\sc Shapiro, E.} 1986.
\newblock {\em The Art of {P}rolog}.
\newblock MIT Press.

\bibitem[\protect\citeauthoryear{{T}e{R}e{S}e}{{T}e{R}e{S}e}{2003}]{Terese03}
{\sc {T}e{R}e{S}e}. 2003.
\newblock {\em {T}erm {R}ewriting {S}ystems}. Cambridge Tracts in Theoretical
  Computer Science, vol. No. 55.
\newblock Cambridge University Press.

\bibitem[\protect\citeauthoryear{{The Mercury Team}}{{The Mercury
  Team}}{2012}]{Henderson96themercury}
{\sc {The Mercury Team}}. 2012.
\newblock The {Mercury} language reference manual, version 11.07.1.

\bibitem[\protect\citeauthoryear{Wadler}{Wadler}{1985}]{Wadler85}
{\sc Wadler, P.} 1985.
\newblock How to replace failure by a list of successes.
\newblock In {\em Proc. Functional Programming and Computer Architecture}.
  Springer LNCS 201.

\bibitem[\protect\citeauthoryear{Wadler and Blott}{Wadler and
  Blott}{1989}]{Wadler:ad-hoc_polymorphism}
{\sc Wadler, P.} {\sc and} {\sc Blott, S.} 1989.
\newblock How to make ad-hoc polymorphism less ad hoc.
\newblock In {\em Pro.16th ACM SIGPLAN-SIGACT Symposium on Principles of
  Programming Languages}. ACM, New York, NY, USA, 60--76.

\bibitem[\protect\citeauthoryear{Warren}{Warren}{1982}]{War82}
{\sc Warren, D.~H.} 1982.
\newblock Higher-order extensions to {Prolog}: are they needed?
\newblock In {\em Machine Intelligence 10}, {J.~Hayes}, {D.~Michie}, {and}
  {Y.-H. Pao}, Eds. Ellis Horwood Ltd., 441--454.

\end{thebibliography}

}


\end{document}